\documentclass{article}

\usepackage[colorlinks=true,
            citecolor=green!35!black,
            linkcolor=black,
            urlcolor=black]{hyperref}
\usepackage{latexsym,amsmath,amsfonts,amssymb,stmaryrd,amsthm}
\usepackage[dvipsnames]{xcolor}
\usepackage{tikz-cd}
\usepackage{url}
\usepackage{proof}
\usepackage[square]{natbib}
\usepackage[capitalize,compress,noabbrev]{cleveref}
\usepackage[backgroundcolor=cyan!50!white!50,linecolor=black]{todonotes}
\usepackage{mathpartir}
\usepackage{xparse}
\usepackage{microtype}
\usepackage{mdframed}
\usepackage{bm}
\usepackage{bbm}
\usepackage{fullpage}
\usepackage{enumitem}
\usepackage{multirow}
\usepackage{mathtools}

\mathchardef\mh="2D
\renewcommand{\_}{\ensuremath{\rule{1ex}{.4pt}}}

\renewcommand{\phi}{\varphi}

\definecolor{Revolutionary}{RGB}{232,70,68}
\newcommand{\redtt}{\texttt{\color{Revolutionary}{red}tt}}
\defcitealias{redtt}{\redtt}
\defcitealias{hott-book}{HoTT}

\makeatletter
\newtheorem*{rep@theorem}{\rep@title}
\newcommand{\newreptheorem}[2]{%
\newenvironment{rep#1}[1]{%
 \def\rep@title{#2 \ref{##1}}%
 \begin{rep@theorem}}%
 {\end{rep@theorem}}}
\makeatother

\makeatletter
\newcommand{\arcr}{\@arraycr}
\makeatother

\theoremstyle{definition}
\newtheorem{definition}{Definition}[section]
\newtheorem{recollection}{Recollection}[section]

\theoremstyle{remark}
\newtheorem{remark}[definition]{Remark}
\newtheorem{notation}[definition]{Notation}

\theoremstyle{plain}
\newtheorem{lemma}[definition]{Lemma}
\newreptheorem{lemma}{Lemma}
\newtheorem{theorem}[definition]{Theorem}
\newtheorem{ruletheorem}[definition]{Rule}
\newtheorem{proposition}[definition]{Proposition}
\newtheorem*{proposition*}{Proposition}
\newtheorem{corollary}[definition]{Corollary}
\newreptheorem{corollary}{Corollary}

\crefname{ruletheorem}{Rule}{Rules}

\newcommand{\eqdef}{\coloneqq}
\newcommand{\iffdef}{\mathrel{\,:\!\!\iff}}
\newcommand{\emp}{\varnothing}

\newcommand{\eq}{\ensuremath{\mathbin{\doteq}}}

\newcommand{\oft}[2]{#1\mathbin{:}#2}

\newcommand{\C}{\ensuremath{\mathcal{C}}}
\newcommand{\D}{\ensuremath{\mathcal{D}}}
\newcommand{\F}{\ensuremath{\mathcal{F}}}
\newcommand{\J}{\ensuremath{\mathcal{J}}}
\newcommand{\T}{\ensuremath{\mathcal{T}}}

\newcommand{\GG}{\ensuremath{\Gamma}}
\newcommand{\Ga}{\ensuremath{\alpha}}
\newcommand{\Gg}{\ensuremath{\gamma}}
\newcommand{\Ge}{\ensuremath{\varepsilon}}

\makeatletter
\def\rightharpoonupfill@{\arrowfill@\relbar\relbar\rightharpoonup}
\newcommand{\overrightharpoonup}{%
\mathpalette{\overarrow@\rightharpoonupfill@}}
\makeatother

\newcommand{\lst}[1]{\overline{#1}}

\NewDocumentCommand\Infer{o m m}{%
  \IfValueTF{#1}
    {\inferrule*[vcenter,right={#1}]{#2}{#3}}
    {\inferrule{#2}{#3}}
  }

\newcommand{\rulename}[1]{(\textsc{#1})}

\newcommand{\subst}[3]{\ensuremath{#1 [#2 / #3]}}
\newcommand{\dsubst}[3]{\ensuremath{#1 \langle{#2}/{#3}\rangle}}
\newcommand{\bsubst}[3]{\ensuremath{#1 \bm{\langle}{#2}\bm{/}{#3}\bm{\rangle}}}

\newcommand{\arr}[2]{\ensuremath{#1 \to #2}}
\NewDocumentCommand\picl{s m m m}{%
  \ensuremath{({#2}{:}{#3}) \IfBooleanF{#1}{\to} #4}}
\newcommand{\lam}[2]{\ensuremath{\lambda{#1}.{#2}}}

\newcommand{\pto}{\ensuremath{\to_*}}

\newcommand{\negate}[1]{\lnot{#1}}

\newcommand{\prd}[2]{\ensuremath{#1 \times #2}}
\newcommand{\sigmacl}[3]{\ensuremath{({#1}{:}{#2}) \times #3}}
\NewDocumentCommand\pair{s m m}{%
  \ensuremath{\IfBooleanF{#1}{\langle} #2,#3\IfBooleanF{#1}{\rangle}}}
\newcommand{\fst}[1]{\ensuremath{\mathsf{fst}(#1)}}
\newcommand{\snd}[1]{\ensuremath{\mathsf{snd}(#1)}}

\NewDocumentCommand\Eq{g g g}{%
  \ensuremath{\mathsf{Eq}\IfValueT{#1}{_{#1}\IfValueT{#2}{(#2,#3)}}}}

\NewDocumentCommand\Path{g g g}{%
  \ensuremath{\mathsf{Path}\IfValueT{#1}{_{#1}\IfValueT{#2}{(#2,#3)}}}}
\NewDocumentCommand\Id{g g g}{%
  \ensuremath{\mathsf{Id}\IfValueT{#1}{_{#1}\IfValueT{#2}{(#2,#3)}}}}
\newcommand{\dlam}[2]{\ensuremath{\lambda^{\mathbb{I}} #1. #2}}
\newcommand{\dapp}[2]{\ensuremath{#1 @ #2}}
\newcommand{\blam}[2]{\ensuremath{\lambda^{\mathbf{2}} #1. #2}}
\newcommand{\bapp}[2]{\ensuremath{#1 \bm{@} #2}}

\NewDocumentCommand\Bridge{g g g}{%
  \ensuremath{\mathsf{Bridge}\IfValueT{#1}{_{#1}\IfValueT{#2}{(#2,#3)}}}}

\newcommand{\UU}{\mathcal{U}}
\newcommand{\UKan}{\ensuremath{\UU}}
\newcommand{\UBDisc}{\ensuremath{\UU}_{\mathsf{BDisc}}}
\newcommand{\UProp}{\ensuremath{\UU}_{\mathsf{Prop}}}
\newcommand{\UPtd}{\ensuremath{\UU_*}}

\newcommand{\void}{\ensuremath{\mathsf{void}}}
\NewDocumentCommand\voidelim{G{{}} g}{%
  \ensuremath{\mathsf{void}\mh\mathsf{elim}_{#1}\IfValueT{#2}{({#2})}}}

\newcommand{\unit}{\ensuremath{\mathsf{unit}}}
\newcommand{\triv}{\ensuremath{\mathsf{\star}}}

\NewDocumentCommand\natrec{g g g}{%
  \ensuremath{\mathsf{natrec}\IfValueT{#1}{(#1;#2,#3)}}}

\newcommand{\bool}{\ensuremath{\mathsf{bool}}}
\newcommand{\true}{\ensuremath{\mathsf{true}}}
\newcommand{\false}{\ensuremath{\mathsf{false}}}
\NewDocumentCommand\ifb{G{{}} g g g}{%
  \ensuremath{\mathsf{if}_{#1}
    \IfValueT{#2}{\IfNoValueTF{#3}{({#2})}{({#2};{#3},{#4})}}}}
\NewDocumentCommand\booleta{g}{\mathsf{\bool\mh\eta}\IfValueT{#1}{(#1)}}

\newcommand{\Circle}[1][1]{\ensuremath{{\mathsf{S}^{#1}}}}

\NewDocumentCommand\lp{g}{\ensuremath{\mathsf{loop}\IfValueT{#1}{^{#1}}}}
\NewDocumentCommand\Celim{G{{}} g g g}{%
  \ensuremath{\Circle\mh\mathsf{elim}_{#1}
    \IfValueT{#2}{\IfNoValueTF{#3}{({#2})}{({#2};{#3},{#4})}}}}

\newcommand{\bint}[1][1]{\ensuremath{\mathbbm{2}}}
\newcommand{\zero}{\ensuremath{\mathsf{zero}}}
\newcommand{\one}{\ensuremath{\mathsf{one}}}
\NewDocumentCommand\seg{g}{\ensuremath{\mathsf{seg}\IfValueT{#1}{^{#1}}}}

\newcommand{\susp}[1]{\ensuremath{\mathsf{susp}(#1)}}
\newcommand{\north}{\ensuremath{\mathsf{north}}}
\newcommand{\south}{\ensuremath{\mathsf{south}}}
\NewDocumentCommand\merid{g g}{\ensuremath{\mathsf{merid}\IfValueT{#1}{^{#1}(#2)}}}
\NewDocumentCommand\suspelim{G{{}} g g g g}{%
  \ensuremath{\mathsf{susp}\mh\mathsf{elim}_{#1}%
    \IfValueT{#2}{({#2};{#3},{#4},{#5})}}}

\NewDocumentCommand\smpair{g g}{\ensuremath{\mathsf{pair}\IfValueT{#1}{(#1,#2)}}}
\newcommand{\smbasel}{\ensuremath{\mathsf{basel}}}
\newcommand{\smbaser}{\ensuremath{\mathsf{baser}}}
\NewDocumentCommand\smgluel{g g}{\ensuremath{\mathsf{gluel}\IfValueT{#1}{^{#1}\IfValueT{#2}{(#2)}}}}
\NewDocumentCommand\smgluer{g g}{\ensuremath{\mathsf{gluer}\IfValueT{#1}{^{#1}\IfValueT{#2}{(#2)}}}}
\NewDocumentCommand\smelim{G{{}} g g g g g g}{%
  \ensuremath{\mathsf{\wedge}\mh\mathsf{elim}_{#1}%
    \IfValueT{#2}{({#2};{#3},{#4},{#5},{#6},{#7})}}}
\NewDocumentCommand\bigsmelim{m m m m m m m}{%
  \ensuremath{\mathsf{\wedge}\mh\mathsf{elim}_{#1}%
    \left({#2};
      \begin{array}{l}
        #3, \\
        #4, \\
        #5, \\
        #6, \\
        #7
      \end{array}
    \right)}}

\newcommand{\dataheading}[2]{$\mathsf{data}~{#1} : {#2}~\mathsf{where}$}
\NewDocumentCommand\niceconstr{m o m g}{$\mid {#1}\IfValueT{#2}{({#2})} : #3\ \IfValueT{#4}{~[{#4}]}$}

\newcommand{\Fiber}[4]{\ensuremath{\mathsf{Fiber}(#1,#2,#3;#4)}}
\newcommand{\isContr}[1]{\ensuremath{\mathsf{isContr}(#1)}}
\newcommand{\isEquiv}[3]{\ensuremath{\mathsf{isEquiv}(#1,#2,#3)}}
\newcommand{\Retract}[2]{\ensuremath{\mathsf{Retract}(#1,#2)}}
\NewDocumentCommand\isProp{g}{\ensuremath{\mathsf{isProp}\IfValueT{#1}{(#1)}}}
\newcommand{\Equiv}[2]{\ensuremath{\mathsf{Equiv}(#1,#2)}}
\newcommand{\ideq}[1]{\ensuremath{\mathsf{ideq}(#1)}}
\newcommand{\QEquiv}[2]{\ensuremath{\mathsf{QEquiv}(#1,#2)}}
\newcommand{\fwd}[1]{\ensuremath{\mathsf{fwd}(#1)}}
\newcommand{\bwd}[1]{\ensuremath{\mathsf{bwd}(#1)}}
\newcommand{\bwdfwd}[1]{\ensuremath{\mathsf{bwd{\mh}fwd}(#1)}}
\newcommand{\fwdbwd}[1]{\ensuremath{\mathsf{fwd{\mh}bwd}(#1)}}
\newcommand{\LEM}{\mathsf{LEM}}
\newcommand{\LEMW}{\mathsf{WLEM}}
\newcommand{\trunc}[1]{{|\!|}{#1}{|\!|}}

\newcommand{\isBDisc}[1]{\ensuremath{\mathsf{isBDisc}(#1)}}
\NewDocumentCommand\loosen{g g}{\ensuremath{\mathsf{loosen}\IfValueT{#1}{_{#1}\IfValueT{#2}{(#2)}}}}
\NewDocumentCommand\loosenrefl{g g}{\ensuremath{\mathsf{loosen{\mh}refl}\IfValueT{#1}{_{#1}\IfValueT{#2}{(#2)}}}}
\newcommand{\isParametric}[3]{\ensuremath{\mathsf{isParametric}(#1,#2,#3)}}

\NewDocumentCommand\V{G{{}} g g g}{%
  \ensuremath{\mathsf{V}_{#1}\IfValueT{#2}{(#2,#3,#4)}}}
\newcommand{\Vin}[3]{\ensuremath{\mathsf{Vin}_{#1}(#2;#3)}}
\newcommand{\Vproj}[3]{\ensuremath{\mathsf{Vproj}_{#1}(#2,#3)}}

\NewDocumentCommand\extent{s g g g g g}{%
  \ensuremath{\mathsf{\IfBooleanTF{#1}{ex}{extent}}\IfValueT{#2}{_{#2}\IfValueT{#3}{(#3;#4,#5,#6)}}}}
\NewDocumentCommand\bigextent{g g g g g}{%
  \ensuremath{\mathsf{extent}\IfValueT{#1}{_{#1}{\left(#2;\begin{array}{l}#3,\\#4,\\#5\end{array}\right)}}}}
\NewDocumentCommand\extentequiv{g g g g}{%
  \ensuremath{\mathsf{extent{\mh}equiv}\IfValueT{#1}{_{#1}{(#2;#3,#4)}}}}

\NewDocumentCommand\Gel{G{{}} g g g}{%
  \ensuremath{\mathsf{Gel}_{#1}\IfValueT{#2}{(#2,#3,#4)}}}
\NewDocumentCommand\gel{G{{}} g g g}{%
  \ensuremath{\mathsf{gel}_{#1}\IfValueT{#2}{(#2,#3,#4)}}}
\NewDocumentCommand{\ungel}{g}{\ensuremath{\mathsf{ungel}\IfValueT{#1}{(#1)}}}

\NewDocumentCommand\PathGel{G{{}} g}{%
  \ensuremath{\mathsf{P}_{#1}\IfValueT{#2}{(#2)}}}
\NewDocumentCommand\Gr{s G{{}} g g g}{%
  \ensuremath{\mathsf{Gr}\IfBooleanT{#1}{^*}_{#2}\IfValueT{#3}{(#3,#4,#5)}}}
\NewDocumentCommand\gr{G{{}} g}{%
  \ensuremath{\mathsf{gr}_{#1}\IfValueT{#2}{(#2)}}}

\NewDocumentCommand{\bridgetorel}{g}{%
  \ensuremath{\mathsf{bridge{\mh}rel}\IfValueT{#1}{(#1)}}}
\NewDocumentCommand{\reltobridge}{g g g}{%
  \ensuremath{\mathsf{ra}\IfValueT{#1}{(#1,#2;#3)}}}
\NewDocumentCommand\link{g g g}{\ensuremath{\mathsf{link}\IfValueT{#1}{(#1,#2;#3)}}}
\NewDocumentCommand\unlink{g g}{\ensuremath{\mathsf{read}\IfValueT{#1}{(#1,#2)}}}

\newcommand{\etc}[1]{\ensuremath{\overrightharpoonup{#1}}}
\newcommand{\tube}[2]{\ensuremath{#1\hookrightarrow #2}}
\newcommand{\arraytube}[2]{\ensuremath{#1&\hookrightarrow& #2}}
\newcommand{\sys}[2]{\etc{\tube{#1}{#2}}}

\NewDocumentCommand\NewCoercionOperator{m m O{\rightsquigarrow} O{M}}{%
  \NewDocumentCommand#1{s G{{}} g g g}{%
    \ensuremath{\mathsf{#2}_{##2}%
    \IfBooleanTF{##1}
      {^{r #3 s}(#4)}
      {\IfValueT{##3}{^{##3 #3 ##4}(##5)}}}}
}

\NewCoercionOperator{\coe}{coe}

\NewDocumentCommand\NewCompositionOperator{s m m O{\rightsquigarrow} O{M} O{y.N_i}}{%
  \IfBooleanTF{#1}
    {\NewDocumentCommand#2{s g g g g}{%
      \IfBooleanTF{##1}
        {\ensuremath{\mathsf{#3}^{r #4 s}(#5;\sys{##2}{#6})}}
        {\IfNoValueTF{##2}
          {\ensuremath{\mathsf{#3}}}
          {\ensuremath{\mathsf{#3}^{##2 #4 ##3}(##4\IfValueT{##5}{;##5})}}}}}
    {\NewDocumentCommand#2{s G{{}} g g g g}{%
      \IfBooleanTF{##1}
        {\ensuremath{\mathsf{#3}_{##2}^{r #4 s}(#5;\sys{##3}{#6})}}
        {\IfNoValueTF{##3}
          {\ensuremath{\mathsf{#3}_{##2}}}
          {\ensuremath{\mathsf{#3}_{##2}^{##3 #4 ##4}(##5\IfValueT{##6}{;##6})}}}}}
}
\NewDocumentCommand\NewBigCompositionOperator{s m m O{\rightsquigarrow} O{M} O{y.N_i}}{%
  \IfBooleanTF{#1}
    {\NewDocumentCommand#2{s g g g g}{%
      \IfBooleanTF{##1}
        {\ensuremath{\mathsf{#3}^{r #4 r'}\left(#5;\sys{##2}{#6}\right)}}
        {\IfNoValueTF{##2}
          {\ensuremath{\mathsf{#3}}}
          {\ensuremath{\mathsf{#3}^{##2 #4 ##3}\left(##4\IfValueT{##5}{;##5}\right)}}}}}
    {\NewDocumentCommand#2{s G{{}} g g g g}{%
      \IfBooleanTF{##1}
        {\ensuremath{\mathsf{#3}_{##2}^{r #4 r'}\left(#5;\sys{##3}{#6}\right)}}
        {\IfNoValueTF{##3}
          {\ensuremath{\mathsf{#3}_{##2}}}
          {\ensuremath{\mathsf{#3}_{##2}^{##3 #4 ##4}\left(##5\IfValueT{##6}{;##6}\right)}}}}}
}
  
\NewCompositionOperator{\hcom}{hcom}
\NewBigCompositionOperator{\bighcom}{hcom}
\NewCompositionOperator{\com}{com}
\NewCompositionOperator*{\fhcom}{fhcom}
\NewCompositionOperator{\fcom}{fcom}
\NewCompositionOperator*{\Kbox}{box}[\rightsquigarrow][M][N_i]
\NewCompositionOperator*{\Kcap}{cap}[\leftsquigarrow][M][y.B_i]

\newcommand{\steps}{\ensuremath{\longmapsto}}
\newcommand{\msteps}{\ensuremath{\longmapsto^*}}
\newcommand{\evals}{\ensuremath{\Downarrow}}

\newcommand{\isval}[1]{#1\ \mathsf{val}}

\NewDocumentCommand\PTy{g}{%
  \ensuremath{\textsc{PTy}\IfValueT{#1}{(#1)}}}
\newcommand{\Coh}[1]{\ensuremath{\textsc{Coh}(#1)}}
\NewDocumentCommand\Tm{g}{%
  \ensuremath{\textsc{Tm}\IfValueT{#1}{(#1)}}}
\NewDocumentCommand\relcts{s m m}{%
  \ensuremath{{#2} \models \IfBooleanTF{#1}{#3}{(#3)}}}
\NewDocumentCommand\PathR{O{} g g g}{%
  \ensuremath{\textsc{Path}^{#1}\IfValueT{#2}{_{#2}(#3,#4)}}}
\NewDocumentCommand\BridgeR{O{} g g g}{%
  \ensuremath{\textsc{Bridge}^{#1}\IfValueT{#2}{_{#2}(#3,#4)}}}
\NewDocumentCommand\VR{O{} g g g g}{%
  \ensuremath{\textsc{Gel}^{#1}\IfValueT{#2}{_{#2}(#3,#4,#5)}}}
\NewDocumentCommand\GelR{O{} g g g g}{%
  \ensuremath{\textsc{Gel}^{#1}\IfValueT{#2}{_{#2}(#3,#4,#5)}}}

\newcommand{\lfp}[1]{\mu{#1}}

\NewDocumentCommand\cx{s o d!! g d<>}{
  \IfBooleanT{#1}{(}%
  \IfValueTF{#3}{#3}{\Phi\IfValueT{#2}{#2}}\mathop{\mid}%
  \IfValueTF{#4}{#4}{\Psi\IfValueT{#2}{#2}}%
  \IfValueT{#5}{\mathop{\mid}#5}%
  \IfBooleanT{#1}{)}}

\newcommand{\dimj}[2][\cx]{\ensuremath{{#2}\ \mathsf{pdim}\ [{#1}]}}
\newcommand{\bridgej}[2][\cx]{\ensuremath{{#2}\ \mathsf{bdim}\ [{#1}]}}
\newcommand{\bridgesj}[2][\cx]{\ensuremath{{#2}\ \mathsf{bdims}\ [{#1}]}}
\newcommand{\tmj}[2][\cx]{\ensuremath{{#2}\ \mathsf{tm}\ [{#1}]}}
\newcommand{\constraintj}[2][\cx]{\ensuremath{{#2}\ \mathsf{eq}\ [#1]}}
\newcommand{\constraintsj}[2][\cx]{\ensuremath{{#2}\ \mathsf{eqs}\ [#1]}}

\newcommand{\evhole}{[-]}
\newcommand{\evalcx}[3]{\ensuremath{{#2} : {#3} \Leftarrow {#1}}}
\newcommand{\plug}[2]{\ensuremath{{#1}[{#2}]}}

\newcommand{\apartcx}[2]{{#1}^{\setminus{#2}}}

\newcommand{\vper}[2][]{\ensuremath{\llbracket #2 \rrbracket}^{#1}}

\newcommand{\tds}[3]{\ensuremath{{#2} : {#1} \to {#3}}}
\newcommand{\psitd}[1][]{\tds{\cx*[#1']}{\psi#1}{\cx*[#1]}}

\newcommand{\td}[2]{\ensuremath{{#1}{#2}}}
\newcommand{\id}{\ensuremath{\mathrm{id}}}

\newcommand{\pre}{\mathsf{pre}}
\newcommand{\Kan}{\mathsf{Kan}}

\NewDocumentCommand\cwfctx{O{\cx} m}
  {#2\ \mathsf{ctx}\ [#1]}

\NewDocumentCommand\cwftype{m O{\cx} m}
  {#3\ \mathsf{type}_\mathsf{#1}\ [#2]}
\NewDocumentCommand\ceqtype{m O{\cx} m m}
  {#3 \eq #4\ \mathsf{type}_\mathsf{#1}\ [#2]}
\NewDocumentCommand\wftype{m O{\cx} m m}
  {#3 \gg #4 \ \mathsf{type}_\mathsf{#1}\ [#2]}
\NewDocumentCommand\eqtype{m O{\cx} m m m}
  {#3 \gg #4 \eq #5\ \mathsf{type}_\mathsf{#1}\ [#2]}

\newcommand{\cwftypep}{\cwftype{\pre}}
\newcommand{\ceqtypep}{\ceqtype{\pre}}
\newcommand{\wftypep}{\wftype{\pre}}

\newcommand{\cwftypek}{\cwftype{\Kan}}
\newcommand{\ceqtypek}{\ceqtype{\Kan}}
\newcommand{\wftypek}{\wftype{\Kan}}
\newcommand{\eqtypek}{\eqtype{\Kan}}

\newcommand{\ceqtypex}{\ceqtype{\kappa}}
\newcommand{\wftypex}{\wftype{\kappa}}
\newcommand{\eqtypex}{\eqtype{\kappa}}

\NewDocumentCommand\coftype{O{\cx} m m}
  {#2 \in #3\ [#1]}
\NewDocumentCommand\ceqtm{O{\cx} m m m}
  {#2 \eq #3 \in #4\ [#1]}
\NewDocumentCommand\oftype{O{\cx} m m m}
  {#2 \gg #3 \in #4\ [#1]}
\NewDocumentCommand\eqtm{O{\cx} m m m m}
  {#2 \gg #3 \eq #4 \in #5\ [#1]}

\NewDocumentCommand\cjudg{O{\cx} m}
  {#2\ [#1]}
\NewDocumentCommand\judg{O{\cx} m m}
  {#2 \gg #3\ [#1]}

\title{Parametric Cubical Type Theory}
\author{Evan Cavallo\\Carnegie Mellon University\\\texttt{ecavallo@cs.cmu.edu}
  \and Robert Harper\\Carnegie Mellon University\\\texttt{rwh@cs.cmu.edu}}

\date{January 2019}

\begin{document}

\maketitle

\begin{abstract}
  We exhibit a computational type theory which combines the higher-dimensional structure of cartesian cubical
  type theory with the internal parametricity primitives of parametric type theory, drawing out the
  similarities and distinctions between the two along the way. The combined theory supports both univalence
  and its relational equivalent, which we call \emph{relativity}. We demonstrate the use of the theory by
  analyzing polymorphic types, including functions between higher inductive types, and we show by example how
  relativity can be used to characterize the relational interpretation of inductive types.
\end{abstract}

\section{Introduction}

This paper brings together two closely-related varieties of ``augmented'' dependent type theory: \emph{cubical
  type theory} \citep{bch,cchm,angiuli17,angiuli18} and \emph{parametric type theory}
\citep{bernardy12,bernardy13,bernardy15,nuyts17}. Each of these theories serves to
internalize a feature found in particular models of Martin-L\"of's dependent type theory, higher-dimensional
structure in the former case and parametric polymorphism in the latter. Moreover, each does so by introducing
a notion of \emph{dimension variable}. A term which varies over a dimension variable expresses some
relationship between its \emph{endpoints}, its values at a fixed set of dimension constants. In cubical type
theory, such a term is called a \emph{path}; we follow \cite{nuyts17} in calling the parametric equivalent a
\emph{bridge}. The connection between higher-dimensional and parametric type theories is no secret; it was
observed already by \citeauthor{bernardy12} in \citeyear{bernardy12}, and the two lines of work have continued
to influence each other. We present a type theory which combines the central tools of both: univalence and
higher inductive types on the cubical side, and what we will call \emph{relativity} on the parametric side.

Parametric type theory grows out of a long line of work on \emph{parametric polymorphism}. A polymorphic
function---one whose type contains free type variables---is intuitively said to be \emph{polymorphic} when its
behavior is uniform in those variables. \cite{reynolds83} observed that a type discipline could ensure
\emph{all} polymorphic functions are parametric, and that this could be proven using a relational
interpretation of the theory. Such interpretations have since been designed for a variety of theories,
including dependent type theories \citep{bernardy10}. Recently, \citeauthor{bernardy12} and
\citeauthor{nuyts17}\ have introduced type theories which internalize their own relational interpretation,
making it possible to prove and exploit parametricity results within the theory. Beginning with
\cite{bernardy13}, these theories have relied on dimension variables to organize the iterated relational
structure which arises thence. Roughly, a term of type $A$ in $n$ dimension variables---an $n$-dimensional
\emph{bridge}---represents a term in the $n$th iterated relational interpretation of $A$.

Cubical type theory, on the other hand, endows dependent type theory with a coarse, proof-relevant notion of
equality, the aforementioned \emph{path}. A path is again a term in a dimension variable, with two endpoints
given by substituting one of two constants $0,1$. As in the relational case, the use of dimension variables
serves to organize the structure of iterated path types. In contrast, however, cubical type theory also
includes \emph{Kan operations} which ensure that all types respect paths. A central feature of cubical type
theory is \emph{univalence}, which for any two types establishes a correspondence between the paths between
them and the equivalences (roughly, isomorphisms) between them. Cubical type theory also permits the
definition of \emph{higher inductive types} (HITs), types inductively defined by higher-dimensional path
generators \citep{coquand18,cavallo19}. Cubical type theory gives a constructive interpretation of
\emph{homotopy type theory}, an extension of dependent type theory with axioms asserting univalence and the
existence of HITs in terms of the Martin-L\"of identity type \citep{hott-book}.

\paragraph{Theory}

The development of the combined type theory and its semantics is largely straightforward, as interaction
between the bridge and path structure is minimal; we only need a minor modification to cubical type theory's
Kan conditions to make bridge types Kan. As compared with the work of Bernardy et al., our parametric side
hews closer to cubical type theory: our bridges have two endpoints instead of one, and we aim for an
\emph{equivalence} between bridges in the universe and relations (the aforementioned relativity) rather than
an exact equality. The latter means that we do not need the technical device of $I$-sets employed by
\cite{bernardy15}; instead, we rely on univalence. Throughout, we present the parametric aspects in a style
meant to clarify the connection to cubical type theory, drawing attention to the essential differences as they
arise. Foremost among these is the use of \emph{structural} dimension variables on the path side (following
\citeauthor{angiuli18}) as opposed to \emph{substructural} dimensions on the bridge side (following
\citeauthor{bernardy15}), which is reflected in the differences between paths and bridges at function and
universe types.

\paragraph{Applications}

By adding cubical structure to parametric type theory, we obtain function extensionality, which is
particularly convenient for working with Church encodings. For example, we can actually show that
$\picl{X}{\UKan}{X \to X}$ is equivalent to $\unit$ (\cref{sec:examples:id}). In the opposite direction, by
adding relational structure to cubical type theory, we are able to derive ``free theorems'' \citep{wadler89}
for polymorphic functions on HITs (\cref{sec:examples:hits}). Properties like these are of particular interest
because coherence obligations can seriously complicate proofs about functions between HITs. For example,
proofs about the monoidal structure of the smash product \citep[\S6.8]{hott-book} are notoriously
difficult. We also develop the methodology of parametric type theory beyond that discussed in prior work: we
introduce the essential notion of bridge-discrete type and show how to characterize the bridge type of data
types such as $\bool$ using relativity (\cref{sec:examples:bool}). These results should transfer in some form
to the theory of \citeauthor{bernardy15}, but have not previously been explored.  Finally, the combined theory
is a witness to the consistency of homotopy type theory and cubical type theory with the negation of (some
versions of) the law of the excluded middle (\cref{sec:examples:lem}).

\paragraph{Outline}

We develop our type theory primarily in the form of a partial equivalence relation (PER) semantics for a
programming language, following the work of \cite{angiuli17} for cubical type theory. We introduce the
language in \cref{sec:language}, followed by the PER semantics (including Kan conditions) in
\cref{sec:type-systems}. In \cref{sec:cubical}, we recall the $\Path$ and $\V$ type formers of
\citeauthor{angiuli17}, which translate into the extended theory without incident, together with a few
standard definitions and theorems of higher-dimensional type theory. We come to the parametric side in
\cref{sec:bridge-types}, introducing $\Bridge$-types and showing that these are Kan. We show how
$\Bridge$-types commute with the connectives of cubical type theory in \cref{sec:compound}; for function
types, this requires the introduction of the $\extent$ operator. Finally, we introduce $\Gel$-types, the
parametric equivalent of $\V$-types, in \cref{sec:gel}.  The $\Bridge$, $\extent$, and $\Gel$ operators
correspond to the $A \ni_i a$, $\langle t,_i u\rangle$, and $(x : A) \times_i P$ constructs of
\cite{bernardy15} respectively. (We include an extended translation dictionary in \cref{sec:related}).

This completes the design of the type theory. In \cref{sec:proof-theory}, we collect the inference rules we
have established for each type into a makeshift proof theory; we use only these rules in the remainder of the
paper.  We begin by proving relativity in \cref{sec:relativity}, establishing the correspondence between
bridges in the universe and relations on their endpoints. \cref{sec:bridge-discrete} defines the sub-universe
of bridge-discrete types. In \cref{sec:examples}, we give a series of examples illustrating the use of the
theory.

\paragraph{Acknowledgments}

We thank Carlo Angiuli, Steve Awodey, Daniel Gratzer, Kuen-Bang Hou (Favonia), Dan Licata, Anders M\"ortberg,
and Jonathan Sterling for their comments and insights. Conversations with the directed type theory group at
the 2017 Mathematics Research Communities workshop and with Emily Riehl and Christian Sattler afterwards were
also instrumental to the first author's understanding of relativity.

We gratefully acknowledge the support of the Air Force Office of Scientific Research through MURI grant
FA9550-15-1-0053. Any opinions, findings and conclusions or recommendations expressed in this material are
those of the authors and do not necessarily reflect the views of the AFOSR.


\section{Programming language}
\label{sec:language}

We begin by introducing the untyped programming language on which our type systems are based. The language has
three sorts: \emph{bridge dimensions}, \emph{path dimensions}, and \emph{terms}.

\subsection{Dimension terms and contexts}

The sorts of bridge and path dimensions are defined by the following grammar.
\[
  \begin{array}{lrcllrcl}
    \text{(bridge dim)} & \bm{r} &::=& \bm{x} \mid \bm{0} \mid \bm{1} \\
    \text{(path dim)} & r &::=& x \mid 0 \mid 1
  \end{array}
\]
We use the letters $x,y,z,\ldots$ for path dimension variables and $r,s,\ldots$ for path dimension terms, with
bridge dimensions using the same letters but written in \textbf{bold type}. We use $\Ge$ to stand for $0$ or
$1$, likewise $\bm{\Ge}$ for $\bm{0}$ or $\bm{1}$. We use $\bm{\rho}$ for lists of bridge dimensions.

\begin{definition}
  A \emph{bridge-path context} is a pair $\cx*$ where $\Phi = \{\bm{x}_1,\ldots,\bm{x}_n\}$ is a set of bridge
  dimension variables and $\Psi = \{x_1,\ldots,x_n\}$ is a set of path dimension variables. We use the letter
  $\D$ for bridge-path contexts. We write $\bridgej{\bm{r}}$ and $\dimj{r}$ to mean
  $\bm{r} \in \Phi \cup \{\bm{0},\bm{1}\}$ and $r \in \Psi \cup \{0,1\}$ respectively. Likewise, we write
  $\bridgesj{\bm{\rho}}$ when $\bm{\rho} = \etc{\bm{r}_i}$ with $\bridgej{\bm{r}_i}$ for each $i$.
\end{definition}

\begin{definition}
  A \emph{bridge-path substitution} $\psitd$ is a function $\psi$ taking $\bm{x} \in \Phi$ to
  $\psi(\bm{x}) \in \Phi \cup \{\bm{0},\bm{1}\}$ and $x \in \Psi$ to $\psi(x) \in \Psi \cup \{0,1\}$, such
  that if $\psi(\bm{x}) = \psi(\bm{y})$ then either $\psi(\bm{x}) \in \{\bm{0},\bm{1}\}$ or $\bm{x} =
  \bm{y}$. We write $\td{\bm{r}}{\psi}$ and $\td{r}{\psi}$ for the action of $\psi$ on some $\bm{r}$ or $r$,
  which replaces a variable by its image under $\psi$ and leaves the constants $\bm{0},\bm{1},0,1$ untouched.
\end{definition}

The category of bridge-path contexts and their substitutions is thus the product of a category of bridge
contexts and substitutions and a category of path contexts and substitutions. The former is adapted from
\citeauthor{bernardy15}; the only change is the addition of a second constant $\bm{1}$ (corresponding to a
move from unary to binary relations). It is also the base category used in the model of homotopy type theory
due to \cite{bch,bezem17}, as well as that used by \cite{johann17} to construct parametric models of System
F. In the terminology of \cite{buchholtz17}, it is $\mathbb{C}_{(\mathrm{we},\cdot)}$: we have weakening and
exchange for bridge variables, but not contraction.  The category of path contexts is taken unchanged from
\citeauthor{angiuli17}; it is the \emph{cartesian cube category} $\mathbb{C}_{(\mathrm{wec},\cdot)}$ which
additionally supports contraction. The choice of the latter category is not essential; we conjecture that
cartesian cubical type theory could be replaced in this development with any other cubical type theory without
much change. This could be the theory of \cite{cchm}, which is based on
$\mathbb{C}_{(\mathrm{wec},\land\lor')}$, or that of \citeauthor{bezem17} (though the latter is problematic
for higher inductive types, as discussed in \cref{sec:related}). The choice of a category without contraction
for bridge dimensions is, however, apparently forced, for reasons we will first encounter in
\cref{sec:compound}.

\begin{definition}
  Given $\Phi$ and $\bm{r} \in \Phi \cup \{\bm{0},\bm{1}\}$, we write
  $\apartcx{\Phi}{\bm{r}} \eqdef \Phi \setminus \{\bm{r}\}$. Given $\D = \cx*$, we write
  $\apartcx{\D}{\bm{r}}$ for $\cx*!\apartcx{\Phi}{\bm{r}}!$. For a list $\bm{\rho} = \etc{\bm{r}_i}$, we write
  $\apartcx{-}{\bm{\rho}}$ to mean $-^{\setminus \bm{r}_1 \cdots \setminus \bm{r}_n}$.
\end{definition}

\begin{definition}
  Given $\psitd$, we write $\tds{\cx*[']!\Phi',\bm{x}!}{(\psi,\bm{x})}{\cx*}$ and
  $\tds{\cx*[']{\Psi',x}}{(\psi,x)}{\cx*}$ for the result of weakening the substitution with
  $\bm{x} \not\in \Phi'$ or $x \not\in \Psi'$ respectively. Given $\bridgej{\bm{r}}$, we write
  $\tds{\cx*!\apartcx{\Phi'}{\td{\bm{r}}{\psi}}!}{\apartcx{\psi}{\bm{r}}}{\cx*!\apartcx{\Phi}{\bm{r}}!}$ for
  the restriction of $\psi$ to $\cx*!\apartcx{\Phi}{\bm{r}}!$. Finally, given $\cx*[_0]$ disjoint from $\cx*$
  and $\cx*[']$, we write
  $\tds{\cx*!\Phi'\Phi_0!{\Psi'\Psi_0}}{\psi \times \cx*[_0]}{\cx*!\Phi\Phi_0!{\Psi\Psi_0}}$ for $\psi$
  extended by the identity on $\cx*[_0]$.
\end{definition}

\subsection{Operational semantics}

\begin{figure}
  \[
    \begin{array}{rcl}
      M &::=& \picl{a}{M}{M} \mid \lam{a}{M} \mid MM \mid { } \\[0.3em]
        && \sigmacl{a}{M}{M} \mid \pair{M}{M} \mid \fst{M} \mid \snd{M} \mid { } \\[0.3em]
        && \Path{x.M}{M}{M} \mid \dlam{x}{M} \mid \dapp{M}{r} \mid { } \\[0.3em]
        && \V{r}{M}{M}{M} \mid \Vin{r}{M}{M} \mid \Vproj{r}{M}{M} \mid { } \\[0.3em]
        && \Bridge{\bm{x}.M}{M}{M} \mid \blam{\bm{x}}{M} \mid \bapp{M}{\bm{r}}  \mid { } \\[0.3em]
        && \Gel{\bm{r}}{M}{x.M}{M} \mid \gel{\bm{r}}{M}{x.M}{M} \mid \ungel{\bm{x}.M} \mid { } \\[0.3em]
        && \extent{\bm{r}}{M}{a.M}{a.M}{a.a.a.M} \mid { } \\[0.3em]
        && \hcom{M}{r}{r}{M}{\sys{\xi}{x.M}} \mid \coe{x.M}{r}{r}{M} \mid \\
        && \cdots
    \end{array}
  \]
  \caption{The term language}
  \label{fig:language}
\end{figure}

\begin{definition}
  We write $M,N,\ldots$ for \emph{terms}, which are drawn from some fixed superset of the grammar shown in
  \cref{fig:language}. We write $\td{M}{\psi}$ for the action of a path-bridge substitution on a term; such
  substitution instances are called the \emph{aspects} of $M$. We write $\tmj{M}$ when every dimension
  variable occurring in $M$ appears either in $\Phi$ or in $\Psi$.
\end{definition}

\begin{definition}
  We write $\bsubst{M}{\bm{r}}{\bm{x}}$ for the result of substituting $\bm{r}$ for $\bm{x}$ in $M$. Likewise,
  we write $\dsubst{M}{r}{x}$ for the substitution of $r$ for $x$ in $M$. We have
  $\tds{\cx*}{\bsubst{}{\bm{r}}{\bm{x}}}{\cx*!\apartcx{\Phi}{\bm{r}},\bm{x}!}$ and
  $\tds{\cx*}{\dsubst{}{r}{x}}{\cx*{\Psi,x}}$.
\end{definition}

\begin{definition}
  An \emph{evaluation system} is a pair of judgments $M \steps M'$ and $\isval{M}$ over terms $M,M'$ which are
  \begin{enumerate}
  \item \emph{deterministic}: if $M \steps M'$ and $M \steps M''$ then $M' = M''$, and it is never the case
    that both $M \steps M'$ and $\isval{M}$,
  \item \emph{context-preserving}: if $\tmj{M}$ and $M \steps M'$ then $\tmj{M'}$.
  \end{enumerate}
  We write $M \evals V$ when $M \msteps V$ and $\isval{V}$.
\end{definition}

We fix an evaluation system for the remainder of the paper. We will require that the judgments $M \steps M'$
and $\isval{M}$ satisfy various inference rules concerning the grammar in \cref{fig:language}, which we
introduce as we discuss each operator. The results hold for any language and evaluation system which extend
those we present.


\section{Type systems and judgments}
\label{sec:type-systems}

We next introduce $\D$-PERs, which serve as the semantics of types in context $\D$, and \emph{path-bridge type
  systems}, which define a partial equivalence relation on type names and associate a $\D$-relation to each
equivalence class of types. Finally, we define a notion of \emph{Kan type}. These definitions constitute a
straightforward extension of the \citeauthor{chtt-iii}\ semantics from path dimension contexts to path-bridge
contexts.

\subsection[D-relations]{$\D$-relations}
\label{sec:type-systems:relations}

\begin{definition}
  \label{def:relation}
  Let a bridge-path context $\D$ be given. A \emph{(value) $\D$-relation} $\Ga$ is a mapping from
  substitutions $\tds{\cx*[']}{\psi'}{\D}$ to relations $\Ga_\psi(-,-)$ on terms (values)
  $\tmj[{\cx[']}]{M,M'}$. We say $\Ga$ is a \emph{$\D$-PER} when each $\Ga_\psi$ is a PER.
\end{definition}

Value $\D$-PERs will serve as the semantics of types in context $\D$. The $\D$-PER $\Ga$ assigned to a type
name $A$ gives, for each $\tds{\D'}{\psi}{\D}$, the PER $\Ga_\psi$ of values dubbed equal in the aspect
$\td{A}{\psi}$ of $A$.

\begin{notation}
  We abbreviate $\Ga_{\id}(M,M')$ as $\Ga(M,M')$. When $\tds{\D'}{\psi}{\D}$ and $\Ga$ is a $\D$-relation, we
  define a $\D'$-relation $\td{\Ga}{\psi}$ by $(\td{\Ga}{\psi})_{\psi'} \eqdef \Ga_{\psi\psi'}$.
\end{notation}

\begin{definition}
  Let $\Ga$ be a value $\D$-relation. We define a $\D$-relation $\Tm{\Ga}$ as follows: $\Tm{\Ga}_\psi(M,M')$
  holds for $\tds{\D'}{\psi}{\D}$ when for every $\tds{D_1}{\psi_1}{\D'}$ and $\tds{\D_2}{\psi_2}{\D_1}$,
  there exist terms $\tmj[\D_1]{M_1,M_1'}$ and $\tmj[\D_2]{M_2,M_2',M_{12},M_{12}'}$ such that
  \[
    \begin{array}{ccc}
      \td{M}{\psi_1} \evals M_1 & \td{M_1}{\psi_2} \evals M_2 & \td{M}{\psi_1\psi_2} \evals M_{12} \\
      \td{M'}{\psi_1} \evals M'_1 & \td{M'_1}{\psi_2} \evals M'_2 & \td{M'}{\psi_1\psi_2} \evals M'_{12}
    \end{array}
  \]
  with $\Ga_{\psi\psi_1\psi_2}(V,V')$ for all $V \in \{M_2,M_{12}\}$ and $V' \in \{M_2',M_{12}'\}$.
\end{definition}

When $\Ga$ is the value $\D$-PER assigned to a type $A$, we use $\Tm{\Ga}$ as the $\D$-PER of \emph{elements}
of $A$: the terms which evaluate to values in $A$ in a way that is coherent with dimension substitution.

\begin{remark}
  The $\D$-relation $\Tm{\Ga}$ is always stable under dimension substitution, in the sense that given
  $\tds{\D'}{\psi}{\D}$ and $\tmj[\D']{M,M'}$, $\Tm{\Ga}_\psi(M,M')$ implies
  $\Tm{\Ga}_{\psi\psi'}(\td{M}{\psi'},\td{M'}{\psi'})$ for every $\tds{\D''}{\psi'}{\D'}$. Determinism of the
  operational semantics ensures that $\Tm{\Ga}$ is a PER whenever $\Ga$ is a PER.
\end{remark}

\begin{definition}
  A value $\D$-relation $\Ga$ is \emph{value-coherent}, written $\Coh{\Ga}$, when $\Ga \subseteq \Tm{\Ga}$.
\end{definition}

To prove theorems about $\D$-relations, we use a toolbox of utility lemmas. These are minor variations on
lemmas used by \citeauthor{chtt-iii}; we include proofs in \cref{app:lemmas}.

\begin{replemma}{lem:introduction}[Introduction]
  Let $\Ga$ be a value $\D$-relation. If for every $\tds{\D'}{\psi}{\D}$, either
  $\Ga_\psi(\td{M}{\psi},\td{M'}{\psi})$ or $\Tm{\Ga}_\psi(\td{M}{\psi},\td{M'}{\psi})$, then
  $\Tm{\Ga}(M,M')$.
\end{replemma}

\begin{replemma}{lem:expansion}[Coherent expansion]
  Let $\Ga$ be a value $\D$-PER and let $\tmj[\D]{M,M'}$. If for every $\tds{\D'}{\psi}{\D}$, there exists
  $M''$ such that $\td{M}{\psi} \msteps M''$ and $\Tm{\Ga}_\psi(M'',\td{M'}{\psi})$, then $\Tm{\Ga}(M,M')$.
\end{replemma}

\begin{replemma}{lem:evaluation}[Evaluation]
  Let $\Ga$ be a value-coherent $\D$-PER and let $\tmj[\D]{M,M'}$ with $\Tm{\Ga}(M,M')$. Then $M \evals V$ and
  $M' \evals V'$ where $\Tm{\Ga}(Q,Q')$ holds for all $Q \in \{M,V\}$ and $Q' \in \{M',V'\}$.
\end{replemma}

We defer a final lemma for proving elimination theorems to \cref{sec:type-systems:type-systems}, as it will be
simpler to state using notation we have not yet introduced.

\subsection{Type systems}
\label{sec:type-systems:type-systems}

\begin{definition}
  A \emph{candidate path-bridge type system} is a four-place relation $\tau(\D,A_0,A'_0,\phi)$ on path-bridge
  contexts $\D$, values $\tmj[\D]{A_0,A_0'}$, and (ordinary) relations $\phi$ on values $\tmj[\D]{V,V'}$.
\end{definition}

\begin{definition}
  Given a candidate path-bridge type system $\tau$, we say that $\PTy{\tau}(\D,A,A',\Ga)$ holds of a
  path-bridge context $\D$, terms $\tmj[\D]{A,A'}$, and a value $\D$-relation $\Ga$ when for all
  $\tds{\D_1}{\psi_1}{\D}$ and $\tds{\D_2}{\psi_2}{\D_1}$, there exist terms $\tmj[D_1]{A_1,A'_1}$ and
  $\tmj[\D_2]{A_2,A'_2,A_{12},A'_{12}}$ such that
  \[
    \begin{array}{ccc}
      \td{A}{\psi_1} \evals A_1 & \td{A_1}{\psi_2} \evals A_2 & \td{A}{\psi_1\psi_2} \evals A_{12} \\
      \td{A'}{\psi_1} \evals A'_1 & \td{A'_1}{\psi_2} \evals A'_2 & \td{A'}{\psi_1\psi_2} \evals A'_{12}
    \end{array}
  \]
  and $\tau(\D_2,V,V',\Ga_{\psi_1\psi_2})$ for all $V \in \{A_2,A_{12}\}$ and
  $V' \in \{A'_2,A'_{12}\}$.
\end{definition}

\begin{definition}
  A candidate path-bridge type system $\tau$ is a \emph{path-bridge type system} when
  \begin{enumerate}
  \item if $\tau(\D,A_0,A_0',\phi)$ and $\tau(\D,A_0,A_0',\phi')$ then $\phi = \phi'$,
  \item if $\tau(\D,A_0,A_0',\phi)$ then $\phi$ is a PER,
  \item $\tau(\D,-,-,\phi)$ is a PER for each $\D,\phi$,
  \item if $\tau(\D,A_0,A_0',\phi)$ then $\PTy{\tau}(\D,A_0,A_0',\Ga)$ for some $\Ga$.
  \end{enumerate}
\end{definition}

\begin{definition}
  \label{def:closed-judgments}
  Given a candidate path-bridge type system $\tau$, we define the closed judgments of type theory as follows.
  \[
    \begin{array}{lcl}
      \relcts*{\tau}{\ceqtypep[\D]{A}{A'}} &\iffdef& \exists \Ga.\ \PTy{\tau}(\D,A,A',\Ga) \land \Coh{\Ga} \\
      \relcts*{\tau}{\ceqtm[\D]{M}{M'}{A}} &\iffdef& \exists \Ga.\ \PTy{\tau}(\D,A,A,\Ga) \land \Ga(M,M')
    \end{array}
  \]
  We abbreviate $\relcts*{\tau}{\ceqtypep[\D]{A}{A}}$ as $\relcts*{\tau}{\cwftypep[\D]{A}}$ and
  $\relcts*{\tau}{\ceqtm[\D]{M}{M}{A}}$ as $\relcts*{\tau}{\coftype[\D]{M}{A}}$.  If $\tau$ is a path-bridge
  type system and $\relcts*{\tau}{\cwftypep[\D]{A}}$, we write $\vper[\tau]{A}$ for the (necessarily unique)
  $\Ga$ such that $\PTy{\tau}(\D,A,A,\Ga)$ holds.  For much of this paper, we work relative to a fixed ambient
  path-bridge type system, so will drop the prefix $\tau \models$ and superscript on $\vper[\tau]{-}$.
\end{definition}

As with $\D$-relations, we have a pair of lemmas for proving that terms are related by $\PTy{\tau}$.

\begin{replemma}{lem:formation}[Formation]
  Let $\tau$ be a bridge-path type system, let $\tmj[\D]{A,A'}$, and let $\Ga$ be a value
  $\D$-relation. If for every $\tds{\D'}{\psi}{\D}$, either
  $\PTy{\tau}(\D',\td{A}{\psi},\td{A'}{\psi},\td{\Ga}{\psi})$ holds or
  $\tau(\D',\td{A}{\psi},\td{A'}{\psi},\Ga_\psi)$ holds, then $\PTy{\tau}(\D,A,A',\Ga)$.
\end{replemma}

\begin{replemma}{lem:type-expansion}[Coherent type expansion]
  Let $\tau$ be a bridge-path type system, let $\tmj[\D]{A,A'}$, and let $\Ga$ be a value $\D$-relation. If
  for all $\tds{\D'}{\psi}{\D}$, there exists $A''$ such that $\td{A}{\psi} \msteps A''$ and
  $\PTy{\tau}(\D',A'',\td{A'}{\psi},\td{\Ga}{\psi})$, then $\PTy{\tau}(\D,A,A',\Ga)$.
\end{replemma}

Finally, we state the elimination lemma referred to in \cref{sec:type-systems:relations}. We use this lemma to
prove typing rules for operators which evaluate their arguments. Certain operators of bridge-path type theory
will evaluate arguments under dimension binders or in a restricted path-bridge context, so we first introduce
a notion of expression context.

\begin{definition}
  An \emph{expression context} $\C$ is a term with at most one hole, written $\evhole$; we write
  $\plug{\C}{M}$ for the result of filling the hole with a term $M$. We write $\evalcx{\D_0}{\C}{\D}$ where
  $\D$ and $\D_0$ are disjoint when $\tmj[\D]{\plug{\C}{M}}$ for every $\tmj[\D\D_0]{M}$. For example, we have
  $\evalcx{\cx*!\bm{x}!{y}}{\blam{\bm{x}}{\dlam{y}{\evhole}}}{\cx*}$. Given $\evalcx{\D_0}{\C}{\D}$ and
  $\tds{\D'}{\psi}{\D}$ with $\D'$ disjoint from $\D_0$, we have $\evalcx{\D_0}{\td{\C}{\psi}}{\D'}$.
\end{definition}

\begin{definition}
  An expression context $\evalcx{\D_0}{\C}{\D}$ is \emph{eager} when for any $\tmj[\D\D_0]{M,M'}$ with
  $M \steps M'$, we have $\plug{\C}{M} \steps \plug{\C}{M'}$.
\end{definition}

\begin{replemma}{lem:elimination}[Elimination]
  Let $\evalcx{\D_0}{\C,\C',\T}{\D}$, $\bridgesj[\D]{\bm{\rho}}$, and let $\Ga$ be a value-coherent
  $(\apartcx{\D}{\bm{\rho}}\D_0)$-PER. Suppose that for every $\tds{\D'}{\psi}{\D}$ with $\D'$ disjoint from
  $\D_0$, we have
  \begin{enumerate}
  \item $\ceqtypep[\D']{\plug{\td{\T}{\psi}}{M}}{\plug{\td{\T}{\psi}}{M'}}$ for all
    $\Tm{\Ga}_{\apartcx{\psi}{\bm{\rho}} \times \D_0}(M,M')$,
  \item $\td{\C}{\psi}$,$\td{\C'}{\psi}$ are eager and
    $\ceqtm[\D']{\plug{\td{\C}{\psi}}{V}}{\plug{\td{\C'}{\psi}}{V'}}{\plug{\td{\T}{\psi}}{V}}$ for all
    $\Ga_{\apartcx{\psi}{\bm{\rho}} \times \D_0}(V,V')$.
  \end{enumerate}
  Then $\ceqtm[\D]{\plug{\C}{M}}{\plug{\C'}{M'}}{\plug{\T}{M}}$ for every $\Tm{\Ga}(M,M')$.
\end{replemma}

\subsection{Kan operations}

\begin{figure}
  \begin{mdframed}
    \begin{mathpar}
      \Infer
      {A \steps A'}
      {\coe{y.A}{r}{s}{M} \steps \coe{y.A'}{r}{s}{M'}}
      \and
      \Infer
      {A \steps A'}
      {\hcom{A}{r}{s}{M}{\sys{\xi_i}{y.N_i}} \steps \hcom{A'}{r}{s}{M}{\sys{\xi_i}{y.N_i}}}
      \and
      \Infer
      { }
      {\com{y.A}{r}{s}{M}{\sys{\xi_i}{y.N_i}} \steps \hcom{\dsubst{A}{s}{y}}{r}{s}{\coe{y.A}{r}{s}{M}}{\sys{\xi_i}{y.\coe{y.A}{y}{s}{N_i}}}}
    \end{mathpar}
  \end{mdframed}
  \caption{Non-type-specific operational semantics of $\hcom$, $\coe$, and $\com$}
  \label{fig:kan-opsem}
\end{figure}

We have so far explained what it means for a term $\tmj{A}$ to be a \emph{pretype}; a \emph{type} is a pretype
which additionally supports the \emph{Kan operations} $\coe$ and $\hcom$. The former, $\coe$, ensures that $A$
respects paths: for any $x \in \Psi$, it gives a function from $\dsubst{A}{r}{x}$ to $\dsubst{A}{s}{x}$ for
every $r$ and $s$. The latter, $\hcom$, is necessary to ensure that the iterated path and bridge types of any
Kan type are also Kan: it takes a term in $A$ and adjusts its lower-dimensional boundary by a given collection
of paths. Both $\coe$ and $\hcom$ are fixed operators which take the type $A$ as a parameter and evaluate it,
as shown in \cref{fig:kan-opsem}. Once $A$ is in canonical form, the further evaluation of $\coe$ and $\hcom$
is dependent on its form; for example, $\coe$ in a pair type steps to a pair of $\coe$s in the component
types. We will introduce the type-specific operational semantics of $\coe$ and $\hcom$ later on, in tandem
with the types themselves.

The definition of $\coe$ is exactly that given by \citeauthor{angiuli18}\ for cubical type theory.

\begin{definition}
  \label{def:coe-kan}
  We say that $\ceqtypep{A}{A'}$ are \emph{equally $\coe$-Kan} when for all
  $\tds{\cx*[']{\Psi',y}}{\psi}{\cx*}$, $\dimj[{\cx[']}]{r,s}$, and
  $\ceqtm[\cx[']]{M}{M'}{\dsubst{\td{A}{\psi}}{r}{y}}$, we have
  \begin{enumerate}
  \item $\ceqtm[\cx[']]{\coe{y.\td{A}{\psi}}{r}{s}{M}}{\coe{y.\td{A'}{\psi}}{r}{s}{M'}}{\dsubst{\td{A}{\psi}}{s}{y}}$,
  \item $\ceqtm[\cx[']]{\coe{y.\td{A}{\psi}}{r}{s}{M}}{M}{\dsubst{\td{A}{\psi}}{s}{y}}$ if $r = s$.
  \end{enumerate}
\end{definition}

For $\hcom$, a minor change is necessary, as we need to ensure that not only the path types but also the
bridge types of a Kan type are Kan. We do so by adding $\bm{r} = \bm{0}$ and $\bm{r} = \bm{1}$ in the
following definition.

\begin{definition}
  \label{def:constraint}
  A \emph{constraint} $\xi$ is an equation drawn from the following grammar.
  \[
    \begin{array}{rcl}
      \xi &::=& \bm{r} = \bm{0} \mid \bm{r} = \bm{1} \mid r = r'
    \end{array}
  \]
  We write $\constraintj{\xi}$ when $\xi$ is an equation on variables in $\cx*$, $\bm{x} \not\in \xi$ to mean
  that $\bm{x}$ does not occur in $\xi$, and $\models \xi$ to mean that $\xi$ is true (i.e., is a reflexive
  equation). We use $\Xi$ for lists of constraints and likewise write $\constraintsj{\Xi}$. We write
  $\apartcx{\Xi}{\bm{r}}$ for the result of removing all equations mentioning a given $\bm{r}$ from $\Xi$.
\end{definition}

\begin{definition}
  For each of the closed judgments $\cjudg{\J}$ defined in \cref{def:closed-judgments}, we define its
  \emph{restricted form} $\cjudg[\cx<\Xi>]{\J}$ to hold when $\cjudg[\cx[']]{\td{\J}{\psi}}$ holds for every
  $\psitd$ such that $\models \td{\Xi}{\psi}$.
\end{definition}

\begin{definition}
  \label{def:hcom-kan}
  We say that $\ceqtypep{A}{A'}$ are \emph{equally $\hcom$-Kan} when for all
  $\tds{\cx*[']}{\psi}{\cx*}$, $\dimj[{\cx[']}]{r,s}$, and $\constraintsj[{\cx[']}]{\etc{\xi_i}}$, if
  \begin{enumerate}
  \item $\ceqtm[\cx[']]{M}{M'}{\td{A}{\psi}}$,
  \item $\ceqtm[\cx[']{\Psi',y}<\xi_i,\xi_j>]{N_i}{N'_j}{\td{A}{\psi}}$ for all $i,j$,
  \item $\ceqtm[\cx[']<\xi_i>]{\dsubst{N_i}{r}{y}}{M}{\td{A}{\psi}}$ for all $i$,
  \end{enumerate}
  then
  \begin{enumerate}
  \item $\ceqtm[\cx[']]{\hcom{\td{A}{\psi}}{r}{s}{M}{\sys{\xi_i}{y.N_i}}}{\hcom{\td{A'}{\psi}}{r}{s}{M'}{\sys{\xi_i}{y.N'_i}}}{\td{A}{\psi}}$,
  \item $\ceqtm[\cx[']]{\hcom{\td{A}{\psi}}{r}{s}{M}{\sys{\xi_i}{y.N_i}}}{\dsubst{N_i}{s}{y}}{\td{A}{\psi}}$ for
    all $i$ with $\models \xi_i$,
  \item $\ceqtm[\cx[']]{\hcom{\td{A}{\psi}}{r}{s}{M}{\sys{\xi_i}{y.N_i}}}{M}{\td{A}{\psi}}$ if $r = s$.
  \end{enumerate}
\end{definition}

\begin{remark}
  \citeauthor{chtt-iii}\ impose an additional \emph{validity condition} on the list of constraints
  $\etc{\xi_i}$, which enables a stronger canonicity result for higher inductive types. This choice is
  orthogonal to the addition of bridges, so for simplicity's sake we will leave it out.
\end{remark}

The $\hcom$ operator takes a term $M$, a list of constraints $\etc{\xi_i}$, and a list of paths $\etc{y.N_i}$
each of which is defined on the corresponding constraint and matches $M$ on its $\dsubst{}{r}{y}$ face. The
output of $\hcom$ is a term which matches the $\dsubst{}{r'}{y}$ face of $N_i$ under $\xi_i$ for each $i$. For
example, given terms $\coftype[\cx!\Phi,\bm{x}!]{M}{A}$ and terms $\coftype[\cx{\Psi,y}]{N_0,N_1}{A}$ which
agree with $M$ on their $\dsubst{}{0}{y}$ faces, we have the following picture.
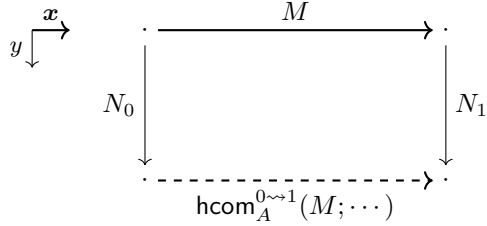
\[
  \begin{tikzpicture}
    \draw (0 , 2) [thick,->] to node [above] {\small $\bm{x}$} (0.5 , 2) ;
    \draw (0 , 2) [->] to node [left] {\small $y$} (0 , 1.5) ;
    \node (tl) at (1.5 , 2) {$\cdot$} ;
    \node (tr) at (5.5 , 2) {$\cdot$} ;
    \node (bl) at (1.5 , 0) {$\cdot$} ;
    \node (br) at (5.5 , 0) {$\cdot$} ;
    \draw (tl) [thick,->] to node [above] {$M$} (tr) ;
    \draw (tl) [->] to node [left] {$N_0$} (bl) ;
    \draw (tr) [->] to node [right] {$N_1$} (br) ;
    \draw (bl) [thick,->,dashed] to node [below] {$\hcom{A}{0}{1}{M}{\cdots}$} (br) ;
  \end{tikzpicture}
\]

A type's support for $\hcom$ implies that its paths \emph{compose}, hence the name: given paths from $M$ to
$N$ and $N$ to $P$ in $A$, $\hcom$ can be used to construct a path from $M$ to $P$ in $A$. Likewise, $\hcom$
can be used to compose bridges with paths. For example, we can combine a bridge from $M$ to $N$ and a path
from $N$ to $P$ into a bridge from $M$ to $P$. However, note that $\hcom$ does \emph{not} allow the
composition of bridges with other bridges.

\begin{definition}
  When $\ceqtypep{A}{A'}$ are both equally $\hcom$-Kan and $\coe$-Kan, we say they are \emph{equally Kan}
  and write $\ceqtypek{A}{A'}$. We will use $\kappa$ as a metavariable standing for either $\pre$ or
  $\Kan$.
\end{definition}

In any Kan type, we have a derived \emph{heterogeneous composition} operation $\com$ which combines the
functions of $\hcom$ and $\coe$. This operator, which has operational semantics defined in
\cref{fig:kan-opsem}, satisfies the following typing rules.

\begin{proposition}
  \label{prop:com}
  Let $\ceqtypek{A}{A'}$. For every $\tds{\cx*[']{\Psi',y}}{\psi}{\cx*}$, $\dimj[{\cx[']}]{r,s}$, and
  $\constraintsj[{\cx[']}]{\etc{\xi_i}}$, if
  \begin{enumerate}
  \item $\ceqtm[\cx[']]{M}{M'}{\dsubst{\td{A}{\psi}}{r}{y}}$,
  \item $\ceqtm[\cx[']{\Psi',y}<\xi_i,\xi_j>]{N_i}{N'_j}{\td{A}{\psi}}$ for all $i,j$,
  \item $\ceqtm[\cx[']<\xi_i>]{\dsubst{N_i}{r}{y}}{M}{\dsubst{\td{A}{\psi}}{r}{y}}$ for all $i$,
  \end{enumerate}
  then
  \begin{enumerate}
  \item
    $\ceqtm[\cx[']]{\com*{y.\td{A}{\psi}}{\xi_i}}{\com{y.\td{A'}{\psi}}{r}{s}{M'}{\sys{\xi_i}{y.N'_i}}}{\dsubst{\td{A}{\psi}}{s}{y}}$,
  \item $\ceqtm[\cx[']]{\com*{y.\td{A}{\psi}}{\xi_i}}{\dsubst{N_i}{s}{y}}{\dsubst{\td{A}{\psi}}{s}{y}}$ for all
    $i$ with $\models \xi_i$,
  \item $\ceqtm[\cx[']]{\com{y.\td{A}{\psi}}{r}{s}{M}{\sys{\xi_i}{y.N_i}}}{M}{\dsubst{\td{A}{\psi}}{s}{y}}$ if
    $r = s$.
  \end{enumerate}
\end{proposition}
\begin{proof}
  See \cite[Theorem 44]{chtt-iii}.
\end{proof}


\subsection{Open judgments}
\label{sec:type-systems:open}

Finally, we extend the closed judgments $\ceqtypek{A}{A'}$ and $\ceqtm{M}{M'}{A}$ to open judgments in a term
context $\GG$. In order to properly reason with the substructural context $\Phi$, we also record the
introduction of bridge dimensions in term contexts \`a la \cite{cheney09,cheney12}. A context of the form
$(\GG,\bm{x},\GG')$ indicates that the dimension $\bm{x}$ was introduced after the variables in $\GG$ but
before those in $\GG'$; thus we may substitute terms mentioning $\bm{x}$ for variables in $\GG'$, but not for
variables in $\GG$.

\begin{definition}
  We define the judgments $\cwfctx{\GG}$, $\ceqtm{\lst{M}}{\lst{M}'}{\GG}$, $\eqtypex{\GG}{B}{B'}$, and
  $\eqtm{\GG}{N}{N'}{B}$ by mutual induction as follows.

  \begin{enumerate}[label=\textbf{\Alph*.}]
  \item The context judgment $\cwfctx{\GG}$ is defined inductively by the following rules:
    \begin{mathpar}
      \Infer
      { }
      {\cwfctx{\emp}}
      \and
      \Infer
      {\cwfctx{\GG} \\ \wftypep{\GG}{A}}
      {\cwfctx{\GG,\oft{a}{A}}}
      \and
      \Infer
      {\bridgej{\bm{r}} \\
        \cwfctx[\cx!\apartcx{\Phi}{\bm{r}}!]{\GG}}
      {\cwfctx{\GG,\bm{r}}}
    \end{mathpar}
  \item The context element equality judgment $\ceqtm{\lst{M}}{\lst{M}'}{\GG}$, which presupposes
    $\cwfctx{\GG}$, is defined inductively by the following rules.
    \begin{mathpar}
      \Infer
      { }
      {\ceqtm{\emp}{\emp}{\emp}}
      \and
      \Infer
      {\ceqtm{\lst{M}}{\lst{M}'}{\GG} \\ \ceqtm{N}{N'}{\subst{A}{\lst{M}}{\Gg}}}
      {\ceqtm{(\lst{M},N)}{(\lst{M}',N')}{(\GG,\oft{a}{A})}}
      \and
      \Infer
      {\bridgej{\bm{r}} \\
        \ceqtm[\cx!\apartcx{\Phi}{\bm{r}}!]{\lst{M}}{\lst{M}'}{\GG}}
      {\ceqtm{\lst{M}}{\lst{M}'}{(\GG,\bm{r})}}
    \end{mathpar}
    In the second rule, we write $\Gg$ for the list of term variables hypothesized in $\GG$.
  \item The open type equality judgment $\eqtypex{\GG}{B}{B'}$, which presupposes $\cwfctx{\GG}$, is defined
    to hold when for every $\psitd$ and $\ceqtm[\cx[']]{\lst{M}}{\lst{M}'}{\td{\GG}{\psi}}$ we have
    $\ceqtypex[\cx[']]{\subst{\td{B}{\psi}}{\lst{M}}{\Gg}}{\subst{\td{B'}{\psi}}{\lst{M}'}{\Gg}}$ (where $\Gg$
    is the list of term variables hypothesized in $\GG$).
  \item The open element equality judgment $\eqtm{\GG}{N}{N'}{B}$, which presupposes $\cwfctx{\GG}$ and
    $\wftypex{\GG}{B}$, is defined to hold when for every $\psitd$ and
    $\ceqtm[\cx[']]{\lst{M}}{\lst{M}'}{\td{\GG}{\psi}}$ we have
    $\ceqtm[\cx[']]{\subst{\td{N}{\psi}}{\lst{M}}{\Gg}}{\subst{\td{N'}{\psi}}{\lst{M}'}{\Gg}}{\subst{\td{B}{\psi}}{\lst{M}}{\Gg}}$
    (where $\Gg$ is the list of term variables hypothesized in $\GG$).
  \end{enumerate}
\end{definition}

\begin{definition}
  For each of the judgments $\judg{\GG}{\J}$ defined above, we define its restricted form
  $\judg[\cx<\Xi>]{\GG}{\J}$ to hold when $\judg[\cx[']]{\td{\GG}{\psi}}{\td{\J}{\psi}}$ holds for every
  $\psitd$ such that $\models \td{\Xi}{\psi}$.
\end{definition}

\begin{definition}
  Given $\cwfctx{\GG}$ and $\bm{r} \in \Phi \cup \{\bm{0},\bm{1}\}$, we define the context
  $\cwfctx[\cx!\apartcx{\Phi}{\bm{r}}!]{\apartcx{\GG}{\bm{r}}}$ of term variables which cannot refer to
  $\bm{r}$ (if it is a variable) by
  \begin{align*}
    \apartcx{\GG}{\bm{r}} &\eqdef
      \left\{
        \begin{array}{ll}
          \GG_1, &\text{if $\bm{r} \in \Phi$ and $\GG = (\GG_1,\bm{r},\GG_2)$ for some $\GG_1,\GG_2$} \\
          \GG, &\text{otherwise}
        \end{array}
      \right.
  \end{align*}
\end{definition}

In the following sections, we will describe various type constructors as operators on $\D$-relations and show
that, when these are included in a type system, they satisfy appropriate introduction and elimination
rules. These rules all take the following form, where each $\J$ may be a dimension, term, or type equality
judgment.
\begin{mathpar}
  \Infer
  {\judg[\cx!\apartcx{\Phi}{\bm{\rho}_1}\Phi_1!{\Psi\Psi_1}<\apartcx{\Xi}{\bm{\rho}_1}\Xi_1>]{\apartcx{\GG}{\bm{\rho}_1}\Phi_1\GG_1}{\J_1} \\ \cdots \\ \judg[\cx!\apartcx{\Phi}{\bm{\rho}_n}\Phi_n!{\Psi\Psi_n}<\apartcx{\Xi}{\bm{\rho}_n}\Xi_n>]{\apartcx{\GG}{\bm{\rho}_n}\Phi_n\GG_n}{\J_n}}
  {\judg[\cx<\Xi>]{\GG}{\J}}
\end{mathpar}
To prove such a rule, we will first prove its restriction to closed instances holds, in the sense that the
following rule is validated.
\begin{mathpar}
  \Infer
  {\judg[\cx!\apartcx{\Phi}{\bm{\rho}_1}\Phi_1!{\Psi\Psi_1}<\Xi_1>]{\GG_1}{\J_1} \\ \cdots \\ \judg[\cx!\apartcx{\Phi}{\bm{\rho}_n}\Phi_n!{\Psi\Psi_n}<\Xi_n>]{\GG_n}{\J_n}}
  {\cjudg{\J}}
\end{mathpar}
We then observe that the validity of the latter implies the validity of the former. This follows by definition
of the interpretation of open judgments; we apply the closed rule ``pointwise'' to validate the open rule.
For suppose that we know the premises of the open rule, and want to show $\judg[\cx<\Xi>]{\GG}{\J}$. This
means we must show $\cjudg[\cx[']]{\J_\psi[\lst{M},\lst{M}']}$ for every $\psitd$ such that
$\models \td{\Xi}{\psi}$ and $\ceqtm[\cx[']]{\lst{M}}{\lst{M}'}{\td{\GG}{\psi}}$ (introducing some impromptu
notation for instantiating a binary open judgment). For each $i$, we have
$\judg[\cx!\apartcx{\Phi}{\bm{\rho}_i}\Phi_i!{\Psi\Psi_i}<\apartcx{\Xi}{\bm{\rho}_i}\Xi_i>]{\apartcx{\GG}{\bm{\rho}_i}\GG_i}{\J_i}$. We
can instantiate these hypothesis judgments with $\psi$ and the prefixes
$\ceqtm[\cx[']!\apartcx{\Phi'}{\td{\bm{\rho_i}}{\psi}}!]{\lst{M_i}}{\lst{M_i}'}{\td{\apartcx{\GG}{\bm{\rho}_i}}{\psi}}$
of $\lst{M}$ and $\lst{M'}$ corresponding to the prefix $\td{\apartcx{\GG}{\bm{\rho}_i}}{\psi}$ of
$\td{\GG}{\psi}$, which gives us
$\judg[\cx!\apartcx{\Phi'}{\td{\bm{\rho_i}}{\psi}}\Phi_i!{\Psi'\Psi_i}<\Xi_i>]{\td{\GG_i}{\psi}[\lst{M_i}]}{(\J_i)_\psi[\lst{M_i},\lst{M_i}']}$. Once
we have done this for each premise, we are in a position to apply the closed rule, and so
$\cjudg[\cx[']]{\J_\psi[\lst{M},\lst{M}']}$ follows.



\section{Imports from cubical type theory}
\label{sec:cubical}

Our parametric type theory is built on the substrate of cubical type theory. The addition of bridge dimension
variables does not disrupt the existing constructs, so we are able to import them wholesale.%
\footnote{There is one new condition that must be checked: that the existing types are closed under
  homogeneous compositions whose tubes contain equations on bridge dimensions. However, this requires only
  cosmetic changes to the existing proofs.}  In this section, we briefly recall those results from cubical
type theory (and homotopy type theory) which our development requires. Details on the cubical constructions
can be found in \cite{chtt-iii}. For homotopy type theory, the standard reference is the ``HoTT Book''
\citep{hott-book}, which we will henceforth cite as \citepalias{hott-book}. Cubical translations of many of
the results we use from homotopy type theory can be found in the standard library of the \redtt\ cubical proof
assistant \citep{redtt}.

\begin{figure}
  \begin{mathpar}
    \Infer
    {\ceqtypek[\cx{\Psi,x}]{A}{A'} \\
      \ceqtm{M_0}{M_0'}{\dsubst{A}{0}{x}} \\
      \ceqtm{M_1}{M_1'}{\dsubst{A}{1}{x}}}
    {\ceqtypek{\Path{x.A}{M_0}{M_1}}{\Path{x.A'}{M_0'}{M_1'}}}
    \and
    \Infer
    {\ceqtm[\cx{\Psi,x}]{P}{P'}{A} \\
      \ceqtm{\dsubst{P}{0}{x}}{M_0}{A} \\
      \ceqtm{\dsubst{P}{1}{x}}{M_1}{A}}
    {\ceqtm{\dlam{x}{P}}{\dlam{x}{P'}}{\Path{x.A}{M_0}{M_1}}}
    \and
    \Infer
    {\ceqtm{Q}{Q'}{\Path{x.A}{M_0}{M_1}}}
    {\ceqtm{\dapp{Q}{r}}{\dapp{Q'}{r}}{\dsubst{A}{r}{x}}}
    \and
    \Infer
    {\Ge \in \{0,1\} \\ \coftype{Q}{\Path{x.A}{M_0}{M_1}}}
    {\ceqtm{\dapp{Q}{\Ge}}{M_{\Ge}}{\dsubst{A}{\Ge}{x}}}
    \and
    \Infer
    {\cwftypek[\cx{\Psi,x}]{A} \\
      \coftype[\cx{\Psi,x}]{P}{A}}
    {\ceqtm{\dapp{(\dlam{x}{P})}{r}}{\dsubst{P}{r}{x}}{\dsubst{A}{r}{x}}}
    \and
    \Infer
    {\coftype{Q}{\Path{x.A}{M_0}{M_1}}}
    {\ceqtm{Q}{\dlam{y}{\dapp{Q}{y}}}{\Path{x.A}{M_0}{M_1}}}
  \end{mathpar}
  \caption{Rules satisfied by the $\Path$ type}
  \label{fig:path-rules}
\end{figure}

\begin{recollection}
  Cubical type theory supports dependent pair, dependent function, and universe types as in standard dependent
  type theory. Given $\cwftypek{A}$ and $\wftypek{\oft{a}{A}}{B}$, we write $\cwftypek{\sigmacl{a}{A}{B}}$ and
  $\cwftypek{\picl{a}{A}{B}}$ for their pair and function types respectively. For simplicity's sake, we will
  only make use of one universe, which we write as $\cwftypek{\UKan}$; if $\coftype{A}{\UKan}$ then
  $\cwftypek{A}$.
\end{recollection}

\begin{recollection}
  Given a ``line of types'' $\cwftypek[\cx{\Psi,x}]{A}$ and endpoint elements
  $\coftype{M_0}{\dsubst{A}{0}{x}}$ and $\coftype{M_1}{\dsubst{A}{1}{x}}$, their \emph{path type}
  $\Path{x.A}{M_0}{M_1}$ classifies values of the form $\dlam{x}{P}$ where $\coftype[\cx{\Psi,x}]{P}{A}$
  satisfies $\ceqtm{\dsubst{P}{0}{x}}{M_0}{\dsubst{A}{0}{x}}$ and
  $\ceqtm{\dsubst{P}{1}{x}}{M_1}{\dsubst{A}{1}{x}}$. This type satisfies rules shown in
  \cref{fig:path-rules}. When $A$ does not depend on $x$, we abbreviate $\Path{\_.A}{M_0}{M_1}$ as
  $\Path{A}{M_0}{M_1}$.
\end{recollection}

\begin{recollection}
  A type $\cwftypek{A}$ is \emph{contractible} when it contains an element to which all other elements are
  connected by a path.
  \[
    \isContr{A} \eqdef \sigmacl{a}{A}{\picl{a'}{A}{\Path{A}{a'}{a}}}
  \]
  Given $\coftype{F}{A \to B}$ and $\coftype{N}{B}$, the \emph{(homotopy) fiber of $F$ at $N$} is the type of
  elements $a : A$ with a path from $Fa$ to $N$.
  \[
    \Fiber{A}{B}{F}{N} \eqdef \sigmacl{a}{A}{\Path{B}{F(a)}{N}}
  \]
  A map $\coftype{F}{A \to B}$ is an \emph{equivalence} if its fibers are contractible: if each element of $B$
  is the image of an element of $A$ under $F$ in a unique way.
  \[
    \isEquiv{A}{B}{F} \eqdef \picl{b}{B}{\isContr{\Fiber{A}{B}{F}{b}}}
  \]
  Finally, we set $\Equiv{A}{B} \eqdef \sigmacl{f}{\arr{A}{B}}{\isEquiv{A}{B}{f}}$. We will abbreviate
  $\Equiv{A}{B}$ as $A \simeq B$.
\end{recollection}

\begin{recollection}[Singleton contractibility]
  \label{rec:singleton-contractibility}
  For every $\cwftypek{A}$ and $\coftype{M}{A}$, the type $\sigmacl{a}{A}{\Path{A}{a}{M}}$ is contractible.
\end{recollection}

\begin{recollection}
  We say that $\cwftypek{A}$ is a \emph{proposition} when $\isProp{A} \eqdef \picl{a,a'}{A}{\Path{A}{a}{a'}} $
  is inhabited.
\end{recollection}

\begin{recollection}
  \label{rec:isequiv-is-prop}
  For any $\cwftypek{A,B}$ and $\coftype{F}{A \to B}$, the type $\isEquiv{A}{B}{F}$ is a proposition.
\end{recollection}
\begin{proof}
  \citepalias[Lemma 4.4.4]{hott-book}.
\end{proof}

\begin{recollection}
  \label{rec:qequiv-to-equiv}
  Given $\cwftypek{A,B}$, define the type $\QEquiv{A}{B}$ of \emph{quasi-equivalences} from $A$ to $B$ as
  follows.
  \[
    \QEquiv{A}{B} \eqdef \sigmacl{f}{A \to B}{\sigmacl{g}{B \to A}{(\picl{b}{B}{\Path{B}{f(gb)}{b}}) \times (\picl{a}{A}{\Path{A}{g(fa)}{a}})}}
  \]
  For any $\cwftypek{A,B}$, there are functions $\QEquiv{A}{B} \to \Equiv{A}{B}$ and
  $\Equiv{A}{B} \to \QEquiv{A}{B}$ which preserve the underlying forward map $A \to B$.
\end{recollection}
\begin{proof}
  \citepalias[Theorems 4.4.5 and 4.2.3]{hott-book}.
\end{proof}

\begin{notation}
  Given an equivalence $E$, we write $\fwd{E}$, $\bwd{E}$, $\fwdbwd{E}$, and $\bwdfwd{E}$ for the four
  components of the quasi-equivalence it induces. We write $\ideq{A}$ for the identity equivalence on $A$.
\end{notation}

\begin{recollection}
  Suppose we have a path dimension $\dimj{r}$, a type $\cwftypek[\cx<\Xi,r=0>]{A}$ at its left endpoint, a
  type $\cwftypek{B}$, and an equivalence $\coftype[\cx<r=0>]{E}{A \simeq B}$; pictorially, a V-shape:
  \[
    \begin{tikzcd}[column sep=4em]
      A \ar{d}[sloped,above]{\simeq} \\
      B_0 \ar{r}[below,font=\normalsize]{B} & B_1 \\[-1.8em]
      {r\to}
    \end{tikzcd}
  \]
  Their \emph{$\V$-type} $\V{r}{A}{B}{E}$ is a type which, viewed as a path in direction $r$, connects $A$ to
  $B_1$.
  \begin{mathpar}
    \Infer
    {\cwftypek{A}}
    {\ceqtypek{\V{0}{A}{B}{E}}{A}}
    \and
    \Infer
    {\cwftypek{B}}
    {\ceqtypek{\V{1}{A}{B}{E}}{B}}
  \end{mathpar}
\end{recollection}

This type is used to validate the \emph{univalence axiom}, which asserts that paths in the universe $\UKan$
correspond to equivalences: for each pair of types $\coftype{A,B}{\UKan}$, there is an equivalence
$\Equiv{A}{B} \simeq \Path{\UKan}{A}{B}$. The forward map of this equivalence is given by the function
$\lam{e}{\dlam{x}{\V{x}{A}{B}{e}}}$. The reverse map takes $\coftype{P}{\Path{\UKan}{A}{B}}$ to the
equivalence given by coercion, i.e., that with forward map $\lam{a}{\coe{x.\dapp{P}{x}}{0}{1}{a}}$. We will
see that univalence has an analogue on the parametric, side which identifies bridges in the universe with
binary relations.

Separately, \emph{$\fcom$-types} are used to implement composition in the universe. In parametric cubical type
theory, these must be modified to accommodate compositions with constraints of the form $\bm{r} =
\bm{\Ge}$. However, this does not require any significant modification to the definition of $\fcom$-types or
the implementation of their Kan operations, so the details are omitted.


\section[Bridge-types]{$\Bridge$-types}
\label{sec:bridge-types}

We now introduce $\Bridge$-types, the analogue of $\Path$-types for bridge dimension variables. Their
operational semantics is shown in \cref{fig:bridge-opsem}. Below, we define their PER semantics and prove that
this definition satisfies the expected introduction, elimination, and equality rules. We collect the rules in
inference rule format as part of a proof theory in \cref{sec:proof-theory:bridge}.  While we include full
proofs of the various rules for the sake of completeness, these are essentially the same proofs as those given
for $\Path$-types by \citet[\S5.3]{chtt-iii}. The only differences come from substructurality: a term $Q$ of
type $\Bridge{\bm{x}.A}{M_{\bm{0}}}{M_{\bm{1}}}$ may only be applied to a dimension variable which is fresh
for $Q$.

\begin{figure}
  \begin{mdframed}
    \begin{mathpar}
      \Infer
      { }
      {\isval{\Bridge{\bm{x}.A}{M_0}{M_{\bm{1}}}}}
      \and
      \Infer
      { }
      {\isval{\blam{\bm{x}}{P}}}
      \and
      \Infer
      {Q \steps Q'}
      {\bapp{Q}{\bm{r}} \steps \bapp{Q'}{\bm{r}}}
      \and
      \Infer
      { }
      {\bapp{(\blam{\bm{x}}{P})}{\bm{r}} \steps \bsubst{P}{\bm{r}}{\bm{x}}}
      \and
      \Infer
      { }
      {\hcom{\Bridge{\bm{x}.A}{M_0}{M_{\bm{1}}}}{r}{s}{M}{\sys{\xi_i}{y.N_i}} \steps \\\\ \raisebox{-0.8em}{\blam{\bm{x}}{\hcom{A}{r}{s}{\bapp{M}{\bm{x}}}{\sys{\xi_i}{y.\bapp{N_i}{\bm{x}}},\tube{\bm{x}=\bm{0}}{\_.M_0},\tube{\bm{x}=\bm{1}}{\_.M_{\bm{1}}}}}}}
      \and
      \Infer
      { }
      {\coe{y.\Bridge{\bm{x}.A}{M_0}{M_{\bm{1}}}}{r}{s}{Q} \steps \blam{\bm{x}}{\com{y.A}{r}{s}{\bapp{Q}{\bm{x}}}{\tube{\bm{x}=\bm{0}}{y.M_0},\tube{\bm{x}=\bm{1}}{y.M_{\bm{1}}}}}}
    \end{mathpar}
  \end{mdframed}
\caption{Operational semantics of $\Bridge$-types}
\label{fig:bridge-opsem}
\end{figure}

\subsection{Definition}

\begin{definition}
  Fix a candidate type system $\tau$, and let $\relcts*{\tau}{\cwftypek[\cx!\Phi,\bm{x}!]{A}}$,
  $\relcts*{\tau}{\coftype{M_{\bm{0}}}{\bsubst{A}{\bm{0}}{\bm{x}}}}$, and
  $\relcts*{\tau}{\coftype{M_{\bm{1}}}{\bsubst{A}{\bm{1}}{\bm{x}}}}$ be given. We define a value $\cx*$-PER
  $\BridgeR[\tau]{\bm{x}.A}{M_{\bm{0}}}{M_{\bm{1}}}$: for each $\psitd$,
  $\BridgeR[\tau]{\bm{x}.A}{M_{\bm{0}}}{M_{\bm{1}}}_\psi(V,V')$ is defined to hold iff $V = \blam{\bm{x}}{P}$
  and $V' = \blam{\bm{x}}{P'}$ where
  \begin{enumerate}
  \item $\relcts*{\tau}{\ceqtm[\cx[']!\Phi',\bm{x}!]{P}{P'}{\td{A}{\psi}}}$,
  \item $\relcts*{\tau}{\ceqtm[\cx[']]{\bsubst{P}{\bm{0}}{\bm{x}}}{\td{M_{\bm{0}}}{\psi}}{\td{A}{\psi}}}$, and
  \item $\relcts*{\tau}{\ceqtm[\cx[']]{\bsubst{P}{\bm{1}}{\bm{x}}}{\td{M_{\bm{1}}}{\psi}}{\td{A}{\psi}}}$.
  \end{enumerate}
  We will drop the superscript $\tau$ when it is inferable.
\end{definition}

\begin{lemma}
  \label{lem:bridge-coherent}
  $\BridgeR{\bm{x}.A}{M_{\bm{0}}}{M_{\bm{1}}}$ is value-coherent.
\end{lemma}
\begin{proof}
  Let $\BridgeR{\bm{x}.A}{M_{\bm{0}}}{M_{\bm{1}}}_\psi(V,V')$ be given. By definition of $\Bridge$ and
  stability of the typing judgments under dimension substitution, this implies
  $\BridgeR{\bm{x}.A}{M_{\bm{0}}}{M_{\bm{1}}}_{\psi\psi'}(\td{V}{\psi},\td{V'}{\psi})$ for every $\psi'$. Thus
  $\Tm{\BridgeR{\bm{x}.A}{M_{\bm{0}}}{M_{\bm{1}}}}_\psi(V,V')$ by \cref{lem:introduction}.
\end{proof}

\begin{proposition}
  There exists a type system $\tau$ which, for every
  $\relcts*{\tau}{\ceqtypek[\cx!\Phi,\bm{x}!]{A}{A'}}$,
  $\relcts*{\tau}{\ceqtm{M_{\bm{0}}}{M'_{\bm{0}}}{\bsubst{A}{\bm{0}}{\bm{x}}}}$, and
  $\relcts*{\tau}{\ceqtm{M_{\bm{1}}}{M'_{\bm{1}}}{\bsubst{A}{\bm{1}}{\bm{x}}}}$, has
  \[
    \tau(\cx*,\Bridge{\bm{x}.A}{M_{\bm{0}}}{M_{\bm{1}}},\Bridge{\bm{x}.A'}{M'_{\bm{0}}}{M'_{\bm{1}}},\BridgeR{\bm{x}.A}{M_{\bm{0}}}{M_{\bm{1}}}_\id),
  \]
  in addition to supporting the standard type formers of cubical type theory. Moreover, there exists such a
  type system in which the universe $\UKan$ is also closed under $\Bridge$-types.
\end{proposition}
\begin{proof}
  See \cref{app:fixed-point}.
\end{proof}

For the remainder of this section, we assume we are working within such a type system.

\subsection{Rules}

\begin{ruletheorem}[$\Bridge$-F]
  \label{rule:bridge-F}
  Let $\ceqtypek[\cx!\Phi,\bm{x}!]{A}{A'}$, $\ceqtm{M_{\bm{0}}}{M'_{\bm{0}}}{\bsubst{A}{\bm{0}}{\bm{x}}}$, and
  $\ceqtm{M_{\bm{1}}}{M'_{\bm{1}}}{\bsubst{A}{\bm{1}}{\bm{x}}}$ be given. Then
  $\ceqtypep{\Bridge{\bm{x}.A}{M_{\bm{0}}}{M_{\bm{1}}}}{\Bridge{\bm{x}.A'}{M'_{\bm{0}}}{M'_{\bm{1}}}}$.
\end{ruletheorem}
\begin{proof}
  By \cref{lem:formation,lem:bridge-coherent}.
\end{proof}

\begin{ruletheorem}[$\Bridge$-I]
  \label{rule:bridge-I}
  Let $\ceqtm[\cx!\Phi,\bm{x}!]{P}{P'}{A}$. Then
  $\ceqtm{\blam{\bm{x}}{P}}{\blam{\bm{x}}{P'}}{\Bridge{\bm{x}.A}{M_{\bm{0}}}{M_{\bm{1}}}}$.
\end{ruletheorem}
\begin{proof}
  This is exactly \cref{lem:bridge-coherent}.
\end{proof}

\begin{ruletheorem}[$\Bridge$-$\beta$]
  \label{rule:bridge-betaF}
  Let $\cwftypek[\cx!\apartcx{\Phi}{\bm{r}},\bm{x}!]{A}$ and
  $\coftype[\cx!\apartcx{\Phi}{\bm{r}},\bm{x}!]{P}{A}$.  Then
  $\ceqtm{\dapp{(\blam{\bm{x}}{P})}{\bm{r}}}{\bsubst{P}{\bm{r}}{\bm{x}}}{\bsubst{A}{\bm{r}}{\bm{x}}}$.
\end{ruletheorem}
\begin{proof}
  By \cref{lem:expansion}, as
  $\td{(\dapp{(\blam{\bm{x}}{P})}{\bm{r}})}{\psi} \steps \td{\bsubst{P}{\bm{r}}{\bm{x}}}{\psi}$ and
  $\coftype[\cx[']]{\td{\bsubst{P}{\bm{r}}{\bm{x}}}{\psi}}{\td{\bsubst{A}{\bm{r}}{\bm{x}}}{\psi}}$ for all
  $\psi$.
\end{proof}

\begin{ruletheorem}[$\Bridge$-E]
  \label{rule:bridge-E}
  Let $\bridgej{\bm{r}}$, $\cwftypek[\cx!\apartcx{\Phi}{\bm{r}},\bm{x}!]{A}$, and
  $\ceqtm[\cx!\apartcx{\Phi}{\bm{r}}!]{Q}{Q'}{\Bridge{\bm{x}.A}{M_{\bm{0}}}{M_{\bm{1}}}}$. Then
  $\ceqtm{\bapp{Q}{\bm{r}}}{\bapp{Q'}{\bm{r}}}{\bsubst{A}{\bm{r}}{\bm{x}}}$.
\end{ruletheorem}
\begin{proof}
  By \cref{lem:elimination} applied with the expression contexts
  $\evalcx{\emp}{\bapp{\evhole}{\bm{r}},\bapp{\evhole}{\bm{r}},\bsubst{A}{\bm{r}}{\bm{x}}}{\cx*}$ and the
  $\cx*!\apartcx{\Phi}{\bm{r}}!$-PER $\BridgeR{\bm{x}.A}{M_{\bm{0}}}{M_{\bm{1}}}$. We need to show that for
  every $\psitd$ and $\Bridge{\bm{x}.A}{M_{\bm{0}}}{M_{\bm{1}}}_\psi(V,V')$, we have
  $\ceqtm[\cx[']]{\bapp{V}{\td{\bm{r}}{\psi}}}{\bapp{V'}{\td{\bm{r}}}{\psi}}{\td{\bsubst{A}{\bm{r}}{\bm{x}}}{\psi}}$. By
  definition of $\Bridge$, we have $V = \blam{\bm{x}}{P}$ and $V' = \blam{\bm{x}}{P'}$ with
  $\ceqtm[\cx[']!\apartcx{\Phi'}{\td{\bm{r}}{\psi}},\bm{x}!]{P}{P'}{\td{A}{\psi}}$. The latter implies
  $\ceqtm[\cx[']]{\bsubst{P}{\td{\bm{r}}{\psi}}{\bm{x}}}{\bsubst{P'}{\td{\bm{r}}{\psi}}{\bm{x}}}{\td{\bsubst{A}{\bm{r}}{\bm{x}}}{\psi}}$
  by stability of the typing judgments. By \cref{rule:bridge-betaF}, we also have
  $\ceqtm[\cx[']]{\bapp{V}{\td{\bm{r}}{\psi}}}{\bsubst{P}{\td{\bm{r}}{\psi}}{\bm{x}}}{\td{\bsubst{A}{\bm{r}}{\bm{x}}}{\psi}}$
  and
  $\ceqtm[\cx[']]{\bapp{V'}{\td{\bm{r}}{\psi}}}{\bsubst{P'}{\td{\bm{r}}{\psi}}{\bm{x}}}{\td{\bsubst{A}{\bm{r}}{\bm{x}}}{\psi}}$. The
  desired equation follows by transitivity.
\end{proof}

\begin{ruletheorem}[$\Bridge$-$\beta_{\bm{\Ge}}$]
  \label{rule:bridge-betapartial}
  If $\bm{\Ge} \in \{\bm{0},\bm{1}\}$ and $\coftype{Q}{\Bridge{\bm{x}.A}{M_{\bm{0}}}{M_{\bm{1}}}}$, then
  $\ceqtm{\bapp{Q}{\bm{\Ge}}}{M_{\bm{\Ge}}}{\bsubst{A}{\bm{\Ge}}{\bm{x}}}$.
\end{ruletheorem}
\begin{proof}
  By \cref{lem:evaluation}, we have $Q \evals V$ with
  $\ceqtm{Q}{V}{\Bridge{\bm{x}.A}{M_{\bm{0}}}{M_{\bm{1}}}}$. By \cref{rule:bridge-E} it follows that
  $\ceqtm{\bapp{Q}{\bm{\Ge}}}{\bapp{V}{\bm{\Ge}}}{\bsubst{A}{\bm{\Ge}}{\bm{x}}}$. We have
  $V = \blam{\bm{x}}{P}$ for some $\coftype[\cx!\Phi,\bm{x}!]{P}{A}$ with
  $\ceqtm{\bsubst{P}{\bm{\Ge}}{\bm{x}}}{M_{\bm{\Ge}}}{\bsubst{A}{\bm{\Ge}}{\bm{x}}}$, so
  $\ceqtm{\dapp{V}{\bm{\Ge}}}{\bsubst{P}{\bm{\Ge}}{\bm{x}}}{\bsubst{A}{\bm{\Ge}}{\bm{x}}}$ by
  \cref{rule:bridge-betaF}. Thus
  $ \bapp{Q}{\bm{\Ge}} \;\eq\; \bapp{V}{\bm{\Ge}} \;\eq\; \bsubst{P}{\bm{\Ge}}{\bm{x}} \;\eq\; M_{\bm{\Ge}} $
  in $\bsubst{A}{\bm{\Ge}}{\bm{x}}$.
\end{proof}

\begin{ruletheorem}[$\Bridge$-$\eta$]
  \label{rule:bridge-eta}
  If $\coftype{Q}{\Bridge{\bm{x}.A}{M_{\bm{0}}}{M_{\bm{1}}}}$, then
  $\ceqtm{Q}{\blam{\bm{y}}{\bapp{Q}{\bm{y}}}}{\Bridge{\bm{x}.A}{M_{\bm{0}}}{M_{\bm{1}}}}$.
\end{ruletheorem}
\begin{proof}
  By \cref{lem:bridge-coherent,lem:evaluation}, we have $Q \evals V$ with
  $\ceqtm{Q}{V}{\Bridge{\bm{x}.A}{M_{\bm{0}}}{M_{\bm{1}}}}$.  We have $V = \blam{\bm{x}}{P}$ for some
  $\coftype[\cx!\Phi,\bm{x}!]{P}{A}$, so
  $\ceqtm{\blam{\bm{y}}{\bapp{V}{\bm{y}}}}{\blam{\bm{y}}{\bsubst{P}{\bm{y}}{\bm{x}}}}{\Bridge{\bm{x}.A}{M_{\bm{0}}}{M_{\bm{1}}}}$
  by \cref{rule:bridge-betaF,rule:bridge-I}. The right-hand side of this equation is $\Ga$-equal to $V$, so
  $\ceqtm{\blam{\bm{y}}{\bapp{V}{\bm{y}}}}{V}{\Bridge{\bm{x}.A}{M_{\bm{0}}}{M_{\bm{1}}}}$ by transitivity. We
  now obtain the result by replacing $V$ with $Q$ everywhere, using
  \cref{rule:bridge-E,rule:bridge-betaF,rule:bridge-I} to do so on the left-hand side.
\end{proof}

\subsection{Kan conditions}

For this section, fix $\ceqtypek[\cx!\Phi,\bm{x}!]{A}{A'}$,
$\ceqtm{M_{\bm{0}}}{M'_{\bm{0}}}{\bsubst{A}{\bm{0}}{\bm{x}}}$, and
$\ceqtm{M_{\bm{1}}}{M'_{\bm{1}}}{\bsubst{A}{\bm{1}}{\bm{x}}}$.

\begin{theorem}
  \label{thm:bridge-coe-kan}
  $\ceqtypep{\Bridge{\bm{x}.A}{M_{\bm{0}}}{M_{\bm{1}}}}{\Bridge{\bm{x}.A'}{M'_{\bm{0}}}{M'_{\bm{1}}}}$ are
  equally $\coe$-Kan.
\end{theorem}
\begin{proof}
  Let $\tds{\cx*[']{\Psi',y}}{\psi}{\cx*}$, $\dimj[{\cx[']}]{r,s}$, and
  $\ceqtm[\cx[']]{Q}{Q'}{\dsubst{\td{\Bridge{\bm{x}.A}{M_{\bm{0}}}{M_{\bm{1}}}}{\psi}}{r}{y}}$ be
  given. Abbreviating $B \eqdef \Bridge{\bm{x}.A}{M_{\bm{0}}}{M_{\bm{1}}}$ and
  $B' \eqdef \Bridge{\bm{x}.A'}{M'_{\bm{0}}}{M'_{\bm{1}}}$, we need to show that
  \begin{enumerate}
  \item $\ceqtm[\cx[']]{\coe{y.\td{B}{\psi}}{r}{s}{Q}}{\coe{y.\td{B'}{\psi}}{r}{s}{Q'}}{\dsubst{\td{B}{\psi}}{s}{y}}$,
  \item $\ceqtm[\cx[']]{\coe{y.\td{B}{\psi}}{r}{s}{Q}}{Q}{\dsubst{\td{B}{\psi}}{s}{y}}$ if $r = s$.

  \end{enumerate}
  We prove these in turn.
  \begin{enumerate}
  \item We apply \cref{lem:expansion} on either side of the equation to reduce our goal to proving
    \begin{gather*}
      \blam{\bm{x}}{\com{y.\td{A}{\psi}}{r}{s}{\bapp{Q}{\bm{x}}}{\tube{\bm{x}=\bm{0}}{y.\td{M_{\bm{0}}}{\psi}},\tube{\bm{x}=\bm{1}}{y.\td{M_{\bm{1}}}{\psi}}}} \\
      \eq \\
      \blam{\bm{x}}{\com{y.\td{A'}{\psi}}{r}{s}{\bapp{Q'}{\bm{x}}}{\tube{\bm{x}=\bm{0}}{y.\td{M'_{\bm{0}}}{\psi}},\tube{\bm{x}=\bm{1}}{y.\td{M'_{\bm{1}}}{\psi}}}}
    \end{gather*}
    in $\dsubst{\td{B}{\psi}}{s}{y}$. We have
    $\ceqtm[\cx[']!\Phi,\bm{x}!]{\bapp{Q}{\bm{x}}}{\bapp{Q'}{\bm{x}}}{\dsubst{\td{A}{\psi}}{r}{y}}$ by
    \cref{rule:bridge-E} and
    $\ceqtm[\cx[']!\Phi,\bm{x}!<\bm{x}=\bm{\Ge}>]{\bapp{Q}{\bm{x}}}{\dsubst{\td{M_{\bm{\Ge}}}{\psi}}{r}{y}}{\dsubst{\td{A}{\psi}}{r}{y}}$
    for each $\bm{\Ge}$ by \cref{rule:bridge-betapartial}. The desired equality thus follows from
    \cref{prop:com} and \cref{rule:bridge-I}.

  \item Suppose $\models r = s$. Again, it suffices by \cref{lem:expansion} to show
    \begin{gather*}
      \ceqtm[\cx[']]{\blam{\bm{x}}{\com{y.\td{A}{\psi}}{r}{s}{\bapp{Q}{\bm{x}}}{\tube{\bm{x}=\bm{0}}{y.\td{M_{\bm{0}}}{\psi}},\tube{\bm{x}=\bm{1}}{y.\td{M_{\bm{1}}}{\psi}}}}}{Q}{\dsubst{\td{B}{\psi}}{s}{y}}.
    \end{gather*}
    This follows from \cref{prop:com} and \cref{rule:bridge-eta}. \qedhere
  \end{enumerate}
\end{proof}

\begin{theorem}
  \label{thm:bridge-hcom-kan}
  $\ceqtypep{\Bridge{\bm{x}.A}{M_{\bm{0}}}{M_{\bm{1}}}}{\Bridge{\bm{x}.A'}{M'_{\bm{0}}}{M'_{\bm{1}}}}$ are
  equally $\hcom$-Kan.
\end{theorem}
\begin{proof}
  Let $\psitd$, $\dimj[{\cx[']}]{r,s}$, $\constraintsj[{\cx[']}]{\etc{\xi_i}}$ be given, and suppose we have
  \begin{enumerate}
  \item $\ceqtm[\cx[']]{M}{M'}{\td{\Bridge{\bm{x}.A}{M_{\bm{0}}}{M_{\bm{1}}}}{\psi}}$,
        
  \item $\ceqtm[\cx[']{\Psi',y}<\xi_i,\xi_j>]{N_i}{N'_j}{\td{\Bridge{\bm{x}.A}{M_{\bm{0}}}{M_{\bm{1}}}}{\psi}}$ for all $i,j$,
  \item
    $\ceqtm[\cx[']<\xi_i>]{\dsubst{N_i}{r}{y}}{M}{\td{\Bridge{\bm{x}.A}{M_{\bm{0}}}{M_{\bm{1}}}}{\psi}}$
    for all $i$,
  \end{enumerate}
  Abbreviating $B \eqdef \Bridge{\bm{x}.A}{M_{\bm{0}}}{M_{\bm{1}}}$ and
  $B' \eqdef \Bridge{\bm{x}.A'}{M'_{\bm{0}}}{M'_{\bm{1}}}$, we need to show
  \begin{enumerate}
  \item
    $\ceqtm[\cx[']]{\hcom{\td{B}{\psi}}{r}{s}{M}{\sys{\xi_i}{y.N_i}}}{\hcom{\td{B'}{\psi}}{r}{s}{M'}{\sys{\xi_i}{y.N'_i}}}{\td{B}{\psi}}$,
  \item
    $\ceqtm[\cx[']]{\hcom{\td{B}{\psi}}{r}{s}{M}{\sys{\xi_i}{y.N_i}}}{\dsubst{N_i}{s}{y}}{\td{B}{\psi}}$
    for all $i$ with $\models \xi_i$,
  \item $\ceqtm[\cx[']]{\hcom{\td{B}{\psi}}{r}{s}{M}{\sys{\xi_i}{y.N_i}}}{M}{\td{B}{\psi}}$ if $r = s$.
  \end{enumerate}
  We prove these in turn.
  \begin{enumerate}
  \item We apply \cref{lem:expansion} on either side of the equation to reduce our goal to proving
    \begin{gather*}
      \blam{\bm{x}}{\hcom{\td{A}{\psi}}{r}{s}{\bapp{M}{\bm{x}}}{\sys{\xi_i}{y.\bapp{N_i}{\bm{x}}},\tube{\bm{x}=\bm{0}}{\_.\td{M_{\bm{0}}}{\psi}},\tube{\bm{x}=\bm{1}}{\_.\td{M_{\bm{1}}}{\psi}}}} \\
      \eq \\
      \blam{\bm{x}}{\hcom{\td{A'}{\psi}}{r}{s}{\bapp{M'}{\bm{x}}}{\sys{\xi_i}{y.\bapp{N'_i}{\bm{x}}},\tube{\bm{x}=\bm{0}}{\_.\td{M'_{\bm{0}}}{\psi}},\tube{\bm{x}=\bm{1}}{\_.\td{M'_{\bm{1}}}{\psi}}}}
    \end{gather*}
    in $\td{B}{\psi}$. By \cref{rule:bridge-E}, we have
    \begin{enumerate}
    \item $\ceqtm[\cx[']!\Phi',\bm{x}!]{\bapp{M}{\bm{x}}}{\bapp{M'}{\bm{x}}}{\td{A}{\psi}}$,
    \item $\ceqtm[\cx[']!\Phi',\bm{x}!{\Psi',y}<\xi_i,\xi_j>]{\bapp{N_i}{\bm{x}}}{\bapp{N_j'}{\bm{x}}}{\td{A}{\psi}}$ for all $i,j$,
    \item $\ceqtm[\cx[']!\Phi',\bm{x}!<\xi_i>]{\dsubst{(\bapp{N_i}{\bm{x}})}{r}{y}}{\bapp{M}{\bm{x}}}{\td{A}{\psi}}$ for all $i$.
    \end{enumerate}
    By \cref{rule:bridge-betapartial}, we also have
    \begin{enumerate}
    \item $\ceqtm[\cx[']!\Phi',\bm{x}!{\Psi',y}<\xi_i,\bm{x}=\bm{\Ge}>]{\bapp{N_i}{\bm{x}}}{\td{M_{\bm{\Ge}}}{\psi}}{\td{A}{\psi}}$ for all $i$ and $\bm{\Ge} \in \bm{0},\bm{1}$,
    \item $\ceqtm[\cx[']!\Phi',\bm{x}!<\bm{x}=\bm{\Ge}>]{\td{M_{\bm{\Ge}}}{\psi}}{\bapp{M}{\bm{x}}}{\td{A}{\psi}}$ for
      $\bm{\Ge} \in \bm{0},\bm{1}$.
    \end{enumerate}
    From the equations, it follows by the $\hcom$-Kan condition on $\ceqtypek{A}{A'}$ that
    \begin{gather*}
      \hcom{\td{A}{\psi}}{r}{s}{\bapp{M}{\bm{x}}}{\sys{\xi_i}{y.\bapp{N_i}{\bm{x}}},\tube{\bm{x}=\bm{0}}{\_.\td{M_{\bm{0}}}{\psi}},\tube{\bm{x}=\bm{1}}{\_.\td{M_{\bm{1}}}{\psi}}} \\
      \eq \\
      \hcom{\td{A'}{\psi}}{r}{s}{\bapp{M'}{\bm{x}}}{\sys{\xi_i}{y.\bapp{N'_i}{\bm{x}}},\tube{\bm{x}=\bm{0}}{\_.\td{M'_{\bm{0}}}{\psi}},\tube{\bm{x}=\bm{1}}{\_.\td{M'_{\bm{1}}}{\psi}}}
    \end{gather*}
    in $\td{A}{\psi}$ at $\cx*[']!\Phi',\bm{x}!$. It also follows that, for each $\bm{\Ge} \in \{\bm{0},\bm{1}\}$,
    we have
    \[
      \ceqtm{\bsubst{\hcom{\td{A}{\psi}}{r}{s}{\bapp{M}{\bm{x}}}{\sys{\xi_i}{y.\bapp{N_i}{\bm{x}}},\tube{\bm{x}=\bm{0}}{\_.\td{M_{\bm{0}}}{\psi}},\tube{\bm{x}=\bm{1}}{\_.\td{M_{\bm{1}}}{\psi}}}}{\bm{\Ge}}{\bm{x}}}{\td{M_{\bm{\Ge}}}{\psi}}{\bsubst{\td{A}{\psi}}{\bm{\Ge}}{\bm{x}}}.
    \]
    Thus we may apply \cref{rule:bridge-I} to obtain the desired equation.
  \item Suppose $\models \xi_i$. By again applying \cref{lem:expansion}, it suffices to show
    \begin{gather*}
      \ceqtm[\cx[']]{\blam{\bm{x}}{\hcom{\td{A}{\psi}}{r}{s}{\bapp{M}{\bm{x}}}{\sys{\xi_i}{y.\bapp{N_i}{\bm{x}}},\tube{\bm{x}=\bm{0}}{\_.\td{M_{\bm{0}}}{\psi}},\tube{\bm{x}=\bm{1}}{\_.\td{M_{\bm{1}}}{\psi}}}}}{\dsubst{N_i}{s}{y}}{\td{B}{\psi}}
    \end{gather*}
    This follows from the $\hcom$-Kan condition for $\ceqtypek{A}{A'}$ and
    \cref{rule:bridge-I,rule:bridge-eta}.
  \item Analogous to 2. \qedhere
  \end{enumerate}
\end{proof}


\section{Bridges in compound types}
\label{sec:compound}

We intend to think of a type $\cwftypek[\cx!\Phi,\bm{x}!]{A}$ varying in a dimension variable $\bm{x}$ as a
type-valued binary relation on its endpoints $\bsubst{A}{\bm{0}}{\bm{x}}$ and
$\bsubst{A}{\bm{1}}{\bm{x}}$. This point of view will be validated when we prove relativity in
\cref{sec:relativity}, but we can already give one direction of the correspondence: the relation corresponding
to $A$ is the family of $\Bridge$-types $\Bridge{\bm{x}.A}{-}{-}$. We therefore expect that for compound
types, such as pair, path, and function types, we can show that their bridge types align with their standard
``logical'' relational interpretations. For example, a bridge in a pair type should correspond uniquely to a
pair of bridges in its component types. For pairs and paths, this is indeed straightforward. For function
types, on the other hand, we will need to introduce a new operator we call $\extent$, previously introduced by
\citeauthor{bernardy15}\ under the name $\langle -,_i-\rangle$. This is the first place where the role of
substructurality becomes evident; the second will be in \cref{sec:gel}.

\begin{theorem}
  \label{thm:sigma-bridge}
  Let $\cwftypek[\cx!\Phi,\bm{x}!]{A}$, $\wftypek[\cx!\Phi,\bm{x}!]{\oft{a}{A}}{B}$,
  $\coftype{M_0}{\dsubst{(\sigmacl{a}{A}{B})}{\bm{0}}{\bm{x}}}$, and
  $\coftype{M_1}{\dsubst{(\sigmacl{a}{A}{B})}{\bm{1}}{\bm{x}}}$. Then we have the following equivalence.
  \[
    \Bridge{\bm{x}.\sigmacl{a}{A}{B}}{M_0}{M_1} \simeq \sigmacl{p}{\Bridge{\bm{x}.A}{\fst{M_0}}{\fst{M_1}}}{\Bridge{\bm{x}.\subst{B}{\bapp{p}{\bm{x}}}{a}}{\snd{M_0}}{\snd{M_1}}}
  \]
\end{theorem}
\begin{proof}
  For the forward map, we send $q$ to
  $\pair{\blam{\bm{x}}{\fst{\bapp{q}{\bm{x}}}}}{\blam{\bm{x}}{\snd{\bapp{q}{\bm{x}}}}}$. For the inverse, we
  send $t$ to $\blam{\bm{x}}{\pair{\bapp{\fst{t}}{\bm{x}}}{\bapp{\snd{t}}{\bm{x}}}}$. It is simple to
  establish via \cref{rec:qequiv-to-equiv} that these maps give rise to an equivalence.
\end{proof}

\begin{theorem}
  \label{thm:path-bridge}
  Let a type $\cwftypek[\cx!\Phi,\bm{x}!{\Psi,y}]{A}$, $\coftype[\cx!\Phi,\bm{x}!]{M_0}{\dsubst{A}{0}{y}}$,
  $\coftype[\cx!\Phi,\bm{x}!]{M_1}{\dsubst{A}{1}{y}}$,
  $\coftype{P_0}{\dsubst{\Path{y.A}{M_0}{M_1}}{\bm{0}}{\bm{x}}}$, and
  $\coftype{P_1}{\dsubst{\Path{y.A}{M_0}{M_1}}{\bm{1}}{\bm{x}}}$ be given. Then we have the following
  equivalence.
  \[
    \Bridge{\bm{x}.\Path{y.A}{M_0}{M_1}}{P_0}{P_1} \simeq \Path{y.\Bridge{\bm{x}.A}{\dapp{P_0}{y}}{\dapp{P_1}{y}}}{\blam{\bm{x}}{M_0}}{\blam{\bm{x}}{M_1}}
  \]
\end{theorem}
\begin{proof}
  For the forward map, we send $p$ to $\dlam{y}{\blam{\bm{x}}{\dapp{\bapp{p}{\bm{x}}}{y}}}$. For the inverse,
  we send $q$ to $\blam{\bm{x}}{\dlam{y}{\bapp{\dapp{q}{y}}{\bm{x}}}}$. It is again simple to see with
  \cref{rec:qequiv-to-equiv} that these give rise to an equivalence.
\end{proof}

\begin{figure}
  \begin{mdframed}
    \begin{mathpar}
      \Infer
      { }
      {\extent{\bm{0}}{M}{a.N}{a'.P}{a.a'.c.Q} \steps \subst{N}{M}{a}}
      \and
      \Infer
      { }
      {\extent{\bm{1}}{M}{a.N}{a'.P}{a.a'.c.Q} \steps \subst{P}{M}{a'}}
      \and
      \Infer
      { }
      {\extent{\bm{x}}{M}{a.N}{a'.P}{a.a'.c.Q} \steps \bapp{\subst{\subst{\subst{Q}{\bsubst{M}{\bm{0}}{\bm{x}}}{a}}{\bsubst{M}{\bm{1}}{\bm{x}}}{a'}}{\blam{\bm{x}}{M}}{c}}{\bm{x}}}
    \end{mathpar}
  \end{mdframed}
\caption{Operational semantics of the $\extent$ operator}
\label{fig:extent-opsem}
\end{figure}

We now come to function types. Our expectation is that a term of type $\Bridge{\bm{x}.A \to B}{F}{F'}$
corresponds to a term of type
$\picl*{a}{\bsubst{A}{\bm{0}}{\bm{x}}}{\picl*{a'}{\bsubst{A}{\bm{1}}{\bm{x}}}{\picl{c}{\Bridge{\bm{x}.A}{a}{a'}}{\Bridge{\bm{x}.\subst{B}{\bapp{c}{\bm{x}}}{a}}{Fa}{F'a'}}}}$,
a function taking each bridge over $A$ to a bridge over $B$ between the images of its endpoints under $F$ and
$F'$. Indeed, it is easy to give the forward direction of this anticipated equivalence: we send $q$ in the
former type to $\lam{a}{\lam{a'}{\lam{c}{\blam{\bm{x}}{(\bapp{q}{\bm{x}})(\bapp{c}{\bm{x}})}}}}$ in the
latter.

It is in the reverse direction that we run into trouble. Suppose we have $g$ in the latter type. Given any
$\bm{x}$ and $a : A$, we need to be able to construct an element of $B$, but $g$ expects a \emph{bridge over}
$A$, not an \emph{element of $A$ varying in $\bm{x}$}. Intuitively, we would like to write
``$g(\bsubst{a}{\bm{0}}{\bm{x}})(\bsubst{a}{\bm{1}}{\bm{x}})(\blam{\bm{x}}{a})$,'' capturing the occurrences
of $\bm{x}$ in $a$. The ability to internally abstract a term over a variable in this way is a characteristic
feature of \emph{nominal sets} \citep{pitts13}. These are equivalent to presheaves on the category of names
and injective substitutions, the subcategory of our category of bridge contexts excluding morphisms which send
variables to $\bm{0}$ or $\bm{1}$. Injectivity, which amounts to the exclusion of diagonal substitutions
$\bsubst{}{\bm{y}}{\bm{x}}$, is essential, as the map $(\bm{x},M) \mapsto \blam{\bm{x}}{M}$ does not commute
with such substitutions. For if $M$ mentions some $\bm{y}$, abstracting $\bm{y}$ after applying
$\bsubst{}{\bm{y}}{\bm{x}}$ will cause the occurrences of $\bm{y}$ in $M$ to be captured; if we abstract
$\bm{x}$ \emph{before} applying $\bsubst{}{\bm{y}}{\bm{x}}$, then these occurrences will not be captured.

We also need to consider the case of substitutions $\bsubst{}{\bm{0}}{\bm{x}}$ and
$\bsubst{}{\bm{1}}{\bm{x}}$, so in the end we will provide a kind of case operator for dimension terms. This
operator takes the form $\extent{\bm{r}}{M}{a.N}{a'.P}{a.a'.c.Q}$, so named because it reveals the full
``extent'' of the term $M$ in the $\bm{r}$ direction. If $\bm{r}$ is $\bm{0}$, then $M$ is supplied to $N$; if
$\bm{r}$ is $\bm{1}$, then $M$ is supplied to $P$. If $\bm{r}$ is some $\bm{x}$, then
$\bsubst{M}{\bm{0}}{\bm{x}}$, $\bsubst{M}{\bm{1}}{\bm{x}}$, and the abstracted $\blam{\bm{x}}{M}$ are supplied
to $Q$, which should be a bridge between $\subst{N}{\bsubst{M}{\bm{0}}{\bm{x}}}{a}$ and
$\subst{P}{\bsubst{M}{\bm{1}}{\bm{x}}}{a'}$. The operational semantics of $\extent$, which do exactly this,
are shown in \cref{fig:extent-opsem}. Below, we prove well-typedness and computation rules for $\extent$,
which are collected as part of the proof theory in \cref{sec:proof-theory:extent}.

\begin{ruletheorem}[$\extent*$-$\beta_{\bm{0}}$]
  \label{rule:extent-beta0}
  If $\cwftypek{A}$, $\wftypek{\oft{d}{A}}{B}$, and $\oftype{\oft{a}{A}}{N}{\subst{B}{a}{d}}$, then
  $\ceqtm{\extent{\bm{0}}{M}{a.N}{a'.P}{a.a'.c.Q}}{\subst{N}{M}{a}}{\subst{B}{M}{d}}$.
\end{ruletheorem}
\begin{proof}
  By \cref{lem:expansion}, as
  $\td{\extent{\bm{0}}{M}{a.N}{a'.P}{a.a'.c.Q}}{\psi} \steps \td{\subst{N}{M}{a}}{\psi}$ for all
  $\psi$.
\end{proof}

\begin{ruletheorem}[$\extent*$-$\beta_{\bm{1}}$]
  \label{rule:extent-beta1}
  If $\cwftypek{A}$, $\wftypek{\oft{d}{A}}{B}$, and $\oftype{\oft{a'}{A}}{P}{\subst{B}{a'}{d}}$, then
  $\ceqtm{\extent{\bm{1}}{M}{a.N}{a'.P}{a.a'.c.Q}}{\subst{P}{M}{a'}}{\subst{B}{M}{d}}$.
\end{ruletheorem}
\begin{proof}
  By \cref{lem:expansion}, as
  $\td{\extent{\bm{1}}{M}{a.N}{a'.P}{a.a'.c.Q}}{\psi} \steps \td{\subst{P}{M}{a'}}{\psi}$ for all
  $\psi$.
\end{proof}

\begin{ruletheorem}[$\extent*$-$\beta$]
  \label{rule:extent-betaF}
  If $\bridgej{\bm{r}}$ and
  \begin{enumerate}
  \item $\cwftypek[\cx!\apartcx{\Phi}{\bm{r}},\bm{x}!]{A}$,
  \item $\wftypek[\cx!\apartcx{\Phi}{\bm{r}},\bm{x}!]{\oft{d}{A}}{B}$,
  \item $\coftype[\cx!\apartcx{\Phi}{\bm{r}},\bm{x}!]{M}{A}$,
  \item $\oftype[\cx!\apartcx{\Phi}{\bm{r}}!]{\oft{a}{\bsubst{A}{\bm{0}}{\bm{x}}}}{N}{\subst{\bsubst{B}{\bm{0}}{\bm{x}}}{a}{d}}$,
  \item $\oftype[\cx!\apartcx{\Phi}{\bm{r}}!]{\oft{a'}{\bsubst{A}{\bm{1}}{\bm{x}}}}{P}{\subst{\bsubst{B}{\bm{1}}{\bm{x}}}{a'}{d}}$,
  \item
    $\oftype[\cx!\apartcx{\Phi}{\bm{r}}!]{\oft{a}{\bsubst{A}{\bm{0}}{\bm{x}}},\oft{a'}{\bsubst{A}{\bm{1}}{\bm{x}}},\oft{c}{\Bridge{\bm{x}.A}{a}{a'}}}{Q}{\Bridge{\bm{x}.\subst{B}{\bapp{c}{\bm{x}}}{d}}{N}{P}}$,
  \end{enumerate}
  then
  $\ceqtm{\extent{\bm{r}}{\bsubst{M}{\bm{r}}{\bm{x}}}{a.N}{a'.P}{a.a'.c.Q}}{\bapp{\subst{\subst{\subst{Q}{\bsubst{M}{\bm{0}}{\bm{x}}}{a}}{\bsubst{M}{\bm{1}}{\bm{x}}}{a'}}{\blam{\bm{x}}{M}}{c}}{\bm{r}}}{\bsubst{\subst{B}{M}{d}}{\bm{r}}{\bm{x}}}$.
\end{ruletheorem}
\begin{proof}
  Via \cref{lem:expansion}. Let $\psitd$ be given. We have three cases:
  \begin{itemize}
  \item $\models \td{\bm{r}}{\psi} = \bm{0}$.

    Then
    $\ceqtm[\cx[']]{\td{\extent{\bm{r}}{\bsubst{M}{\bm{r}}{\bm{x}}}{a.N}{a'.P}{a.a'.c.Q}}{\psi}}{\td{\subst{N}{\bsubst{M}{\bm{r}}{\bm{x}}}{a}}{\psi}}{\td{\bsubst{\subst{B}{M}{d}}{\bm{r}}{\bm{x}}}{\psi}}$
    by \cref{rule:extent-beta0}, and the right-hand side is equal to
    $\td{(\bapp{\subst{\subst{\subst{Q}{\bsubst{M}{\bm{0}}{\bm{x}}}{a}}{\bsubst{M}{\bm{1}}{\bm{x}}}{a'}}{\blam{\bm{x}}{M}}{c}}{\bm{r}})}{\psi}$
    by \cref{rule:bridge-betapartial}.
  \item $\models \td{\bm{r}}{\psi} = \bm{1}$.

    Analogous to the previous case.
    
  \item $\not\models \td{\bm{r}}{\psi} = \bm{\Ge}$ for all $\bm{\Ge} \in \{\bm{0},\bm{1}\}$.
    Then
    \[
      \td{\extent{\bm{r}}{\bsubst{M}{\bm{r}}{\bm{x}}}{a.N}{a'.P}{a.a'.c.Q}}{\psi}
      \steps \td{(\bapp{\subst{\subst{\subst{Q}{\bsubst{M}{\bm{0}}{\bm{x}}}{a}}{\bsubst{M}{\bm{1}}{\bm{x}}}{a'}}{\blam{\bm{x}}{M}}{c}}{\bm{r}})}{\psi}
    \]
    and the reduct is well-typed by \cref{rule:bridge-I,rule:bridge-betapartial,rule:bridge-E}. \qedhere
  \end{itemize}
\end{proof}

\begin{ruletheorem}[$\extent*$]
  \label{rule:extent}
  If $\bridgej{\bm{r}}$ and
  \begin{enumerate}
  \item $\cwftypek[\cx!\apartcx{\Phi}{\bm{r}},\bm{x}!]{A}$,
  \item $\wftypek[\cx!\apartcx{\Phi}{\bm{r}},\bm{x}!]{\oft{d}{A}}{B}$,
  \item $\ceqtm{M}{M'}{\bsubst{A}{\bm{r}}{\bm{x}}}$,
  \item $\eqtm[\cx!\apartcx{\Phi}{\bm{r}}!]{\oft{a}{\bsubst{A}{\bm{0}}{\bm{x}}}}{N}{N'}{\subst{\bsubst{B}{\bm{0}}{\bm{x}}}{a}{d}}$,
  \item $\eqtm[\cx!\apartcx{\Phi}{\bm{r}}!]{\oft{a'}{\bsubst{A}{\bm{1}}{\bm{x}}}}{P}{P'}{\subst{\bsubst{B}{\bm{1}}{\bm{x}}}{a'}{d}}$,
  \item $\eqtm[\cx!\apartcx{\Phi}{\bm{r}}!]{\oft{a}{\bsubst{A}{\bm{0}}{\bm{x}}},\oft{a'}{\bsubst{A}{\bm{1}}{\bm{x}}},\oft{c}{\Bridge{\bm{x}.A}{a}{a'}}}{Q}{Q'}{\Bridge{\bm{x}.\subst{B}{\bapp{c}{\bm{x}}}{d}}{N}{P}}$,
  \end{enumerate}
  then
  $\ceqtm{\extent{\bm{r}}{M}{a.N}{a'.P}{a.a'.c.Q}}{\extent{\bm{r'}}{M'}{a.N'}{a'.P'}{a.a'.c.Q'}}{\subst{\bsubst{B}{\bm{r}}{\bm{x}}}{M}{d}}$.
\end{ruletheorem}
\begin{proof}
  We have two cases: either $\models \bm{r} = \bm{\Ge}$ for some $\bm{\Ge} \in \{\bm{0},\bm{1}\}$ or not. In
  the former case, we reduce either side with \cref{rule:extent-beta0,rule:extent-beta1} and apply the typing
  assumptions to equate the reducts. In the latter case, $\bm{r}$ must be some variable $\bm{y} \in \Phi$. In
  that case, we can reduce either side with \cref{rule:extent-betaF} and then equate the reducts with the
  typing assumptions and \cref{rule:bridge-I,rule:bridge-betapartial,rule:bridge-E}.
\end{proof}

\begin{ruletheorem}[$\extent*$-$\eta$]
  \label{rule:extent-eta}
  If $\bridgej{\bm{r}}$ and
  \begin{enumerate}
  \item $\cwftypek[\cx!\apartcx{\Phi}{\bm{r}},\bm{x}!]{A}$,
  \item $\wftypek[\cx!\apartcx{\Phi}{\bm{r}},\bm{x}!]{\oft{d}{A}}{B}$,
  \item $\coftype{M}{\bsubst{A}{\bm{r}}{\bm{x}}}$,
  \item $\oftype[\cx!\apartcx{\Phi}{\bm{r}},\bm{x}!]{\oft{d}{A}}{N}{B}$,
  \end{enumerate}
  then
  \[
    \subst{\bsubst{N}{\bm{r}}{\bm{x}}}{M}{a}
    \eq
    \extent{\bm{r}}{M}{a.\subst{\bsubst{N}{\bm{0}}{\bm{x}}}{a}{d}}{a'.\subst{\bsubst{N}{\bm{1}}{\bm{x}}}{a'}{d}}{a.a'.c.\subst{\blam{\bm{x}}{N}}{\bapp{c}{\bm{x}}}{d}}
  \]
  in $\subst{\bsubst{B}{\bm{r}}{\bm{x}}}{M}{d}$ at $\cx*$.
\end{ruletheorem}
\begin{proof}
  By case analysis on $\bm{r}$. If $\bm{r} = \bm{0}$, this follows from \cref{rule:extent-beta0}; if
  $\bm{r} = \bm{1}$, it follows from \cref{rule:extent-beta1}. If $\bm{r} = \bm{y}$, then
  $M = \bsubst{\bsubst{M}{\bm{x}}{\bm{y}}}{\bm{r}}{\bm{x}}$, and we have
  \[
    \bsubst{\subst{N}{\bsubst{M}{\bm{x}}{\bm{y}}}{a}}{\bm{r}}{\bm{x}}
    \eq
    \extent{\bm{r}}{\bsubst{\bsubst{M}{\bm{x}}{\bm{y}}}{\bm{r}}{\bm{x}}}{a.\subst{\bsubst{N}{\bm{0}}{\bm{x}}}{a}{d}}{a'.\subst{\bsubst{N}{\bm{1}}{\bm{x}}}{a'}{d}}{a.a'.c.\subst{\blam{\bm{x}}{N}}{\bapp{c}{\bm{x}}}{d}}
  \]
  by \cref{rule:bridge-betaF,rule:extent-betaF}.
\end{proof}

\begin{remark}
  A weak version of the rule \rulename{$\extent*$-$\eta$}, in which one obtains a path rather than an equality,
  is derivable from the other rules for $\extent$ (much as with $\eta$-principles for positive types). Given
  the hypotheses of \cref{rule:extent-eta}, we have
  \[
    \extent{\bm{r}}{M}{a.\dlam{y}{\bsubst{N}{\bm{0}}{\bm{x}}}}{a.\bsubst{N}{\bm{1}}{\bm{x}}}{a.a'.c.\blam{\bm{x}}{\dlam{y}{\subst{N}{\dapp{c}{\bm{x}}}{a}}}}
  \]
  of type
  \[
  \Path{\subst{\bsubst{B}{\bm{r}}{\bm{x}}}{M}{d}}{\subst{\bsubst{N}{\bm{r}}{\bm{x}}}{M}{a}}{\extent{\bm{r}}{M}{a.\subst{\bsubst{N}{\bm{0}}{\bm{x}}}{a}{d}}{a'.\subst{\bsubst{N}{\bm{1}}{\bm{x}}}{a'}{d}}{a.a'.c.\subst{\blam{\bm{x}}{N}}{\bapp{c}{\bm{x}}}{d}}}.
  \]
  The weaker rule suffices for the proof of \cref{thm:function-bridge} below, which is the only place we use
  \rulename{$\extent*$-$\eta$}. Thus, the strict rule may safely be omitted from a proof theory.
\end{remark}

This completes the set of rules we need for $\extent$. Using these rules, we can now characterize the type of
bridges over a function type.

\begin{theorem}
  \label{thm:function-bridge}
  Let $\cwftypek[\cx!\Phi,\bm{x}!]{A}$, $\wftypek[\cx!\Phi,\bm{x}!]{\oft{a}{A}}{B}$,
  $\coftype{F}{\picl{a}{\bsubst{A}{\bm{0}}{\bm{x}}}{\bsubst{B}{\bm{0}}{\bm{x}}}}$, and
  $\coftype{F'}{\picl{a}{\bsubst{A}{\bm{1}}{\bm{x}}}{\bsubst{B}{\bm{1}}{\bm{x}}}}$ be given. Then there is an
  equivalence
  \[
    \Bridge{\bm{x}.\picl{a}{A}{B}}{F}{F'} \simeq \picl*{a}{\bsubst{A}{\bm{0}}{\bm{x}}}{\picl*{a'}{\bsubst{A}{\bm{1}}{\bm{x}}}{\picl{c}{\Bridge{\bm{x}.A}{a}{a'}}{\Bridge{\bm{x}.\subst{B}{\bapp{c}{\bm{x}}}{a}}{Fa}{F'a'}}}}.
  \]
\end{theorem}
\begin{proof}
  \newcommand{\bfunext}{\mathsf{bfunext}}
  \newcommand{\bfunapp}{\mathsf{bfunapp}}

  By \cref{rec:qequiv-to-equiv}. We define forward and backward functions as follows.
  \begin{align*}
    \bfunapp &\eqdef \lam{q}{\lam{a}{\lam{a'}{\lam{c}{\blam{\bm{x}}{(\bapp{q}{\bm{x}})(\bapp{c}{\bm{x}})}}}}} \\
    \bfunext^{F,F'} &\eqdef \lam{h}{\blam{\bm{x}}{\lam{d}{\extent{\bm{x}}{d}{a.Fa}{a'.F'a'}{a.a'.c.haa'c}}}}
  \end{align*}
  We show that $\bfunapp$ and $\bfunext^{F,F'}$ are mutually inverse, then apply
  \cref{rec:qequiv-to-equiv}. In a context with $q : \Bridge{\bm{x}.\picl{a}{A}{B}}{F}{F'}$, we have
  \begin{align*}
    \bfunext^{F,F'}(\bfunapp(q))
    &\eq \blam{\bm{x}}{\lam{d}{\extent{\bm{x}}{d}{a.Fa}{a'.F'a'}{a.a'.c.\blam{\bm{y}}{(\bapp{q}{\bm{y}})(\bapp{c}{\bm{x}})}}}} \\
    &\eq \blam{\bm{x}}{\lam{d}{\extent{\bm{x}}{d}{a.(\bapp{q}{\bm{0}})a}{a'.(\bapp{q}{\bm{1}})a'}{a.a'.c.\blam{\bm{y}}{(\bapp{q}{\bm{y}})(\bapp{c}{\bm{y}})}}}} \\
    &\eq \blam{\bm{x}}{\lam{d}{(\bapp{q}{\bm{x}})d}} \\
    &\eq q
  \end{align*}
  in $\Bridge{\bm{x}.\picl{a}{A}{B}}{F}{F'}$, where the first equation is $\beta$ for bridges and functions,
  the second is \cref{rule:bridge-betapartial}, the third is \cref{rule:extent-eta}, and the fourth is $\eta$
  for bridges and functions. For the other inverse, in a context with
  \[
    h :
    \picl{a}{\bsubst{A}{\bm{0}}{\bm{x}}}{\picl{a'}{\bsubst{A}{\bm{1}}{\bm{x}}}{\picl{c}{\Bridge{\bm{x}.A}{a}{a'}}{\Bridge{\bm{x}.\subst{B}{\bapp{c}{\bm{x}}}{a}}{Fa}{F'a'}}}},
  \]
  we have
  \begin{align*}
    \bfunapp(\bfunext^{F,F'}(h))
    &\eq \lam{a}{\lam{a'}{\lam{c}{\blam{\bm{x}}{\extent{\bm{x}}{\bapp{c}{\bm{x}}}{a.Fa}{a'.F'a'}{a.a'.k.haa'k}}}}} \\
    &\eq \lam{a}{\lam{a'}{\lam{c}{\bapp{h(\bapp{c}{\bm{0}})(\bapp{c}{\bm{1}})(\blam{\bm{y}}{\bapp{c}{\bm{y}}})}{\bm{x}}}}} \\
    &\eq \lam{a}{\lam{a'}{\lam{c}{\bapp{haa'(\blam{\bm{y}}{\bapp{c}{\bm{y}}})}{\bm{x}}}}} \\
    &\eq h
  \end{align*}
  in the same type, where the first step is $\beta$ for bridges and functions, the second is
  \cref{rule:extent-betaF}, the third is \cref{rule:bridge-betapartial}, and the fourth is $\eta$ for bridges
  and functions.
\end{proof}

One way to conceptualize the difference between structural and substructural dimensions is in terms of the
different ``function extensionality'' principles they provide. For substructural dimensions, we have the
theorem just proven. If bridges were structural, on the other hand, we would instead be able to prove the
following incomparable principle.
\[
  \Bridge{\picl{a}{A}{B}}{F}{F'} \overset{\times}{\simeq} \picl{a}{A}{\Bridge{B}{Fa}{F'a}}
\]
We would define this equivalence by taking $q$ in the former type to
$\lam{a}{\blam{\bm{x}}{(\bapp{q}{\bm{x}})a}}$ in the latter and $h$ in the latter to
$\blam{\bm{x}}{\lam{a}{\bapp{ha}{\bm{x}}}}$ in the former. While the first map is still well-defined
substructurally, the second is not: $\bm{x}$ is not fresh for $a$, so $ha$ cannot be applied at
$\bm{x}$. Conversely, without the $\extent$ operator which substructurality enables, we would not be able to
prove \cref{thm:function-bridge}. Note that on the path side, \emph{both} principles are provable; the
equivalent of \cref{thm:function-bridge} follows from the structural principle using Kan operations not
available on the bridge side. Likewise, the substructural cubical type theory of \cite{bch} enjoys the same
functional extensionality principle as the structural cubical type theories. Without Kan operations, however,
the principles are distinct, and it is the substructural version which matches the standard definition of a
logical relation at a function type.

Finally, we observe that the function extensionality principle induces a corresponding ``equivalence
extensionality'' principle, which we will use in the proof of relativity.

\begin{corollary}
  \label{cor:equiv-bridge}
  Let $\cwftypek[\cx!\apartcx{\Phi}{\bm{r}},\bm{x}!]{A,B}$. Suppose we have
  \begin{enumerate}
  \item $\coftype[\cx!\apartcx{\Phi}{\bm{r}}!]{E_{\bm{0}}}{\bsubst{A}{\bm{0}}{\bm{x}} \simeq \bsubst{B}{\bm{0}}{\bm{x}}}$,
  \item $\coftype[\cx!\apartcx{\Phi}{\bm{r}}!]{E_{\bm{1}}}{\bsubst{A}{\bm{1}}{\bm{x}} \simeq \bsubst{B}{\bm{1}}{\bm{x}}}$,
  \item $\oftype[\cx!\apartcx{\Phi}{\bm{r}}!]{\oft{a_{\bm{0}}}{\bsubst{A}{\bm{0}}{\bm{x}}},\oft{a_{\bm{1}}}{\bsubst{A}{\bm{1}}{\bm{x}}}}{E}{\Bridge{\bm{x}.A}{a_{\bm{0}}}{a_{\bm{1}}} \simeq \Bridge{\bm{x}.B}{\fwd{E_{\bm{0}}}(a_{\bm{0}})}{\fwd{E_{\bm{1}}}(a_{\bm{1}})}}$.
  \end{enumerate}
  Then there is a term
  \begin{align*}
    \coftype{\extentequiv{\bm{r}}{E_{\bm{0}}}{E_{\bm{1}}}{E}}{\bsubst{A}{\bm{r}}{\bm{x}} \simeq \bsubst{B}{\bm{r}}{\bm{x}}}
  \end{align*}
  that satisfies
  $\ceqtm[\cx<\bm{r}=\bm{\Ge}>]{\extentequiv{\bm{r}}{E_{\bm{0}}}{E_{\bm{1}}}{E}}{E_{\bm{\Ge}}}{\bsubst{A}{\bm{\Ge}}{\bm{x}}\simeq\bsubst{B}{\bm{\Ge}}{\bm{x}}}$
  for $\bm{\Ge} \in \{\bm{0},\bm{1}\}$.
\end{corollary}
\begin{proof}
  \newcommand{\encode}{\mathsf{in}_{\bm{r}}}
  \newcommand{\decode}{\mathsf{out}_{\bm{r}}}
  
  The proof is lengthy but straightforward; we will give a sketch. We first construct a quasi-equivalence
  between $\bsubst{A}{\bm{r}}{\bm{x}}$ and $\bsubst{B}{\bm{r}}{\bm{x}}$. We define a forward map $\mathsf{in}$
  by
  \begin{align*}
    \encode &\eqdef \lam{a}{\extent{\bm{r}}{a}{a_{\bm{0}}.\fwd{E_{\bm{0}}}(a_{\bm{0}})}{a_{\bm{1}}.\fwd{E_{\bm{1}}}(a_{\bm{1}})}{a_{\bm{0}}.a_{\bm{1}}.p.\fwd{E}(p)}}.
  \end{align*}
  To define the reverse map, we first derive a term
  \begin{gather*}
    \oftype[\cx!\apartcx{\Phi}{\bm{r}}!]{\oft{b_{\bm{0}}}{\bsubst{B}{\bm{0}}{\bm{x}}},\oft{b_{\bm{1}}}{\bsubst{B}{\bm{1}}{\bm{x}}}}{F}{\Bridge{\bm{x}.A}{\bwd{E_{\bm{0}}}(b_{\bm{0}})}{\bwd{E_{\bm{1}}}(b_{\bm{1}})} \simeq \Bridge{\bm{x}.B}{b_{\bm{0}}}{b_{\bm{1}}}}
  \end{gather*}
  from $E$ using the fact that $E_{\bm{0}}$ and $E_{\bm{1}}$ are equivalences. We then set
  \begin{align*}
    \decode &\eqdef \lam{b}{\extent{\bm{r}}{b}{b_{\bm{0}}.\bwd{E_{\bm{0}}}(b_{\bm{0}})}{b_{\bm{1}}.\bwd{E_{\bm{1}}}(b_{\bm{1}})}{b_{\bm{0}}.b_{\bm{1}}.q.\bwd{F}(q)}}.
  \end{align*}
  Proofs that $\encode$ and $\decode$ are mutually inverse can again be constructed by using $\extent$ to case
  on $\bm{r}$. By applying \cref{rec:qequiv-to-equiv}, this shows that $\encode$ is an equivalence.  However,
  although we have ensured that
  $\ceqtm[\cx<\bm{r}=\bm{\Ge}>]{\encode}{\fst{E_{\bm{\Ge}}}}{\bsubst{A}{\bm{\Ge}}{\bm{x}}\to\bsubst{B}{\bm{\Ge}}{\bm{x}}}$
  by construction, we do not know that the proof that $\encode$ is an equivalence has the correct boundary. To
  fix this, we use \cref{rec:isequiv-is-prop}, which implies that the boundary of our equivalence term is
  connected by a pair of paths to the desired boundary. We can then modify it using an $\hcom$ in
  $\bsubst{A}{\bm{r}}{\bm{x}} \simeq \bsubst{B}{\bm{r}}{\bm{x}}$ (with tube faces at $\bm{r} = \bm{0}$ and
  $\bm{r} = \bm{1}$) to construct an equivalence which has the correct boundary up to exact equality.
\end{proof}


\section[Gel-types]{$\Gel$-types}
\label{sec:gel}

The final constructor we need to complete the type theory is $\Gel$, which takes a relation between two types
and produces a bridge between them. This gives the inverse to the operation
$C \mapsto \Bridge{\bm{x}.\bapp{C}{\bm{x}}}{-}{-}$ mentioned at the beginning of \cref{sec:compound}, making
it possible to prove relativity (\cref{sec:relativity}). The $\Gel$ operator is for the bridge side what
$\V$ is for the path side, but there are important differences which again derive from and motivate the
substructurality of bridge dimensions.

Recall from \cref{sec:cubical} that the type $\V{x}{A}{B}{E}$ takes a type $A$ at $x = 0$, a type line $B$ in
$x$, and an equivalence $E$ between $A$ and $\dsubst{B}{0}{x}$ at $x = 0$, and composes them to create a new
type line in $x$ from $A$ to $\dsubst{B}{1}{x}$. By taking $B$ to be a constant path, which corresponds to an
identity equivalence, we can use $\V$ to convert any equivalence into a path. However, it appears necessary to
take this indirect route via composition with a line: we cannot restrict the typing rule of $\V$ to only allow
types $B$ which are constant in $x$, as such an apartness criterion is incompatible with diagonal
substitutions. (This is far from a complete justification of the shape of $\V$, but we hope it gives the
reader a sense of the situation.) This is a problem if we want to translate $\V$-types to the bridge
side. Unlike paths and equivalences, the constant bridge $B$ does \emph{not} necessarily correspond to the
identity relation on $B$, which is to say $\Path{B}{-}{-}$; rather, it corresponds to the bridge relation
$\Bridge{B}{-}{-}$. This means we could only use a ``bridge $\V$'' to construct bridges to $B$ corresponding
to relations that factor through $\Bridge{B}$.

Instead of resembling $\V$-types, a $\Gel$-type thus takes the form $\Gel{\bm{x}}{A}{B}{a.b.R}$ where both
types $A,B$ as well as the relation $R$ on $A$ and $B$ must be apart from $\bm{x}$. Besides being the binary
analogue of \citeauthor{bernardy15}'s $A \ni_i a$ types, $\Gel$-types are essentially the $\mathsf{G}$-types
of \cite{bezem17} adapted to the r\textbf{el}ational case, hence the name. The definitions of the Kan
operations for $\Gel$ are, however, much simpler than for $\V$ or $\mathsf{G}$: the principal direction of a
$\coe$ or $\hcom$ is always a path dimension, so can never coincide with the direction $\bm{x}$ of the
$\Gel$-type.

The operational semantics of $\Gel$-types are shown in \cref{fig:gel-opsem}. We define PER the semantics of
$\Gel$ and prove typing rules in this section, the latter of which are collected in inference rule format in
\cref{sec:proof-theory:gel}. The values of type $\Gel{\bm{x}}{A}{B}{a.b.R}$ are triples
$\gel{\bm{x}}{M}{N}{P}$ where $M$ is in $A$, $N$ is in $B$, and $P$ is a proof that the two are related by
$R$. The eliminator $\ungel$ takes a bridge over $\bm{x}.\Gel{\bm{x}}{A}{B}{a.b.R}$ and produces a proof that
its left and right endpoints (in $A$ and $B$ respectively) are related by $R$. This is a second point of
departure from $\V$ or $\mathsf{G}$ types. For equivalence-to-path types, the eliminator is a projection
function that takes (in the case of $\V$, for example) an element of $\V{r}{A}{B}{E}$ and extracts an element
of $B$. This is not possible with $\Gel{\bm{r}}{A}{B}{a.b.R}$, as there is no way to produce such an element
of $B$ when $\bm{r}$ is $\bm{0}$. As such, $\ungel$ takes not an element but a bridge over the $\Gel$-type. In
order to make use of $\ungel$, the implementations of $\hcom$ and $\coe$ for $\Gel{\bm{x}}$ capture
occurrences of $\bm{x}$ in their arguments, just as $\extent{\bm{x}}$ does.

\begin{figure}
  \begin{mdframed}
    \begin{mathpar}
      \Infer
      { }
      {\Gel{\bm{0}}{A}{B}{a.b.R} \steps A}
      \and
      \Infer
      { }
      {\Gel{\bm{1}}{A}{B}{a.b.R} \steps B}
      \and
      \Infer
      { }
      {\isval{\Gel{\bm{x}}{A}{B}{a.b.R}}}
      \and
      \Infer
      { }
      {\gel{\bm{0}}{M}{N}{P} \steps M}
      \and
      \Infer
      { }
      {\gel{\bm{1}}{M}{N}{P} \steps N}
      \and
      \Infer
      { }
      {\isval{\gel{\bm{x}}{M}{N}{P}}}
      \and
      \Infer
      {M \steps M'}
      {\ungel{\bm{x}.M} \steps \ungel{\bm{x}.M'}}
      \and
      \Infer
      { }
      {\ungel{\bm{x}.\gel{\bm{x}}{M}{N}{P}} \steps \bsubst{P}{\bm{0}}{\bm{x}}\footnote{In well-typed code, the substitution $\bsubst{}{\bm{0}}{\bm{x}}$ is always a no-op. We include it here only to ensure that the evaluation system is context-preserving.}}
      \and
      \Infer
      {M^y \eqdef \hcom{A}{r}{y}{\bsubst{O}{\bm{0}}{\bm{x}}}{\sys{\bsubst{\xi_i}{\bm{0}}{\bm{x}}}{y.\bsubst{Q_i}{\bm{0}}{\bm{x}}}} \\
        N^y \eqdef \hcom{B}{r}{y}{\bsubst{O}{\bm{1}}{\bm{x}}}{\sys{\bsubst{\xi_i}{\bm{1}}{\bm{x}}}{y.\bsubst{Q_i}{\bm{1}}{\bm{x}}}} \\
        P \eqdef \com{y.\subst{\subst{R}{M^y}{a}}{N^y}{b}}{r}{s}{\ungel{\bm{x}.O}}{(\sys{\xi_i}{y.\ungel{\bm{x}.Q_i}})_{\bm{x} \not\in \xi_i}}}
      {\hcom{\Gel{\bm{x}}{A}{B}{a.b.R}}{r}{s}{O}{\sys{\xi_i}{y.Q_i}} \steps
        \gel{\bm{x}}{M^s}{N^s}{P}
      }
      \and
      \Infer
      {M^y \eqdef \coe{y.A}{r}{y}{\bsubst{O}{\bm{0}}{\bm{x}}} \\
        N^y \eqdef \coe{y.B}{r}{y}{\bsubst{O}{\bm{1}}{\bm{x}}} \\
        P \eqdef \coe{y.\subst{\subst{R}{M^y}{a}}{N^y}{b}}{r}{s}{\ungel{\bm{x}.O}}}
      {\coe{y.\Gel{\bm{x}}{A}{B}{a.b.R}}{r}{s}{O} \steps \gel{\bm{x}}{M^s}{N^s}{P}}
    \end{mathpar}
  \end{mdframed}
  \caption{Operational semantics of $\Gel$-types}
  \label{fig:gel-opsem}
\end{figure}

\subsection{Definition}

\begin{definition}
  \label{def:gel}
  Let $\tau$ be a candidate type system. Let $\bridgej{\bm{r}}$ and
  \begin{enumerate}
  \item $\relcts*{\tau}{\cwftypek[\cx!\apartcx{\Phi}{\bm{r}}!]{A}}$,
  \item $\relcts*{\tau}{\cwftypek[\cx!\apartcx{\Phi}{\bm{r}}!]{B}}$,
  \item $\relcts*{\tau}{\wftypek[\cx!\apartcx{\Phi}{\bm{r}}!]{\oft{a}{A},\oft{b}{B}}{R}}$
  \end{enumerate}
  be given. We define a value $\cx*$-PER $\GelR[\tau]{\bm{r}}{A}{B}{a.b.R}$ by saying that, for each $\psitd$,
  $\GelR[\tau]{\bm{r}}{A}{B}{a.b.R}_\psi(V,V')$ holds iff one of the following holds:
  \begin{itemize}
  \item $\td{\bm{r}}{\psi} = \bm{0}$ and $\vper[\tau]{A}_\psi(V,V')$,
  \item $\td{\bm{r}}{\psi} = \bm{1}$ and $\vper[\tau]{B}_\psi(V,V')$,
  \item $\td{\bm{r}}{\psi} = \bm{y}$ and
    $V = \gel{\bm{y}}{M}{N}{P}$, $V' = \gel{\bm{y}}{M'}{N'}{P'}$ with
    \begin{enumerate}
    \item $\relcts*{\tau}{\ceqtm[\cx[']!\apartcx{\Phi'}{\bm{y}}!]{M}{M'}{\td{A}{\psi}}}$,
    \item $\relcts*{\tau}{\ceqtm[\cx[']!\apartcx{\Phi'}{\bm{y}}!]{N}{N'}{\td{B}{\psi}}}$,
    \item $\relcts*{\tau}{\ceqtm[\cx[']!\apartcx{\Phi'}{\bm{y}}!]{P}{P'}{\subst{\subst{\td{R}{\psi}}{M}{a}}{N}{b}}}$.
    \end{enumerate}
  \end{itemize}
  We will drop the superscript $\tau$ when it is inferable.
\end{definition}

\begin{lemma}
  \label{lem:gel-I0}
  If $\td{\bm{r}}{\psi} = \bm{0}$ and
  $\coftype[\cx[']]{M}{\td{A}{\psi}}$, then
  $\Tm{\GelR{\bm{r}}{A}{B}{a.b.R}}_\psi(\gel{\td{\bm{r}}{\psi}}{M}{N}{P},M)$.
\end{lemma}
\begin{proof}
  By \cref{lem:expansion}, as
  $\td{\gel{\bm{0}}{M}{N}{P}}{\psi'} \steps \td{M}{\psi'}$ for all $\psi'$.
\end{proof}

\begin{lemma}
  \label{lem:gel-I1}
  If $\td{\bm{r}}{\psi} = \bm{1}$ and
  $\coftype[\cx[']]{N}{\td{B}{\psi}}$, then
  $\Tm{\GelR{\bm{r}}{A}{B}{a.b.R}}_\psi(\gel{\td{\bm{r}}{\psi}}{M}{N}{P},N)$.
\end{lemma}
\begin{proof}
  By \cref{lem:expansion}, as
  $\td{\gel{\bm{1}}{M}{N}{P}}{\psi'} \steps \td{N}{\psi'}$ for all $\psi'$.
\end{proof}

\begin{lemma}
  \label{lem:gel-coherent}
  $\GelR{\bm{r}}{A}{B}{a.b.R}$ is value-coherent.
\end{lemma}
\begin{proof}
  Let $\GelR{\bm{r}}{A}{B}{a.b.R}_\psi(V,V')$ be given. If $\td{\bm{r}}{\psi} = \bm{0}$, then
  $\td{\GelR{\bm{r}}{A}{B}{a.b.R}}{\psi} = \td{\vper{A}}{\psi}$, so by value-coherence of $A$ we have
  $\Tm{\GelR{\bm{r}}{A}{B}{a.b.R}}_\psi(V,V')$. Likewise, if $\td{\bm{r}}{\psi} = \bm{1}$, then the same
  follows from value-coherence of $\vper{B}$. Now suppose $\td{\bm{r}}{\psi} = \bm{y}$. Then we have
  $V = \gel{\bm{y}}{M}{N}{P}$ and $V' = \gel{\bm{y}}{M'}{N'}{P'}$ satisfying the conditions in
  \cref{def:gel}. We go by \cref{lem:introduction}. For any $\psitd[']$, we are in one of three cases.
  \begin{itemize}
  \item $\td{\bm{r}}{\psi\psi'} = \bm{0}$.

    Then $\Tm{\GelR{\bm{r}}{A}{B}{a.b.R}}_{\psi\psi'}(\td{V}{\psi'},\td{V'}{\psi'})$ by \cref{lem:gel-I0},
    $\ceqtm[\cx[']!\apartcx{\Phi'}{\bm{y}}!]{M}{M'}{\td{A}{\psi}}$, and transitivity of
    $\Tm{\GelR{\bm{r}}{A}{B}{a.b.R}}$.
  \item $\td{\bm{r}}{\psi\psi'} = \bm{1}$.

    Then we apply \cref{lem:gel-I1} analogously to the previous case.
  \item $\td{\bm{r}}{\psi\psi'} \not\in \{\bm{0},\bm{1}\}$.

    Then $\GelR{\bm{r}}{A}{B}{a.b.R}_{\psi\psi'}(\td{V}{\psi'},\td{V'}{\psi'})$ by definition of
    $\GelR$. \qedhere
  \end{itemize}
\end{proof}

\begin{proposition}
  There exists a type system $\tau$ which, for all $\relcts*{\tau}{\ceqtypek{A}{A'}}$,
  $\relcts*{\tau}{\ceqtypek{B}{B'}}$, $\relcts*{\tau}{\eqtypek{\oft{a}{A},\oft{b}{B}}{R}{R'}}$, and
  $\bm{x} \not\in \Phi$, has
  \[
    \tau(\cx*!\Phi,\bm{x}!,\Gel{\bm{x}}{A}{B}{a.b.R},\Gel{\bm{x}}{A'}{B'}{a.b.R'},\GelR{\bm{x}}{A}{B}{a.b.R}_\id)
  \]
  in addition to $\Bridge$-types and the standard constructs of cubical type theory. Moreover, there exists
  such a type system in which the universe $\UKan$ is also closed under $\Bridge$- and $\Gel$-types.
\end{proposition}
\begin{proof}
  See \cref{app:fixed-point}.
\end{proof}

For the remainder of this section, we assume we are working in such a type system.

\subsection{Rules}

\begin{ruletheorem}[$\Gel$-F$_0$]
  \label{rule:gel-F0}
  If $\cwftypek{A}$ then $\ceqtypep{\Gel{\bm{0}}{A}{B}{a.b.R}}{A}$.
\end{ruletheorem}
\begin{proof}
  By \cref{lem:type-expansion}, as $\td{\Gel{\bm{0}}{A}{B}{a.b.R}}{\psi} \steps \td{A}{\psi}$ for all $\psi$.
\end{proof}

\begin{ruletheorem}[$\Gel$-F$_1$]
  \label{rule:gel-F1}
  If $\cwftypek{B}$, then $\ceqtypep{\Gel{\bm{1}}{A}{B}{a.b.R}}{B}$.
\end{ruletheorem}
\begin{proof}
  By \cref{lem:type-expansion}, as $\td{\Gel{\bm{1}}{A}{B}{a.b.R}}{\psi} \steps \td{B}{\psi}$ for all $\psi$.
\end{proof}

\begin{ruletheorem}[$\Gel$-F]
  \label{rule:gel-F}
  If
  \begin{enumerate}
  \item $\ceqtypek[\cx!\apartcx{\Phi}{\bm{r}}!]{A}{A'}$,
  \item $\ceqtypek[\cx!\apartcx{\Phi}{\bm{r}}!]{B}{B'}$,
  \item $\eqtypek[\cx!\apartcx{\Phi}{\bm{r}}!]{\oft{a}{A},\oft{b}{B}}{R}{R'}$,
  \end{enumerate}
  then $\ceqtypep{\Gel{\bm{r}}{A}{B}{a.b.R}}{\Gel{\bm{r}}{A'}{B'}{a.b.R'}}$.
\end{ruletheorem}
\begin{proof}
  By \cref{lem:formation}. Let $\psitd$ be given; we have three cases. If $\td{\bm{r}}{\psi} = \bm{0}$, then
  we have
  \[
    \PTy{\tau}(\cx*['],\td{\Gel{\bm{r}}{A}{B}{a.b.R}}{\psi},\td{\Gel{\bm{r}}{A'}{B'}{a.b.R'}}{\psi},\td{\GelR{\bm{r}}{A}{B}{a.b.R}}{\psi})
  \]
  by \cref{rule:gel-F0}, $\ceqtypek[\cx[']]{\td{A}{\psi}}{\td{A'}{\psi}}$, and transitivity of
  $\PTy{\tau}$. If $\td{\bm{r}}{\psi} = \bm{1}$, then the same follows by \cref{rule:gel-F1},
  $\ceqtypek[\cx[']]{\td{B}{\psi}}{\td{B'}{\psi}}$, and transitivity of $\PTy{\tau}$. Finally, if
  $\td{\bm{r}}{\psi} = \bm{y}$, then we have
  $\tau(\cx*['],\td{\Gel{\bm{r}}{A}{B}{a.b.R}}{\psi},\Gel{\bm{r}}{A'}{B'}{a.b.R'},\GelR{\bm{r}}{A}{B}{a.b.R}_\psi)$ by
  our assumption on $\tau$.
\end{proof}

The following three rules are simply restatements of \cref{lem:gel-I0,lem:gel-I1,lem:gel-coherent}.

\begin{ruletheorem}[$\Gel$-I$_0$]
  \label{rule:gel-I0}
  If $\coftype{M}{A}$, then $\ceqtm{\gel{\bm{0}}{M}{N}{P}}{M}{A}$.
\end{ruletheorem}

\begin{ruletheorem}[$\Gel$-I$_1$]
  \label{rule:gel-I1}
  If $\coftype{N}{B}$, then $\ceqtm{\gel{\bm{1}}{M}{N}{P}}{N}{B}$.
\end{ruletheorem}

\begin{ruletheorem}[$\Gel$-I]
  \label{rule:gel-I}
  If
  \begin{enumerate}
  \item $\ceqtm[\cx!\apartcx{\Phi}{\bm{r}}!]{M}{M'}{A}$,
  \item $\ceqtm[\cx!\apartcx{\Phi}{\bm{r}}!]{N}{N'}{B}$,
  \item $\ceqtm[\cx!\apartcx{\Phi}{\bm{r}}!]{P}{P'}{\subst{\subst{R}{M}{a}}{N}{b}}$,
  \end{enumerate}
  then $\ceqtm{\gel{\bm{r}}{M}{N}{P}}{\gel{\bm{r}}{M'}{N'}{P'}}{\Gel{\bm{r}}{A}{B}{a.b.R}}$.
\end{ruletheorem}

\begin{ruletheorem}[$\Gel$-$\beta$]
  \label{rule:gel-beta}
  If
  \begin{enumerate}
  \item $\coftype{M}{A}$,
  \item $\coftype{N}{B}$,
  \item $\coftype{P}{\subst{\subst{R}{M}{a}}{N}{b}}$,
  \end{enumerate}
  then $\ceqtm{\ungel{\bm{x}.\gel{\bm{x}}{M}{N}{P}}}{P}{\subst{\subst{R}{M}{a}}{N}{b}}$.
\end{ruletheorem}
\begin{proof}
  By \cref{lem:expansion}, as
  $\td{\ungel{\bm{x}.\gel{\bm{x}}{M}{N}{P}}}{\psi} \steps \td{\bsubst{P}{\bm{0}}{\bm{x}}}{\psi}$ for all
  $\psi$ and $\bm{x}$ does not occur in $P$.
\end{proof}

\begin{ruletheorem}[$\Gel$-E]
  \label{rule:gel-E}
  If $\cwftypek{A,B}$, $\wftypek{\oft{a}{A},\oft{b}{B}}{R}$, and
  \begin{enumerate}
  \item $\ceqtm[\cx!\Phi,\bm{x}!]{Q}{Q'}{\Gel{\bm{x}}{A}{B}{a.b.R}}$,
  \end{enumerate}
  then
  $\ceqtm{\ungel{\bm{x}.Q}}{\ungel{\bm{x}.Q'}}{\subst{\subst{R}{\bsubst{Q}{\bm{0}}{\bm{x}}}{a}}{\bsubst{Q}{\bm{1}}{\bm{x}}}{b}}$.
\end{ruletheorem}
\begin{proof}
  By \cref{lem:elimination} with
  $\evalcx{\cx*!\bm{x}!{\emp}}{\ungel{\bm{x}.\evhole},\ungel{\bm{x}.\evhole},\subst{\subst{R}{\bapp{(\blam{\bm{x}}{\evhole})}{\bm{0}}}{a}}{\bapp{(\blam{\bm{x}}{\evhole})}{\bm{1}}}{b}}{\cx*}$.  We need to check that
  the equations hold when $Q$ and $Q'$ are values in $\vper{\Gel{\bm{x}}{A}{B}{a.b.R}}$. In that case, each
  follows by reducing with \cref{rule:gel-beta} on either side and then applying the assumptions given by
  $\vper{\Gel{\bm{x}}{A}{B}{a.b.R}}(Q,Q')$.
\end{proof}

\begin{ruletheorem}[$\Gel$-$\eta$]
  \label{rule:gel-eta}
  If
  \begin{enumerate}
  \item $\coftype[\cx!\Phi,\bm{x}!]{Q}{\Gel{\bm{x}}{A}{B}{a.b.R}}$,
  \end{enumerate}
  then
  $\ceqtm{\bsubst{Q}{\bm{r}}{\bm{x}}}{\gel{\bm{r}}{\bsubst{Q}{\bm{0}}{\bm{x}}}{\bsubst{Q}{\bm{1}}{\bm{x}}}{\ungel{\bm{x}.Q}}}{\Gel{\bm{r}}{A}{B}{a.b.R}}$.
\end{ruletheorem}
\begin{proof}
  By \cref{lem:evaluation}, we have $Q \evals V$ with
  $\ceqtm[\cx!\Phi,\bm{x}!]{Q}{V}{\Gel{\bm{x}}{A}{B}{a.b.R}}$. By \cref{rule:gel-I,rule:gel-E}, it thus suffices
  to show
  \[
    \ceqtm{\bsubst{V}{\bm{r}}{\bm{x}}}{\gel{\bm{r}}{\bsubst{V}{\bm{0}}{\bm{x}}}{\bsubst{V}{\bm{1}}{\bm{x}}}{\ungel{\bm{x}.V}}}{\Gel{\bm{r}}{A}{B}{a.b.R}}.
  \]
  By definition of $\GelR$, we have $V = \gel{\bm{x}}{M}{N}{P}$ for some $\coftype{M}{A}$, $\coftype{N}{B}$,
  and $\coftype{P}{\subst{\subst{R}{M}{a}}{N}{b}}$.  We have
  \begin{enumerate}
  \item $\ceqtm{M}{\bsubst{V}{\bm{0}}{\bm{x}}}{A}$ by \cref{rule:gel-I0},
  \item $\ceqtm{N}{\bsubst{V}{\bm{1}}{\bm{x}}}{B}$ by \cref{rule:gel-I1},
  \item $\ceqtm{P}{\ungel{\bm{x}.V}}{\subst{\subst{R}{M}{a}}{N}{b}}$ by \cref{rule:gel-beta},
  \end{enumerate}
  so the equation follows by \cref{rule:gel-I}.
\end{proof}

\subsection{Kan conditions}

The definitions of the Kan operations for $\Gel{\bm{x}}{A}{B}{a.b.R}$ are actually quite simple: we extract
the endpoints and relation data from each argument, apply the corresponding Kan conditions at $A$,$B$, and $R$
respectively, and then recombine them with $\gel$. The endpoint data is extracted with the substitutions
$\bsubst{}{\bm{0}}{\bm{x}}$ and $\bsubst{}{\bm{1}}{\bm{x}}$, while the relation data is extracted with
$\ungel{\bm{x}.-}$. Note that in using the latter, we capture occurrences of $\bm{x}$ in the argument(s).

\begin{lemma}[$\coe{\Gel}$-$\beta_{\bm{\Ge}}$]
  \label{lem:coe-gel-betapartial}
  Let $\cwftypek[\cx!\apartcx{\Phi}{\bm{r}}!{\Psi,y}]{A,B}$ and
  $\wftypek[\cx!\apartcx{\Phi}{\bm{r}}!{\Psi,y}]{\oft{a}{A},\oft{b}{B}}{R}$. Suppose we have $\dimj{r,s}$ and
  $\coftype{O}{\dsubst{\Gel{\bm{r}}{A}{B}{a.b.R}}{r}{y}}$. Then
  \begin{enumerate}
  \item if $\bm{r} = \bm{0}$, then
    $\ceqtm{\coe{y.\Gel{\bm{r}}{A}{B}{a.b.R}}{r}{s}{O}}{\coe{y.A}{r}{s}{O}}{\dsubst{A}{s}{y}}$, and
  \item if $\bm{r} = \bm{1}$, then
    $\ceqtm{\coe{y.\Gel{\bm{r}}{A}{B}{a.b.R}}{r}{s}{O}}{\coe{y.B}{r}{s}{O}}{\dsubst{B}{s}{y}}$.
  \end{enumerate}
\end{lemma}
\begin{proof}
  By \cref{lem:expansion}, using the fact that $A$ and $B$ are $\coe$-Kan to type the reducts.
\end{proof}

\begin{lemma}[$\coe{\Gel}$-$\beta$]
  \label{lem:coe-gel-beta}
  Let $\cwftypek[\cx{\Psi,y}]{A,B}$ and $\wftypek[\cx{\Psi,y}]{\oft{a}{A},\oft{b}{B}}{R}$. Suppose we have
  $\dimj[\cx!\Phi,\bm{x}!]{r,s}$ and $\coftype[\cx!\Phi,\bm{x}!]{O}{\dsubst{\Gel{\bm{x}}{A}{B}{a.b.R}}{r}{y}}$.
  Then
  \[
    \ceqtm[\cx!\Phi,\bm{x}!]{\coe{y.\Gel{\bm{x}}{A}{B}{a.b.R}}{r}{s}{O}}{\gel{\bm{x}}{M^s}{N^s}{P}}{\dsubst{\Gel{\bm{x}}{A}{B}{a.b.R}}{s}{y}}
  \]
  where we define
  \begin{align*}
    M^y &\eqdef \coe{y.A}{r}{y}{\bsubst{O}{\bm{0}}{\bm{x}}} &
    N^y &\eqdef \coe{y.B}{r}{y}{\bsubst{O}{\bm{1}}{\bm{x}}} &
    P &\eqdef \coe{y.\subst{\subst{R}{M^y}{a}}{N^y}{b}}{r}{s}{\ungel{\bm{x}.O}}.
  \end{align*}
\end{lemma}
\begin{proof}
  Observe that the reduct is well-typed by the $\coe$-Kan conditions for $A$ and $B$, \cref{rule:gel-E}, and
  the $\coe$-Kan condition for $R$. We go by \cref{lem:expansion}. Let
  $\tds{\cx*[']}{\psi}{\cx*!\Phi,\bm{x}!}$ be given; we have three cases.
  \begin{itemize}
  \item $\td{\bm{x}}{\psi} = \bm{0}$.

    Then $\td{\coe{y.\Gel{\bm{x}}{A}{B}{a.b.R}}{r}{s}{O}}{\psi} \steps \td{\coe{y.A}{r}{s}{O}}{\psi}$. We have
    $\coftype[\cx[']]{\td{M^s}{\psi}}{\dsubst{\td{A}{\psi}}{s}{y}}$ by the
    $\coe$-Kan condition for $A$ and
    $\ceqtm[\cx[']]{\td{M^s}{\psi}}{\td{\gel{\bm{x}}{M^s}{N^s}{P}}{\psi}}{\td{A}{\psi}}$ by
    \cref{rule:gel-I0}.

  \item $\td{\bm{x}}{\psi} = \bm{1}$.

    Analogous to the $\td{\bm{x}}{\psi} = \bm{0}$ case.

  \item $\td{\bm{x}}{\psi} \not\in \{\bm{0},\bm{1}\}$.

    Then
    $\td{\coe{y.\Gel{\bm{x}}{A}{B}{a.b.R}}{r}{s}{O}}{\psi} \steps \td{\gel{\bm{x}}{M^s}{N^s}{P^s}}{\psi}$, and
    we have already shown that the reduct is well-typed. \qedhere
  \end{itemize}
\end{proof}

\begin{theorem}
  Let $\ceqtypek[\cx!\apartcx{\Phi}{\bm{r}}!]{A}{A'}$, $\ceqtypek[\cx!\apartcx{\Phi}{\bm{r}}!]{B}{B'}$, and
  $\eqtypek[\cx!\apartcx{\Phi}{\bm{r}}!]{\oft{a}{A},\oft{b}{B}}{R}{R'}$ be given. Then
  $\ceqtypep{\Gel{\bm{r}}{A}{B}{a.b.R}}{\Gel{\bm{r}}{A'}{B'}{a.b.R'}}$ are equally $\coe$-Kan.
\end{theorem}
\begin{proof}
  Let $\tds{\cx*[']{\Psi',y}}{\psi}{\cx*}$, $\dimj[{\cx[']}]{r,s}$, and
  $\ceqtm[\cx[']]{O}{O'}{\dsubst{\td{\Gel{\bm{r}}{A}{B}{a.b.R}}{\psi}}{r}{y}}$ be given. We need to show
  \begin{enumerate}
  \item
    $\ceqtm[\cx[']]{\coe{y.\td{\Gel{\bm{r}}{A}{B}{a.b.R}}{\psi}}{r}{s}{O}}{\coe{y.\td{\Gel{\bm{r}}{A'}{B'}{a.b.R'}}{\psi'}}{r}{s}{O'}}{\dsubst{\td{\Gel{\bm{r}}{A}{B}{a.b.R}}{\psi}}{s}{y}}$, and
  \item if $r=s$, then 
    $\ceqtm[\cx[']]{\coe{y.\td{\Gel{\bm{r}}{A}{B}{a.b.R}}{\psi}}{r}{s}{O}}{O}{\dsubst{\td{\Gel{\bm{r}}{A}{B}{a.b.R}}{\psi}}{s}{y}}$.
  \end{enumerate}
  We have three cases, depending on the status of $\td{\bm{r}}{\psi}$; we prove the two equations for each
  case in turn.
  \begin{itemize}
  \item $\td{\bm{r}}{\psi} = \bm{0}$.

    Then the equations follow from the assumption that $\ceqtypek[\cx[']]{\td{A}{\psi}}{\td{A'}{\psi}}$ by
    rewriting each $\coe$ term with \cref{lem:coe-gel-betapartial}.

  \item $\td{\bm{r}}{\psi} = \bm{1}$.

    Analogous to the $\td{\bm{r}}{\psi} = \bm{0}$ case.

  \item $\td{\bm{r}}{\psi} = \bm{x}$.

    We apply \cref{lem:coe-gel-beta}, which gives
    \[
      \ceqtm[\cx[']]{\coe{y.\td{\Gel{\bm{r}}{A}{B}{R}}{\psi}}{r}{s}{O}}{\gel{\bm{x}}{M^s}{N^s}{P}}{\dsubst{\td{\Gel{\bm{x}}{A}{B}{a.b.R}}{\psi}}{s}{y}}
    \]
    \[
      \ceqtm[\cx[']]{\coe{y.\td{\Gel{\bm{r}}{A'}{B'}{R'}}{\psi}}{r}{s}{O'}}{\gel{\bm{x}}{{M'}^s}{{N'}^s}{P'}}{\dsubst{\td{\Gel{\bm{x}}{A}{B}{a.b.R}}{\psi}}{s}{y}}
    \]
    where the reduct subterms are as defined there. We conclude that the first equation holds by a simple binary
    generalization of the well-typedness argument in the proof of that lemma.

    For the second equation, suppose that $r = s$. Then we have
    \[
      \ceqtm[\cx[']]{\gel{\bm{x}}{M^s}{N^s}{P}}{\gel{\bm{x}}{\bsubst{O}{\bm{0}}{\bm{x}}}{\bsubst{O}{\bm{1}}{\bm{x}}}{\ungel{\bm{x}.O}}}{\dsubst{\td{\Gel{\bm{x}}{A}{B}{a.b.R}}{\psi}}{s}{y}}
    \]
    by the $\coe$-Kan condition for $A$, $B$, and $R$, and the right-hand side is equal to $O$ by
    \cref{rule:gel-eta}. \qedhere
  \end{itemize}
\end{proof}

\begin{lemma}[$\hcom{\Gel}$-$\beta_{\bm{\Ge}}$]
  \label{lem:hcom-gel-betapartial}
  Let $\cwftypek[\cx!\apartcx{\Phi}{\bm{r}}!]{A,B}$ and
  $\wftypek[\cx!\apartcx{\Phi}{\bm{r}}!]{\oft{a}{A},\oft{b}{B}}{R}$. Suppose we have $\dimj{r,s}$,
  $\constraintsj{\etc{\xi_i}}$, and
  \begin{enumerate}
  \item $\coftype{O}{\Gel{\bm{r}}{A}{B}{a.b.R}}$,
  \item $\ceqtm[\cx{\Psi,y}<\xi_i,\xi_j>]{Q_i}{Q_j}{\Gel{\bm{r}}{A}{B}{a.b.R}}$ for all $i,j$,
  \item $\ceqtm[\cx<\xi_i>]{\dsubst{Q_i}{r}{y}}{O}{\Gel{\bm{r}}{A}{B}{a.b.R}}$ for all $i$.
  \end{enumerate}
  Then
  \begin{enumerate}
  \item if $\bm{r} = \bm{0}$, then
    $\ceqtm{\hcom{\Gel{\bm{r}}{A}{B}{a.b.R}}{r}{s}{O}{\sys{\xi_i}{Q_i}}}{\hcom{A}{r}{s}{O}{\sys{\xi_i}{Q_i}}}{A}$, and
  \item if $\bm{r} = \bm{1}$, then
    $\ceqtm{\hcom{\Gel{\bm{r}}{A}{B}{a.b.R}}{r}{s}{O}{\sys{\xi_i}{Q_i}}}{\hcom{B}{r}{s}{O}{\sys{\xi_i}{Q_i}}}{B}$.
  \end{enumerate}
\end{lemma}
\begin{proof}
  By \cref{lem:expansion}, using the fact that $A$ and $B$ are $\hcom$-Kan to type the reducts.
\end{proof}

\begin{lemma}[$\hcom{\Gel}$-$\beta$]
  \label{lem:hcom-gel-beta}
  Let $\cwftypek{A,B}$ and $\wftypek{\oft{a}{A},\oft{b}{B}}{R}$. Suppose we have
  $\dimj[\cx!\Phi,\bm{x}!]{r,s}$, $\constraintsj[\cx!\Phi,\bm{x}!]{\etc{\xi_i}}$,
  and
  \begin{enumerate}
  \item $\coftype[\cx!\Phi,\bm{x}!]{O}{\Gel{\bm{x}}{A}{B}{a.b.R}}$,
  \item $\ceqtm[\cx!\Phi,\bm{x}!{\Psi,y}<\xi_i,\xi_j>]{Q_i}{Q_j}{\Gel{\bm{x}}{A}{B}{a.b.R}}$ for all $i,j$,
  \item
    $\ceqtm[\cx!\Phi,\bm{x}!<\xi_i>]{\dsubst{Q_i}{r}{y}}{O}{\Gel{\bm{x}}{A}{B}{a.b.R}}$
    for all $i$.
  \end{enumerate}
  Then
  \[
    \ceqtm[\cx!\Phi,\bm{x}!]{\hcom{\Gel{\bm{x}}{A}{B}{a.b.R}}{r}{s}{O}{\sys{\xi_i}{Q_i}}}{\gel{\bm{x}}{M^s}{N^s}{P}}{\Gel{\bm{x}}{A}{B}{a.b.R}}
  \]
  where we define
  \[
    M^y \eqdef \hcom{A}{r}{y}{\bsubst{O}{\bm{0}}{\bm{x}}}{\sys{\bsubst{\xi_i}{\bm{0}}{\bm{x}}}{y.\bsubst{Q_i}{\bm{0}}{\bm{x}}}} \qquad
    N^y \eqdef \hcom{B}{r}{y}{\bsubst{O}{\bm{1}}{\bm{x}}}{\sys{\bsubst{\xi_i}{\bm{1}}{\bm{x}}}{y.\bsubst{Q_i}{\bm{1}}{\bm{x}}}}
  \]\[
    P \eqdef \com{y.\subst{\subst{R}{M^y}{a}}{N^y}{b}}{r}{s}{\ungel{\bm{x}.O}}{(\sys{\xi_i}{y.\ungel{\bm{x}.Q_i}})_{\bm{x} \not\in \xi_i}}.
  \]
\end{lemma}
\begin{proof}
  We first argue that the reduct is well-typed. First, we have $\coftype[\cx{\Psi,y}]{M^y}{A}$ by
  \cref{rule:gel-F0} and the $\hcom$-Kan condition for $A$, likewise $\coftype[\cx{\Psi,y}]{N^y}{B}$ by
  \cref{rule:gel-F1} and the $\hcom$-Kan condition for $B$. Second, \cref{rule:gel-E} gives
  \begin{enumerate}
  \item
    $\coftype{\ungel{\bm{x}.O}}{\subst{\subst{R}{\bsubst{O}{\bm{0}}{\bm{x}}}{a}}{\bsubst{O}{\bm{1}}{\bm{x}}}{b}}$,
  \item
    $\ceqtm[\cx{\Psi,y}<\xi_i,\xi_j>]{\ungel{\bm{x}.Q_i}}{\ungel{\bm{x}.Q_j}}{\subst{\subst{R}{\bsubst{Q_i}{\bm{0}}{\bm{x}}}{a}}{\bsubst{Q_i}{\bm{1}}{\bm{x}}}{b}}$
    for all $i,j$ with $\bm{x} \not\in \xi_i,\xi_j$,
  \item
    $\ceqtm[\cx<\Xi,\xi_i>]{\ungel{\bm{x}.\dsubst{Q_i}{r}{y}}}{\ungel{\bm{x}.O}}{\subst{\subst{R}{\bsubst{O}{\bm{0}}{\bm{x}}}{a}}{\bsubst{O}{\bm{1}}{\bm{x}}}{b}}$
    for all $i$ with $\bm{x} \not\in \xi_i$.
  \end{enumerate}
  Thus $\coftype{P}{\subst{\subst{R}{M^s}{a}}{N^s}{b}}$ by \cref{prop:com}, and so
  $\coftype[\cx!\Phi,\bm{x}!]{\gel{\bm{x}}{M^s}{N^s}{P}}{\Gel{\bm{x}}{A}{B}{a.b.R}}$ by \cref{rule:gel-I}.

  Now we prove the desired equation using \cref{lem:expansion}. Let $\psitd$ be given; we have three cases.
  \begin{itemize}
  \item $\td{\bm{x}}{\psi} = \bm{0}$.

    Then
    $\td{\hcom{\Gel{\bm{x}}{A}{B}{a.b.R}}{r}{s}{O}{\sys{\xi_i}{Q_i}}}{\psi} \steps
    \td{\hcom{A}{r}{s}{O}{\sys{\xi_i}{Q_i}}}{\psi}$. By the $\hcom$-Kan condition for $A$ we have
    $\ceqtm[\cx[']]{\td{\hcom{A}{r}{s}{O}{\sys{\xi_i}{Q_i}}}{\psi}}{\td{M^s}{\psi}}{\td{A}{\psi}}$, and then
    $\ceqtm[\cx[']]{\td{M^s}{\psi}}{\td{\gel{\bm{x}}{M^s}{N^s}{P}}{\psi}}{\td{A}{\psi}}$ by
    \cref{rule:gel-I0}.

  \item $\td{\bm{x}}{\psi} = \bm{1}$.

    Analogous to the $\Xi' \models \td{\bm{x}}{\psi} = \bm{0}$ case.

  \item $\td{\bm{x}}{\psi} \not\in \{\bm{0},\bm{1}\}$.

    Then
    $\td{\hcom{\Gel{\bm{x}}{A}{B}{a.b.R}}{r}{s}{O}{\sys{\xi_i}{Q_i}}}{\psi} \steps
    \td{\gel{\bm{x}}{M^s}{N^s}{P^s}}{\psi}$, which is well-typed by the argument above. \qedhere
  \end{itemize}
\end{proof}

\begin{theorem}
  Let $\ceqtypek[\cx!\apartcx{\Phi}{\bm{r}}!]{A}{A'}$, $\ceqtypek[\cx!\apartcx{\Phi}{\bm{r}}!]{B}{B'}$, and
  $\eqtypek[\cx!\apartcx{\Phi}{\bm{r}}!]{\oft{a}{A},\oft{b}{B}}{R}{R'}$ be given. Then
  $\ceqtypep{\Gel{\bm{r}}{A}{B}{a.b.R}}{\Gel{\bm{r}}{A'}{B'}{a.b.R'}}$ are equally $\hcom$-Kan.
\end{theorem}
\begin{proof}
  Let $\psitd$, $\dimj[{\cx[']}]{r,s}$, $\constraintsj[{\cx[']}]{\etc{\xi_i}}$ be given, and suppose we have
  \begin{enumerate}
  \item $\ceqtm[\cx[']]{O}{O'}{\td{\Gel{\bm{r}}{A}{B}{a.b.R}}{\psi}}$,

  \item $\ceqtm[\cx[']{\Psi',y}<\xi_i,\xi_j>]{Q_i}{Q'_j}{\td{\Gel{\bm{r}}{A}{B}{a.b.R}}{\psi}}$ for all $i,j$,
  \item
    $\ceqtm[\cx[']<\xi_i>]{\dsubst{Q_i}{r}{y}}{O}{\td{\Gel{\bm{r}}{A}{B}{a.b.R}}{\psi}}$
    for all $i$.
  \end{enumerate}
  Abbreviating $C \eqdef \Gel{\bm{r}}{A}{B}{a.b.R}$ and $C' \eqdef \Gel{\bm{r}}{A'}{B'}{a.b.R'}$, we need
  to show
  \begin{enumerate}
  \item
    $\ceqtm[\cx[']]{\hcom{\td{C}{\psi}}{r}{s}{O}{\sys{\xi_i}{y.Q_i}}}{\hcom{\td{C'}{\psi}}{r}{s}{O'}{\sys{\xi_i}{y.Q'_i}}}{\td{C}{\psi}}$,
  \item
    $\ceqtm[\cx[']]{\hcom{\td{C}{\psi}}{r}{s}{O}{\sys{\xi_i}{y.Q_i}}}{\dsubst{Q_i}{s}{y}}{\td{C}{\psi}}$
    for all $i$ with $\models \xi_i$,
  \item $\ceqtm[\cx[']]{\hcom{\td{C}{\psi}}{r}{s}{O}{\sys{\xi_i}{y.Q_i}}}{O}{\td{C}{\psi}}$ if $r=s$.
  \end{enumerate}

  We have three cases, depending on the status of $\td{\bm{r}}{\psi}$; we prove the three equations for each
  case.
  \begin{itemize}
  \item $\td{\bm{r}}{\psi} = \bm{0}$.

    Then the equations follow from the assumption that $\ceqtypep[\cx[']]{\td{A}{\psi}}{\td{A'}{\psi}}$ are
    equally $\hcom$-Kan by rewriting each $\hcom$ term with \cref{lem:hcom-gel-betapartial}.

  \item $\td{\bm{r}}{\psi} = \bm{1}$.

    Analogous to the $\td{\bm{r}}{\psi} = \bm{0}$ case.

  \item $\td{\bm{r}}{\psi} = \bm{x}$.

    We apply \cref{lem:hcom-gel-beta}, which gives
    \[
      \ceqtm[\cx[']]{\hcom{\td{C}{\psi}}{r}{s}{O}{\sys{\xi_i}{Q_i}}}{\gel{\bm{x}}{M^s}{N^s}{P}}{\td{\Gel{\bm{x}}{A}{B}{a.b.R}}{\psi}}
    \]
    \[
      \ceqtm[\cx[']]{\hcom{\td{C'}{\psi}}{r}{s}{O'}{\sys{\xi_i}{Q'_i}}}{\gel{\bm{x}}{{M'}^s}{{N'}^s}{P'}}{\td{\Gel{\bm{x}}{A}{B}{a.b.R}}{\psi}}
    \]
    where the reduct subterms are as defined there. We conclude that the first equation holds by a simple
    binary generalization of the well-typedness argument in the proof of that lemma.

    For the second equation, let $i$ be given with $\models \xi_i$. If $\xi_i = (\bm{x} = \bm{0})$, then we
    have $\ceqtm[\cx[']]{\gel{\bm{x}}{M^s}{N^s}{P}}{M^s}{\td{\Gel{\bm{x}}{A}{B}{a.b.R}}{\psi}}$ by
    \cref{rule:gel-I0}, and
    $\ceqtm[\cx[']<\xi_i>]{M^s}{\bsubst{Q_i}{\bm{0}}{\bm{x}} \eq Q_i}{\td{\Gel{\bm{x}}{A}{B}{a.b.R}}{\psi}}$
    by the $\hcom$-Kan condition for $A$. The $\xi_i = (\bm{x} = \bm{1})$ case is similar. Finally, if
    $\bm{x} \not\in \xi_i$, we have
    \[
      \ceqtm[\cx[']<\xi_i>]{\gel{\bm{x}}{M^s}{N^s}{P}}{\gel{\bm{x}}{\bsubst{Q_i}{\bm{0}}{\bm{x}}}{\bsubst{Q_i}{\bm{1}}{\bm{x}}}{\ungel{\bm{x}.Q_i}}}{\td{\Gel{\bm{x}}{A}{B}{a.b.R}}{\psi}}
    \]
    by the $\hcom$-Kan condition for $A,B$ and \cref{prop:com} for $R$. The right-hand side is then equal to
    $Q_i$ by \cref{rule:gel-eta}.

    For the third equation, suppose $r = s$. As above, we have
    \[
      \ceqtm[\cx[']]{\gel{\bm{x}}{M^s}{N^s}{P}}{\gel{\bm{x}}{\bsubst{O}{\bm{0}}{\bm{x}}}{\bsubst{O}{\bm{1}}{\bm{x}}}{\ungel{\bm{x}.O}}}{\td{\Gel{\bm{x}}{A}{B}{a.b.R}}{\psi}}
    \]
    by the $\hcom$-Kan condition for $A,B$ and \cref{prop:com} for $R$, and the right-hand side is equal to
    $O$ by \cref{rule:gel-eta}. \qedhere
  \end{itemize}
\end{proof}


\section{Proof theory}
\label{sec:proof-theory}

Having introduced all the operators we will need ($\Bridge$, $\extent$, and $\Gel$), we now collect the rules
we have proven into a makeshift proof theory for parametric cubical type theory. This is not meant to be a
complete or definitive set of rules: we omit rules for constructs already introduced by \cite{chtt-iii}, do
not include the (notationally burdensome) rules for $\hcom$ and $\coe$ at specific types, and only list those
structural rules which concern the treatment of bridge dimensions. Rather, it is intended to be a sufficient
basis for the remainder of the paper, in which we prove various results within the proof theory without
reference to the operational semantics.

\subsection{Structural}
\label{sec:proof-theory:structural}

\begin{mathpar}
  \Infer
  {\judg[\cx<\Xi>]{\GG\GG'}{\J} \\ \bridgej{\bm{r}}}
  {\judg[\cx<\Xi>]{\GG,\bm{r},\GG'}{\J}}
  \and
  \Infer
  {\judg[\cx<\Xi>]{\GG,\bm{\Ge},\GG'}{\J} \\ \bm{\Ge} \in \{\bm{0},\bm{1}\}}
  {\judg[\cx<\Xi>]{\GG\GG'}{\J}}
\end{mathpar}

\subsection{Kan operations}
\label{sec:proof-theory:kan}

\begin{mathpar}
  \Infer
  {\ceqtypek[\cx{\Psi,z}<\Xi>]{A}{A'} \\
    \ceqtm[\cx<\Xi>]{M}{M'}{\dsubst{A}{r}{z}}}
  {\ceqtm[\cx<\Xi>]{\coe{z.A}{r}{s}{M}}{\coe{z.A'}{r}{s}{M'}}{\dsubst{A}{s}{z}}}
  \and
  \Infer
  {\cwftypek[\cx{\Psi,z}<\Xi>]{A} \\
    \coftype[\cx<\Xi>]{M}{\dsubst{A}{r}{z}}}
  {\ceqtm[\cx<\Xi>]{\coe{z.A}{r}{r}{M}}{M}{\dsubst{A}{r}{z}}}
  \and
  \Infer
  {\ceqtypek[\cx<\Xi>]{A}{A'} \\
    \ceqtm[\cx<\Xi>]{M}{M'}{A} \\\\
    (\forall i,j)\;\ceqtm[\cx{\Psi,y}<\Xi,\xi_i,\xi_j>]{N_i}{N_j}{A} \\
    (\forall i)\;\ceqtm[\cx<\Xi,\xi_i>]{\dsubst{N_i}{r}{y}}{M}{A}}
  {\ceqtm[\cx<\Xi>]{\hcom{A}{r}{s}{M}{\sys{\xi_i}{y.N_i}}}{\hcom{A'}{r}{s}{M'}{\sys{\xi_i}{y.N_i'}}}{A}}
  \and
  \Infer
  {\cwftypek[\cx<\Xi>]{A} \\
    \coftype[\cx<\Xi>]{M}{A} \\\\
    (\forall i,j)\;\ceqtm[\cx{\Psi,y}<\Xi,\xi_i,\xi_j>]{N_i}{N'_j}{A} \\
    (\forall i)\;\ceqtm[\cx<\Xi,\xi_i>]{\dsubst{N_i}{r}{y}}{M}{A} \\ \models \xi_i}
  {\ceqtm[\cx<\Xi>]{\hcom{A}{r}{s}{M}{\sys{\xi_i}{y.N_i}}}{\dsubst{N_i}{s}{y}}{A}}
  \and
  \Infer
  {\cwftypek[\cx<\Xi>]{A} \\
    \coftype[\cx<\Xi>]{M}{A} \\\\
    (\forall i,j)\;\ceqtm[\cx{\Psi,y}<\Xi,\xi_i,\xi_j>]{N_i}{N'_j}{A} \\
    (\forall i)\;\ceqtm[\cx<\Xi,\xi_i>]{\dsubst{N_i}{r}{y}}{M}{A}}
  {\ceqtm[\cx<\Xi>]{\hcom{A}{r}{r}{M}{\sys{\xi_i}{y.N_i}}}{M}{A}}
\end{mathpar}

\subsection[Bridge-types]{$\Bridge$-types}
\label{sec:proof-theory:bridge}

\begin{mathpar}
  \Infer[\hyperref[rule:bridge-F]{($\Bridge$-F)}]
  {\eqtypek[\cx!\Phi,\bm{x}!<\Xi>]{\GG,\bm{x}}{A}{A'} \\
    \eqtm[\cx<\Xi>]{\GG}{M_{\bm{0}}}{M_{\bm{0}}'}{\bsubst{A}{\bm{0}}{\bm{x}}} \\
    \eqtm[\cx<\Xi>]{\GG}{M_{\bm{1}}}{M_{\bm{1}}'}{\bsubst{A}{\bm{1}}{\bm{x}}}}
  {\eqtypek[\cx<\Xi>]{\GG}{\Bridge{\bm{x}.A}{M_{\bm{0}}}{M_{\bm{1}}}}{\Bridge{\bm{x}.A'}{M_{\bm{0}}'}{M_{\bm{1}}'}}}
  \and
  \Infer[\hyperref[rule:bridge-I]{($\Bridge$-I)}]
  {\eqtm[\cx!\Phi,\bm{x}!<\Xi>]{\GG,\bm{x}}{P}{P'}{A} \\
    \eqtm[\cx<\Xi>]{\GG}{\bsubst{P}{\bm{0}}{\bm{x}}}{M_{\bm{0}}}{A} \\
    \eqtm[\cx<\Xi>]{\GG}{\bsubst{P}{\bm{1}}{\bm{x}}}{M_{\bm{1}}}{A}}
  {\eqtm[\cx<\Xi>]{\GG}{\blam{\bm{x}}{P}}{\blam{\bm{x}}{P'}}{\Bridge{\bm{x}.A}{M_{\bm{0}}}{M_{\bm{1}}}}}
  \and
  \Infer[\hyperref[rule:bridge-E]{($\Bridge$-E)}]
  {\eqtm[\cx!\apartcx{\Phi}{\bm{r}}!<\apartcx{\Xi}{\bm{r}}>]{\apartcx{\GG}{\bm{r}}}{Q}{Q'}{\Bridge{\bm{x}.A}{M_{\bm{0}}}{M_{\bm{1}}}}}
  {\eqtm[\cx<\Xi>]{\GG}{\bapp{Q}{\bm{r}}}{\bapp{Q'}{\bm{r}}}{\bsubst{A}{\bm{r}}{\bm{x}}}}
  \and
  \Infer[\hyperref[rule:bridge-betaF]{($\Bridge$-$\beta$)}]
  {\oftype[\cx!\apartcx{\Phi}{\bm{r}},\bm{x}!<\apartcx{\Xi}{\bm{r}}>]{\apartcx{\GG}{\bm{r}},\bm{x}}{P}{A}}
  {\eqtm[\cx<\Xi>]{\GG}{\bapp{(\blam{\bm{x}}{P})}{\bm{r}}}{\bsubst{P}{\bm{r}}{\bm{x}}}{\bsubst{A}{\bm{r}}{\bm{x}}}}
  \and
  \Infer[\hyperref[rule:bridge-betapartial]{($\Bridge$-$\beta_{\bm{\Ge}}$)}]
  {\oftype[\cx<\Xi>]{\GG}{Q}{\Bridge{\bm{x}.A}{M_{\bm{0}}}{M_{\bm{1}}}}}
  {\eqtm[\cx<\Xi>]{\GG}{\bapp{Q}{\bm{\Ge}}}{M_{\bm{\Ge}}}{\bsubst{A}{\bm{\Ge}}{\bm{x}}}}
  \and
  \Infer[\hyperref[rule:bridge-eta]{($\Bridge$-$\eta$)}]
  {\oftype[\cx<\Xi>]{\GG}{Q}{\Bridge{\bm{x}.A}{M_{\bm{0}}}{M_{\bm{1}}}}}
  {\eqtm[\cx<\Xi>]{\GG}{Q}{\blam{\bm{y}}{\bapp{Q}{\bm{y}}}}{\Bridge{\bm{x}.A}{M_{\bm{0}}}{M_{\bm{1}}}}}
\end{mathpar}

\subsection{Extent}
\label{sec:proof-theory:extent}

\begin{mathpar}
  \Infer[\hyperref[rule:extent]{($\extent*$)}]
  {\wftypek[\cx!\apartcx{\Phi}{\bm{r}},\bm{x}!<\apartcx{\Xi}{\bm{r}}>]{\apartcx{\GG}{\bm{r}},\bm{x}}{A} \\
    \wftypek[\cx!\apartcx{\Phi}{\bm{r}},\bm{x}!<\apartcx{\Xi}{\bm{r}}>]{\apartcx{\GG}{\bm{r}},\bm{x},\oft{d}{A}}{B} \\
    \eqtm[\cx<\Xi>]{\GG}{M}{M'}{\bsubst{A}{\bm{r}}{\bm{x}}} \\
    \eqtm[\cx!\apartcx{\Phi}{\bm{r}}!<\apartcx{\Xi}{\bm{r}}>]{\apartcx{\GG}{\bm{r}},\oft{a}{\bsubst{A}{\bm{0}}{\bm{x}}}}{N}{N'}{\subst{\bsubst{B}{\bm{0}}{\bm{x}}}{a}{d}} \\
    \eqtm[\cx!\apartcx{\Phi}{\bm{r}}!<\apartcx{\Xi}{\bm{r}}>]{\apartcx{\GG}{\bm{r}},\oft{a'}{\bsubst{A}{\bm{1}}{\bm{x}}}}{P}{P'}{\subst{\bsubst{B}{\bm{1}}{\bm{x}}}{a'}{d}} \\
    \eqtm[\cx!\apartcx{\Phi}{\bm{r}}!<\apartcx{\Xi}{\bm{r}}>]{\apartcx{\GG}{\bm{r}},\oft{a}{\bsubst{A}{\bm{0}}{\bm{x}}},\oft{a'}{\bsubst{A}{\bm{1}}{\bm{x}}},\oft{c}{\Bridge{\bm{x}.A}{a}{a'}}}{Q}{Q'}{\Bridge{\bm{x}.\subst{B}{\bapp{c}{\bm{x}}}{d}}{N}{P}}
  }
  {\eqtm[\cx<\Xi>]{\GG}{\extent{\bm{r}}{M}{a.N}{a'.P}{a.a'.c.Q}}{\extent{\bm{r}}{M'}{a.N'}{a'.P'}{a.a'.c.Q'}}{\subst{\bsubst{B}{\bm{r}}{\bm{x}}}{M}{d}}}
  \and
  \Infer[\hyperref[rule:extent-beta0]{($\extent*$-$\beta_{\bm{0}}$)}]
  {\wftypek[\cx<\Xi>]{\GG}{A} \\
    \wftypek[\cx<\Xi>]{\GG,\oft{d}{A}}{B} \\\\
    \oftype[\cx<\Xi>]{\GG}{M}{A} \\
    \oftype[\cx<\Xi>]{\GG,\oft{a}{A}}{N}{\subst{B}{a}{d}}
  }
  {\eqtm[\cx<\Xi>]{\GG}{\extent{\bm{0}}{M}{a.N}{a'.P}{a.a'.c.Q}}{\subst{N}{M}{a}}{\subst{B}{M}{d}}}
  \and
  \Infer[\hyperref[rule:extent-beta1]{($\extent*$-$\beta_{\bm{1}}$)}]
  {\wftypek[\cx<\Xi>]{\GG}{A} \\
    \wftypek[\cx<\Xi>]{\GG,\oft{d}{A}}{B} \\\\
    \oftype[\cx<\Xi>]{\GG}{M}{A} \\
    \oftype[\cx<\Xi>]{\GG,\oft{a'}{A}}{P}{\subst{B}{a'}{d}}
  }
  {\eqtm[\cx<\Xi>]{\GG}{\extent{\bm{1}}{M}{a.N}{b.P}{a.a'.c.Q}}{\subst{P}{M}{a'}}{\subst{B}{M}{d}}}
  \and
  \Infer[\hyperref[rule:extent-betaF]{($\extent*$-$\beta$)}]
  {\wftypek[\cx!\apartcx{\Phi}{\bm{r}},\bm{x}!<\apartcx{\Xi}{\bm{r}}>]{\apartcx{\GG}{\bm{r}},\bm{x}}{A} \\
    \wftypek[\cx!\apartcx{\Phi}{\bm{r}},\bm{x}!<\apartcx{\Xi}{\bm{r}}>]{\apartcx{\GG}{\bm{r}},\bm{x},\oft{d}{A}}{B} \\
    \oftype[\cx!\apartcx{\Phi}{\bm{r}},\bm{x}!<\apartcx{\Xi}{\bm{r}}>]{\apartcx{\GG}{\bm{r}},\bm{x}}{M}{A} \\
    \oftype[\cx!\apartcx{\Phi}{\bm{r}}!<\apartcx{\Xi}{\bm{r}}>]{\apartcx{\GG}{\bm{r}},\oft{a}{\bsubst{A}{\bm{0}}{\bm{x}}}}{N}{\subst{\bsubst{B}{\bm{0}}{\bm{x}}}{a}{d}} \\
    \oftype[\cx!\apartcx{\Phi}{\bm{r}}!<\apartcx{\Xi}{\bm{r}}>]{\apartcx{\GG}{\bm{r}},\oft{a'}{\bsubst{A}{\bm{1}}{\bm{x}}}}{P}{\subst{\bsubst{B}{\bm{1}}{\bm{x}}}{a'}{d}} \\
    \oftype[\cx!\apartcx{\Phi}{\bm{r}}!<\apartcx{\Xi}{\bm{r}}>]{\apartcx{\GG}{\bm{r}},\oft{a}{\bsubst{A}{\bm{0}}{\bm{x}}},\oft{a'}{\bsubst{A}{\bm{1}}{\bm{x}}},\oft{c}{\Bridge{\bm{x}.A}{a}{a'}}}{Q}{\Bridge{\bm{x}.\subst{B}{\bapp{c}{\bm{x}}}{d}}{N}{P}}    
  }
  {\eqtm[\cx<\Xi>]{\GG}{\extent{\bm{r}}{\bsubst{M}{\bm{r}}{\bm{x}}}{a.N}{a'.P}{a.a'.c.Q}}{T}{\bsubst{\subst{B}{M}{d}}{\bm{r}}{\bm{x}}} \\\\
  \text{where } T \eqdef \bapp{\subst{\subst{\subst{Q}{\bsubst{M}{\bm{0}}{\bm{x}}}{a}}{\bsubst{M}{\bm{1}}{\bm{x}}}{a'}}{\blam{\bm{x}}{M}}{c}}{\bm{r}}}
  \and
  \Infer[\hyperref[rule:extent-eta]{($\extent*$-$\eta$)}]
  {\wftypek[\cx!\apartcx{\Phi}{\bm{r}},\bm{x}!<\apartcx{\Xi}{\bm{r}}>]{\apartcx{\GG}{\bm{r}}}{A} \\
    \wftypek[\cx!\apartcx{\Phi}{\bm{r}},\bm{x}!<\apartcx{\Xi}{\bm{r}}>]{\apartcx{\GG}{\bm{r}},\bm{x},\oft{d}{A}}{B} \\
    \oftype[\cx<\Xi>]{\GG}{M}{\bsubst{A}{\bm{r}}{\bm{x}}} \\
    \oftype[\cx!\apartcx{\Phi}{\bm{r}},\bm{x}!<\apartcx{\Xi}{\bm{r}}>]{\apartcx{\GG}{\bm{r}},\bm{x},\oft{d}{A}}{N}{B}
  }
  {\eqtm[\cx<\Xi>]{\GG}{\subst{\bsubst{N}{\bm{r}}{\bm{x}}}{M}{a}}{E}{\subst{\bsubst{B}{\bm{r}}{\bm{x}}}{M}{d}} \\\\
  \text{where }E \eqdef \extent{\bm{r}}{M}{a.\subst{\bsubst{N}{\bm{0}}{\bm{x}}}{a}{d}}{a'.\subst{\bsubst{N}{\bm{1}}{\bm{x}}}{a'}{d}}{a.a'.c.\subst{\blam{\bm{x}}{N}}{\bapp{c}{\bm{x}}}{d}}}
\end{mathpar}

\subsection[Gel-types]{$\Gel$-types}
\label{sec:proof-theory:gel}

\paragraph{Type}
\begin{mathpar}
  \Infer[\hyperref[rule:gel-F]{($\Gel$-F)}]
  {\eqtypek[\cx!\apartcx{\Phi}{\bm{r}}!<\apartcx{\Xi}{\bm{r}}>]{\apartcx{\GG}{\bm{r}}}{A}{A'} \\
    \eqtypek[\cx!\apartcx{\Phi}{\bm{r}}!<\apartcx{\Xi}{\bm{r}}>]{\apartcx{\GG}{\bm{r}}}{B}{B'} \\
    \eqtypek[\cx!\apartcx{\Phi}{\bm{r}}!<\apartcx{\Xi}{\bm{r}}>]{\apartcx{\GG}{\bm{r}},\oft{a}{A},\oft{b}{B}}{R}{R'}}
  {\eqtypek[\cx<\Xi>]{\GG}{\Gel{\bm{r}}{A}{B}{a.b.R}}{\Gel{\bm{r}}{A'}{B'}{a.b.R'}}}
  \and
  \Infer[\hyperref[rule:gel-F0]{($\Gel$-F$_0$)}]
  {\wftypek[\cx<\Xi>]{\GG}{A}}
  {\eqtypek[\cx<\Xi>]{\GG}{\Gel{\bm{0}}{A}{B}{a.b.R}}{A}}
  \and
  \Infer[\hyperref[rule:gel-F1]{($\Gel$-F$_1$)}]
  {\wftypek[\cx<\Xi>]{\GG}{B}}
  {\eqtypek[\cx<\Xi>]{\GG}{\Gel{\bm{1}}{A}{B}{a.b.R}}{B}}
\end{mathpar}
\paragraph{Introduction}
\begin{mathpar}
  \Infer[\hyperref[rule:gel-I]{($\Gel$-I)}]
  {\eqtm[\cx!\apartcx{\Phi}{\bm{r}}!<\apartcx{\Xi}{\bm{r}}>]{\apartcx{\GG}{\bm{r}}}{M}{M'}{A} \\
    \eqtm[\cx!\apartcx{\Phi}{\bm{r}}!<\apartcx{\Xi}{\bm{r}}>]{\apartcx{\GG}{\bm{r}}}{N}{N'}{B} \\
    \eqtm[\cx!\apartcx{\Phi}{\bm{r}}!<\apartcx{\Xi}{\bm{r}}>]{\apartcx{\GG}{\bm{r}}}{P}{P'}{\subst{\subst{R}{M}{a}}{N}{b}}}
  {\eqtm[\cx<\Xi>]{\GG}{\gel{\bm{r}}{M}{N}{P}}{\gel{\bm{r}}{M'}{N'}{P'}}{\Gel{\bm{r}}{A}{B}{a.b.R}}}
  \and
  \Infer[\hyperref[rule:gel-I0]{($\Gel$-I$_0$)}]
  {\oftype[\cx<\Xi>]{\GG}{M}{A}}
  {\eqtm[\cx<\Xi>]{\GG}{\gel{\bm{0}}{M}{N}{P}}{M}{A}}
  \and
  \Infer[\hyperref[rule:gel-I1]{($\Gel$-I$_1$)}]
  {\oftype[\cx<\Xi>]{\GG}{N}{B}}
  {\eqtm[\cx<\Xi>]{\GG}{\gel{\bm{1}}{M}{N}{P}}{N}{B}}
\end{mathpar}

\paragraph{Elimination}

\begin{mathpar}
  \Infer[\hyperref[rule:gel-E]{($\Gel$-E)}]
  {\wftypek[\cx<\Xi>]{\GG}{A,B} \\
    \wftypek[\cx!\apartcx{\Phi}{\bm{r}}!<\apartcx{\Xi}{\bm{r}}>]{\apartcx{\GG}{\bm{r}},\oft{a}{A},\oft{b}{B}}{R} \\\\
    \eqtm[\cx!\Phi,\bm{x}!<\Xi>]{\GG,\bm{x}}{Q}{Q'}{\Gel{\bm{x}}{A}{B}{a.b.R}}}
  {\eqtm[\cx<\Xi>]{\GG}{\ungel{\bm{x}.Q}}{\ungel{\bm{x}.Q'}}{\subst{\subst{R}{\bsubst{Q}{\bm{0}}{\bm{x}}}{a}}{\bsubst{Q}{\bm{1}}{\bm{x}}}{b}}}
  \and
  \Infer[\hyperref[rule:gel-beta]{($\Gel$-$\beta$)}]
  {\oftype[\cx<\Xi>]{\GG}{M}{A} \\
    \oftype[\cx<\Xi>]{\GG}{N}{B} \\
    \oftype[\cx<\Xi>]{\GG}{P}{\subst{\subst{R}{M}{a}}{N}{b}}}
  {\eqtm[\cx<\Xi>]{\GG}{\ungel{\bm{x}.\gel{\bm{x}}{M}{N}{P}}}{P}{\subst{\subst{R}{M}{a}}{N}{b}}}
  \and
  \Infer[\rulename{$\Gel$-$\eta$}]
  {\oftype[\cx!\Phi,\bm{x}!<\Xi>]{\GG,\bm{x}}{Q}{\Gel{\bm{x}}{A}{B}{a.b.R}}}
  {\eqtm[\cx<\Xi>]{\GG}{\bsubst{Q}{\bm{r}}{\bm{x}}}{\gel{\bm{r}}{\bsubst{Q}{\bm{0}}{\bm{x}}}{\bsubst{Q}{\bm{1}}{\bm{x}}}{\ungel{\bm{x}.Q}}}{\Gel{\bm{r}}{A}{B}{a.b.R}}}
\end{mathpar}


\section{Relativity}
\label{sec:relativity}

With all the pieces in place, we can prove the long-awaited correspondence between bridges in the universe
$\UKan$ and type-valued relations. For the remainder of this section, we fix $\coftype{A,B}{\UKan}$.

\begin{notation}
  For $\coftype{R}{\prd{A}{B} \to \UKan}$, we abbreviate $\Gel{\bm{r}}{A}{B}{a.b.R\pair{a}{b}}$ as
  $\Gel{\bm{r}}{A}{B}{R}$. We will also avail ourselves of pattern-matching notation for products:
  $\lam{\pair{a}{b}}{M}$ is short for $\lam{t}{\subst{\subst{M}{\fst{t}}{a}}{\snd{t}}{b}}$.
\end{notation}

\begin{lemma}
  \label{lem:link-unlink}
  For any $\coftype{R}{\prd{A}{B} \to \UKan}$, $\coftype{M}{A}$, and $\coftype{N}{B}$, the types
  $R\pair{M}{N}$ and $\Bridge{\bm{x}.\Gel{\bm{x}}{A}{B}{R}}{M}{N}$ are equivalent.
\end{lemma}
\begin{proof}
  By \cref{rec:qequiv-to-equiv}. In the forward direction we send $p$ to
  $\blam{\bm{x}}{\gel{\bm{x}}{M}{N}{p}}$, while in the reverse direction we send $q$ to
  $\ungel{\bm{x}.\bapp{q}{\bm{x}}}$. These are mutual inverses (up to exact equality, in fact) by the $\beta$
  and $\eta$ rules for $\Gel$-types (\cref{rule:gel-beta,rule:gel-eta}).
\end{proof}

\begin{theorem}[Relativity]
  \label{thm:relativity}
  The map
  \[
    \coftype{\bridgetorel \eqdef \lam{C}{\lam{\pair{a}{b}}{\Bridge{\bm{x}.\bapp{C}{\bm{x}}}{a}{b}}}}{\Bridge{\UKan}{A}{B} \to (\prd{A}{B} \to \UKan)}
  \]
  is an equivalence.
\end{theorem}
\begin{proof}
  By \cref{rec:qequiv-to-equiv}. For our candidate inverse, we take $\reltobridge{A}{B}{-}$ defined by
  \[
    \coftype{\reltobridge{A}{B}{R} \eqdef \blam{\bm{x}}{\Gel{\bm{x}}{A}{B}{R}}}{\Bridge{\UKan}{A}{B}}.
  \]
  For the first inverse condition, we have $R : A \times B \to \UKan$ and want a path from
  $\bridgetorel{\reltobridge{A}{B}{R}}$ to $R$ in $A \times B \to \UKan$. The former is equal to
  $\lam{\pair{a}{b}}{\Bridge{\bm{x}.\Gel{\bm{x}}{A}{B}{R}}{a}{b}}$ by $\beta$-reduction for
  $\Bridge$-types. For all $a : A$ and $b : B$, the type $\Bridge{\bm{x}.\Gel{\bm{x}}{A}{B}{R}}{a}{b}$ is
  equivalent to $R\pair{a}{b}$ by \cref{lem:link-unlink}. By univalence and function extensionality for paths,
  we thus have a path between $\lam{\pair{a}{b}}{\Bridge{\bm{x}.\Gel{\bm{x}}{A}{B}{R}}{a}{b}}$ and
  $\lam{\pair{a}{b}}{R\pair{a}{b}}$, the latter of which is equal to $R$ by the $\eta$-rules for product and
  function types.

  For the second inverse condition, we have $C : \Bridge{\UKan}{A}{B}$ and want a path from
  $\reltobridge{A}{B}{\bridgetorel{C}}$ to $C$ in $\Bridge{\UKan}{A}{B}$. The former is equal to
  $\blam{\bm{x}}{\Gel{\bm{x}}{A}{B}{a.b.\Bridge{\bm{x}.\bapp{C}{\bm{x}}}{a}{b}}}$, so we want to construct a
  term of type
  \begin{gather*}
    \Path{\Bridge{\UKan}{A}{B}}{\blam{\bm{x}}{\Gel{\bm{x}}{A}{B}{a.b.\Bridge{\bm{x}.\bapp{C}{\bm{x}}}{a}{b}}}}{\blam{\bm{x}}{\bapp{C}{\bm{x}}}}.
  \end{gather*}
  By swapping the bridge and path binders (\cref{thm:path-bridge}), this is the same as constructing a term of
  type
  \begin{gather*}
    \Bridge{\bm{x}.\Path{\UKan}{\Gel{\bm{x}}{A}{B}{a.b.\Bridge{\bm{x}.\bapp{C}{\bm{x}}}{a}{b}}}{\bapp{C}{\bm{x}}}}{\dlam{\_}{A}}{\dlam{\_}{B}}.
  \end{gather*}
  By the univalence axiom and the fact that identity equivalences correspond to constant paths under
  univalence (see \citepalias[\S2.10]{hott-book}), this is equivalent to constructing a path of type
  \begin{align*}
    \Bridge{\bm{x}.\Equiv{\Gel{\bm{x}}{A}{B}{a.b.\Bridge{\bm{x}.\bapp{C}{\bm{x}}}{a}{b}}}{\bapp{C}{\bm{x}}}}{\ideq{A}}{\ideq{B}}.
  \end{align*}
  By applying \cref{cor:equiv-bridge}, we can further reduce this to constructing a term of the following type.
  \begin{align*}
    \picl{a}{A}{\picl{b}{B}{(\Bridge{\bm{x}.\Gel{\bm{x}}{A}{B}{a.b.\Bridge{\bm{x}.\bapp{C}{\bm{x}}}{a}{b}}}{a}{b} \simeq \Bridge{\bm{x}.\bapp{C}{\bm{x}}}{a}{b})}}.
  \end{align*}
  Finally, we have such a term by \cref{lem:link-unlink}.
\end{proof}

\begin{remark}
  Beyond its use in the proof of relativity (via \cref{cor:equiv-bridge}), we observe that $\extent$ is
  also necessary to derive its higher-dimensional instances. For example, consider a two-dimensional bridge in
  the universe as shown below.
  
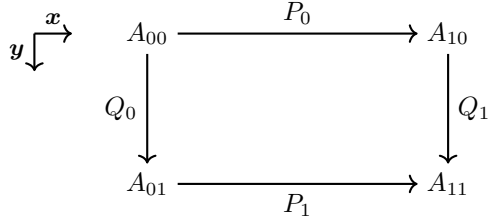
\[
    \begin{tikzpicture}
      \draw (0 , 2) [thick,->] to node [above] {\small $\bm{x}$} (0.5 , 2) ;
      \draw (0 , 2) [thick,->] to node [left] {\small $\bm{y}$} (0 , 1.5) ;
      \node () at (7, 2) { };
      \node (tl) at (1.5 , 2) {$A_{00}$} ;
      \node (tr) at (5.5 , 2) {$A_{10}$} ;
      \node (bl) at (1.5 , 0) {$A_{01}$} ;
      \node (br) at (5.5 , 0) {$A_{11}$} ;
      \draw (tl) [thick,->] to node [above] {$P_0$} (tr) ;
      \draw (tl) [thick,->] to node [left] {$Q_0$} (bl) ;
      \draw (tr) [thick,->] to node [right] {$Q_1$} (br) ;
      \draw (bl) [thick,->] to node [below] {$P_1$} (br) ;
    \end{tikzpicture}
  \]
  By relativity, the type $\Bridge{\bm{y}.\Bridge{\UKan}{Q_0}{Q_1}}{\blam{\bm{x}}{P_0}}{\blam{\bm{x}}{P_1}}$
  of fillers for this square is equivalent to the following type.
  \[
    \Bridge{\bm{y}.Q_0 \times Q_1 \to \UKan}{\lam{q_0}{\lam{q_1}{\Bridge{\bm{x}.P_0}{q_0}{q_1}}}}{\lam{q_0}{\lam{q_1}{\Bridge{\bm{x}.P_1}{q_0}{q_1}}}}
  \]
  To simplify further, we need a characterization of bridges in a function type. With
  \cref{thm:function-bridge} and a second application of relativity, we see that the type is indeed equivalent
  to the following type of two-dimensional relations on the boundary of the square.
  \[
    \begin{array}{l}
      \picl*{a_{00}}{A_{00}}{\picl*{a_{01}}{A_{01}}{\picl*{a_{10}}{A_{10}}{\picl*{a_{11}}{A_{11}}{}}}} \\
      \quad \to \Bridge{\bm{x}.P_0}{a_{00}}{a_{10}} \to \Bridge{\bm{x}.P_1}{a_{01}}{a_{11}} \\
      \quad \to \Bridge{\bm{y}.Q_1}{a_{00}}{a_{01}} \to \Bridge{\bm{y}.Q_1}{a_{10}}{a_{11}} \\
      \quad \to \UKan
    \end{array}
  \]
  We will make use of such a two-dimensional relation (in a slightly massaged form) in our proof that the
  booleans are bridge-discrete (\cref{sec:examples:bool}).
\end{remark}


\section{Bridge-discrete types}
\label{sec:bridge-discrete}

Part of the standard parametricity toolkit is the \emph{identity extension lemma} \citep{reynolds83}, which
implies in particular that the relational interpretation of a closed type is the identity relation. In our
setting, the analogous result would be to have $\Bridge{\bm{x}.A}{M}{M'} \simeq \Path{\bm{x}.A}{M}{M'}$
whenever $\bm{x}$ does not occur in $A$. However, we follow \citeauthor{bernardy15}\ in imposing no such
condition on types. The condition is of course violated by the universe $\UKan$, where bridges are relations
but paths are equivalences; as the theory stands, we could consistently impose it on types in $\UKan$ (i.e.,
small types), but it is debatable whether this is desirable. For example, \cite{nuyts17} productively use a
type $\mathsf{Size}$ which has discrete $\Path$ but codiscrete $\Bridge$ structure. Moreover, we would need to
give a computational interpretation of such an axiom in any case. We will therefore take a conservative
approach: we internally define a sub-universe $\UBDisc$ of \emph{bridge-discrete types}, show it is closed
under various type formers, and use it in place of $\UKan$ when the condition is necessary for a proof.

There is a canonical family of maps taking paths in $A$ to bridges in $A$. We will say that $A$ is
bridge-discrete when this map is an equivalence. This choice of definition ensures that bridge-discreteness is
a \emph{proposition}, that is, that any two proofs that $A$ is bridge-discrete are equal up to a
path. However, we observe below that it suffices to construct \emph{any} family of equivalences between
$\Bridge{A}{-}{-}$ and $\Path{A}{-}{-}$; indeed, it is enough to show that the former is a retract of the
latter. (This is a consequence of a standard lemma used to characterize path spaces in homotopy type theory.)

\begin{definition}
  Given $\cwftypek{A}$, $\coftype{M,N}{A}$, and $\coftype{P}{\Path{A}{M}{N}}$, we define a term
  $\coftype{\loosen{A}{P}}{\Bridge{A}{M}{N}}$ by
  $\loosen{A}{P} \eqdef \coe{z.\Bridge{A}{\dapp{P}{0}}{\dapp{P}{z}}}{0}{1}{\blam{\_}{\dapp{P}{0}}}$.
\end{definition}

\begin{lemma}
  \label{lem:loosen-refl}
  For any $\cwftypek{A}$ and $\coftype{M}{A}$, we have a term
  \[
    \coftype{\loosenrefl{A}{M}}{\Path{\Bridge{A}{M}{M}}{\loosen{A}{\dlam{\_}{M}}}{\blam{\_}{M}}}.
  \]
\end{lemma}
\begin{proof}
  Take $\loosenrefl{A}{M} \eqdef \dlam{z}{\coe{\_.\Bridge{A}{M}{M}}{z}{1}{\blam{\_}{M}}}$.
\end{proof}

\begin{definition}
  We say that $\cwftypek{A}$ is \emph{bridge-discrete} when the type
  \[
    \isBDisc{A} \eqdef \picl{a,a'}{A}{\isEquiv{\Path{A}{a}{a'}}{\Bridge{A}{a}{a'}}{\lam{p}{\loosen{A}{p}}}}
  \]
  is inhabited.
\end{definition}

\begin{lemma}
  \label{lem:isbdisc-prop}
  For any $\cwftypek{A}$, $\isBDisc{A}$ is a proposition.
\end{lemma}
\begin{proof}
  By \cref{rec:isequiv-is-prop} and the fact that a function type is a proposition when its codomain is a
  proposition \citepalias[Example 3.6.2]{hott-book}.
\end{proof}

\begin{definition}
  Given $\cwftypek{A,B}$ and $\coftype{G}{B \to A}$ we define the type
  \[
    \Retract{A}{B} \eqdef \sigmacl{f}{A \to B}{\sigmacl{g}{B \to A}{\picl{a}{A}{\Path{A}{g(f(a))}{a}}}}.
  \]
  When $\Retract{A}{B}$ is inhabited, we say that \emph{$A$ is a retract of $B$}.
\end{definition}

\begin{lemma}
  \label{lem:path-retract-to-equiv}
  Let $\coftype{A}{\UKan}$ and $\wftypek{\oft{a,a'}{A}}{R}$ be given. If there exists a family of retracts
  $\coftype{S}{\picl{a,a'}{A}{\Retract{Raa'}{\Path{A}{a}{a'}}}}$, then $\fst{Saa'}$ is an equivalence for all
  $a,a' : A$.
\end{lemma}
\begin{proof}
  This is a known result in homotopy type theory,\footnote{See \url{https://github.com/HoTT/book/issues/718}.}
  but we sketch the proof for lack of a convenient reference. The term $S$ straightforwardly gives rise to a
  family of retracts of total spaces, a term of the following type.
  \[
    \picl{a}{A}{\Retract{\sigmacl{a'}{A}{\Bridge{A}{a}{a'}}}{\sigmacl{a'}{A}{\Path{A}{a}{a'}}}}
  \]
  For each $a:A$, the type $\sigmacl{a'}{A}{\Path{A}{a}{a'}}$ is contractible by
  \cref{rec:singleton-contractibility}; any retract of a contractible type is also contractible
  \citepalias[Lemma 3.11.7]{hott-book}, so $\sigmacl{a'}{A}{\Bridge{A}{a}{a'}}$ is as well. Any function
  between contractible types is an equivalence, so in particular the function
  \[
    \oftype{\oft{a}{A}}{\lam{\pair{a'}{p}}{\pair{a'}{\loosen{A}{p}}}}{(\sigmacl{a'}{A}{\Path{A}{a}{a'}}) \to (\sigmacl{a'}{A}{\Bridge{A}{a}{a'}})}
  \]
  is an equivalence. This equivalence on total spaces implies a fiberwise equivalence: for every $a':A$, the
  function $\lam{p}{\loosen{A}{p}}$ is an equivalence \citepalias[Theorem 4.7.7]{hott-book}.
\end{proof}

\begin{corollary}
  \label{cor:retract-to-discrete}
  Let $\coftype{A}{\UKan}$. If $\Bridge{A}{a}{a'}$ is a retract of $\Path{A}{a}{a'}$ for all $a,a' : A$, then
  $A$ is bridge-discrete. In particular, if $\Bridge{A}{a}{a'}$ and $\Path{A}{a}{a'}$ are equivalent for all
  $a,a' : A$, then $A$ is bridge-discrete.
\end{corollary}

\begin{definition}
  We define the \emph{sub-universe of bridge-discrete types} by
  $\UBDisc \eqdef \sigmacl{X}{\UKan}{\isBDisc{X}}$.
\end{definition}

\begin{lemma}
  \label{lem:compound-bridge-discrete}
  $\UBDisc$ is closed under pair, function, $\Path$-, and $\Bridge$-types, in the sense that each compound
  type is bridge-discrete when its component types are bridge-discrete.
\end{lemma}
\begin{proof}
  These are straightforward corollaries of \cref{thm:sigma-bridge,thm:path-bridge,thm:function-bridge}.
\end{proof}

With a little more work, we can show that $\UBDisc$ is also relativistic, in the sense that bridges in
$\UBDisc$ correspond to $\UBDisc$-valued relations on the first components of their endpoints.

\begin{lemma}
  \label{lem:prop-bridge}
  Let $\cwftypek[\cx!\Phi,\bm{x}!]{A}$ and $\coftype[\cx!\Phi,\bm{x}!]{P}{\isProp{A}}$ be given. For any
  $\coftype{M_0}{\bsubst{A}{\bm{0}}{\bm{x}}}$ and $\coftype{M_1}{\bsubst{A}{\bm{1}}{\bm{x}}}$, the type
  $\Bridge{\bm{x}.A}{M_0}{M_1}$ is a proposition.
\end{lemma}
\begin{proof}
  Let $q_0,q_1 : \Bridge{\bm{x}.A}{M_0}{M_1}$. Then we have
  \[
    \dlam{y}{\blam{\bm{x}}{\bighcom{A}{0}{1}{\bapp{q_0}{\bm{x}}}{%
          \begin{array}{lcl}
            \arraytube{\bm{x}=\bm{0}}{z.\dapp{PM_0M_0}{z}} \\
            \arraytube{\bm{x}=\bm{1}}{z.\dapp{PM_1M_1}{z}} \\
            \arraytube{y=0}{z.\dapp{P(\bapp{q_0}{\bm{x}})(\bapp{q_0}{\bm{x}})}{z}} \\
            \arraytube{y=1}{z.\dapp{P(\bapp{q_0}{\bm{x}})(\bapp{q_1}{\bm{x}})}{z}}
          \end{array}
        }}}
  \]
  of type $\Path{\Bridge{\bm{x}.A}{M_0}{M_1}}{q_0}{q_1}$.
\end{proof}

\begin{lemma}
  \label{lem:isbdisc-bridge}
  For any $\cwftypek[\cx!\Phi,\bm{x}!]{A}$, $\coftype{D_0}{\isBDisc{\bsubst{A}{\bm{0}}{\bm{x}}}}$,
  $\coftype{D_1}{\isBDisc{\bsubst{A}{\bm{1}}{\bm{x}}}}$, there is an equivalence
  \[
    \Bridge{\bm{x}.\isBDisc{A}}{D_0}{D_1} \simeq \picl*{a_0}{\bsubst{A}{\bm{0}}{\bm{x}}}{\picl{a_1}{\bsubst{A}{\bm{1}}{\bm{x}}}{\isBDisc{\Bridge{\bm{x}.A}{a_0}{a_1}}}}.
  \]
\end{lemma}
\begin{proof}
  The left-hand type is a proposition by \cref{lem:isbdisc-prop} and \cref{lem:prop-bridge}, while the
  right-hand type is a proposition by \cref{lem:isbdisc-prop} and the fact that a function type with
  propositional codomain is propositional \citepalias[Example 3.6.2]{hott-book}. Using
  \cref{rec:qequiv-to-equiv}, it therefore suffices to construct any pair of functions between the two.

  In the forward direction, we are given $t : \Bridge{\bm{x}.\isBDisc{A}}{D_0}{D_1}$,
  $a_0 : \bsubst{A}{\bm{0}}{\bm{x}}$, $a_1 : \bsubst{A}{\bm{1}}{\bm{x}}$, and we need to show
  $\Bridge{\bm{x}.A}{a_0}{a_1}$ is bridge-discrete. By \cref{cor:retract-to-discrete}, it suffices to show
  that for any $q,q' : \Bridge{\bm{x}.A}{a_0}{a_1}$ we have an equivalence
  $\Bridge{\Bridge{\bm{x}.A}{a_0}{a_1}}{q}{q'} \simeq \Path{\Bridge{\bm{x}.A}{a_0}{a_1}}{q}{q'}$. This follows
  from the chain of equivalences
  \begin{align*}
    \Bridge{\Bridge{\bm{x}.A}{a_0}{a_1}}{q}{q'}
    &\simeq \Bridge{\bm{x}.\Bridge{A}{\bapp{q}{\bm{x}}}{\bapp{q'}{\bm{x}}}}{\blam{\_}{a_0}}{\blam{\_}{a_1}} \\
    &\simeq \Bridge{\bm{x}.\Path{A}{\bapp{q}{\bm{x}}}{\bapp{q'}{\bm{x}}}}{\dlam{\_}{a_0}}{\dlam{\_}{a_1}} \\
    &\simeq \Path{\Bridge{\bm{x}.A}{a_0}{a_1}}{q}{q'}
  \end{align*}
  where the center equivalence is given by $\bapp{t}{\bm{x}}$ and \cref{lem:loosen-refl}, while the first and
  third are simply rearranging binders.

  In the reverse direction, we are given
  $u :
  \picl*{a_0}{\bsubst{A}{\bm{0}}{\bm{x}}}{\picl{a_1}{\bsubst{A}{\bm{1}}{\bm{x}}}{\isBDisc{\Bridge{\bm{x}.A}{a_0}{a_1}}}}$.
  We can construct a term of the desired type $\Bridge{\bm{x}.\isBDisc{A}}{D_0}{D_1}$ from a term of
  type
  \[
    \Bridge{\bm{x}.\picl{a,a'}{A}{(\Path{A}{a}{a'} \simeq \Bridge{A}{a}{a'})}}{\lam{a}{\lam{a'}{\pair{\loosen{A}}{D_0aa'}}}}{\pair{\loosen{A}}{D_1aa'}}
  \]
  by \cref{cor:retract-to-discrete} (using \cref{rec:isequiv-is-prop} to adjust endpoints if necessary). By
  \cref{thm:function-bridge} applied twice and \cref{cor:equiv-bridge}, we can in turn construct such a term
  from a term of the following type.
  \[
    \begin{array}{l}
      \picl*{a_0}{\bsubst{A}{\bm{0}}{\bm{x}}}{\picl*{a_1}{\bsubst{A}{\bm{1}}{\bm{x}}}{\picl*{\overline{a}}{\Bridge{\bm{x}.A}{a_0}{a_1}}{}}} \\
      \picl*{a'_0}{\bsubst{A}{\bm{0}}{\bm{x}}}{\picl*{a'_1}{\bsubst{A}{\bm{1}}{\bm{x}}}{\picl*{\overline{a}'}{\Bridge{\bm{x}.A}{a'_0}{a'_1}}{}}} \\
      \picl*{p_0}{\Path{A_0}{a_0}{a_0'}}{\picl*{p_1}{\Path{A_1}{a_1}{a_1'}}{}} \\
      \to (\Bridge{\bm{x}.\Path{A}{\bapp{\overline{a}}{\bm{x}}}{\bapp{\overline{a}'}{\bm{x}}}}{p_0}{p_1} \simeq \Bridge{\bm{x}.\Bridge{A}{\bapp{\overline{a}}{\bm{x}}}{\bapp{\overline{a}'}{\bm{x}}}}{\loosen{A_0}{p_0}}{\loosen{A_1}{p_1}})
    \end{array}
  \]
  By \cref{rec:singleton-contractibility,lem:loosen-refl}, this is equivalent to the following
  type.
  \[
    \begin{array}{l}
      \picl*{a_0}{\bsubst{A}{\bm{0}}{\bm{x}}}{\picl*{a_1}{\bsubst{A}{\bm{1}}{\bm{x}}}{\picl*{\overline{a},\overline{a}'}{\Bridge{\bm{x}.A}{a_0}{a_1}}{}}} \\
      \to (\Bridge{\bm{x}.\Path{A}{\bapp{\overline{a}}{\bm{x}}}{\bapp{\overline{a}'}{\bm{x}}}}{\dlam{\_}{a_0}}{\dlam{\_}{a_1}} \simeq \Bridge{\bm{x}.\Bridge{A}{\bapp{\overline{a}}{\bm{x}}}{\bapp{\overline{a}'}{\bm{x}}}}{\blam{\_}{a_0}}{\blam{\_}{a_1}})
    \end{array}
  \]
  Finally, by rearranging the binders on either side of the equivalence, the codomain of this type is
  equivalent to the type
  $\Path{\Bridge{\bm{x}.A}{a_0}{a_1}}{\overline{a}}{\overline{a}'} \simeq
  \Bridge{\Bridge{\bm{x}.A}{a_0}{a_1}}{\overline{a}}{\overline{a}'}$, which is inhabited by $ua_0a_1$.
\end{proof}

\begin{theorem}
  \label{thm:bdisc-relativity}
  For any $\coftype{\tilde A, \tilde B}{\UBDisc}$, we have
  $\Bridge{\UBDisc}{\tilde A}{\tilde B} \simeq \fst{\tilde A} \times \fst{\tilde B} \to \UBDisc$.
\end{theorem}
\begin{proof}
  Abbreviating $A \eqdef \fst{\tilde A}$, $B \eqdef \fst{\tilde B}$, $D_A \eqdef \snd{\tilde A}$, and
  $D_B \eqdef \snd{\tilde B}$, we have the following chain of equivalences.
  \begin{align*}
    \Bridge{\UBDisc}{\tilde A}{\tilde B}
    &\simeq \sigmacl{Y}{\Bridge{\UKan}{A}{B}}{\Bridge{\bm{x}.\isBDisc{\bapp{Y}{\bm{x}}}}{D_A}{D_B}} &\text{(\cref{thm:sigma-bridge})} \\
    &\simeq \sigmacl{Y}{\Bridge{\UKan}{A}{B}}{\picl*{a}{A}{\picl{b}{B}{\isBDisc{\Bridge{\bm{x}.\bapp{Y}{\bm{x}}}{a}{b}}}}} &\text{(\cref{lem:isbdisc-bridge})} \\
    &\simeq \sigmacl{R}{A \times B \to \UKan}{\picl*{a}{A}{\picl{b}{B}{\isBDisc{R\pair{a}{b}}}}} &\text{(\cref{thm:relativity})} \\
    &\simeq A \times B \to \UBDisc
      &&\qedhere
  \end{align*}
  
\end{proof}

\begin{remark}
  Although we will not give a proof here, one can show by a similar argument that the sub-universe
  $\UProp \eqdef \sigmacl{X}{\UKan}{\isProp{X}}$ of propositions is relativistic (has bridges corresponding to
  $\UProp$-valued relations), and more generally that the sub-universe of $n$-types is relativistic for all
  $n \ge -2$. It is not clear to us whether there is a general theorem which has these results and
  \cref{thm:bdisc-relativity} as corollaries.
\end{remark}


\section{Examples}
\label{sec:examples}

In this final section, we give five examples of results which can be proven within parametric cubical type
theory. As a warm-up, we prove in \cref{sec:examples:id} that all terms of type $\picl{X}{\UKan}{\arr{X}{X}}$
are path-equal to the polymorphic identity function $\lam{X}{\lam{a}{a}}$. The proof is essentially that given
by \citet[Example 3.3]{bernardy15} in their system, although we obtain a slightly stronger result thanks to
function extensionality for paths. In \cref{sec:examples:leibniz}, we show that Leibniz equality is equivalent
to path equality for bridge-discrete types, a more involved result of a similar kind. In
\cref{sec:examples:bool}, we show that the type of booleans with its standard typing rules is
bridge-discrete. This proof includes a use of so-called ``iterated parametricity,'' i.e., nested
$\Gel$-types. As a corollary, we see that there is no bridge between $\true$ and $\false$ in $\bool$, which
allows us to refute the law of the excluded middle for propositions in \cref{sec:examples:lem}. Finally, in
\cref{sec:examples:hits}, we give an example of a parametricity result for functions between higher inductive
types. Apart from uses of function extensionality at the edges, we expect that the proofs in
\cref{sec:examples:leibniz,sec:examples:bool,sec:examples:lem} could be repeated in a binary version of the
system of \cite{bernardy15}; \cref{sec:examples:hits} of course requires a system with higher inductive types.

\subsection{Polymorphic identity function}
\label{sec:examples:id}

In this section, we show that any term inhabiting the type $\picl{X}{\UKan}{\arr{X}{X}}$ is equal to the
polymorphic identity function up to a path. As mentioned above, this proof is not novel, but it will serve to
introduce the methodology of internal parametricity proofs. We assume the existence of a $\unit$ type with a
single element $\triv$.

\begin{theorem}
  \label{thm:polymorphic-identity}
  Let $\coftype{F}{\picl{X}{\UKan}{\arr{X}{X}}}$ be given. Then there is a path
  from $F$ to $\lam{\_}{\lam{a}{a}}$.
\end{theorem}
\begin{proof}
  Set $\GG \eqdef (\oft{X}{\UKan},\oft{a}{X})$. Define the relation
  $\oftype{\GG}{R \eqdef \lam{\pair{a'}{\_}}{\Path{X}{a'}{a}}}{\arr{\prd{X}{\unit}}{\UKan}}$.  Abstracting
  over a bridge dimension variable $\bm{x}$, we can apply $F$ first at the $\Gel$-type given by this relation,
  then at its inhabitant $\gel{\bm{x}}{a}{\triv}{\dlam{\_}{a}}$ expressing that $R$ relates $a$ and $\triv$.
  \[
    F(\Gel{\bm{x}}{X}{\unit}{R})(\gel{\bm{x}}{a}{\triv}{\dlam{\_}{a}}) \in \Gel{\bm{x}}{X}{\unit}{R}
  \]
  This is a bridge over $\bm{x}.\Gel{\bm{x}}{X}{\unit}{R}$ whose endpoints reduce to $FXa$ (at
  $\bm{x} = \bm{0}$) and $F(\unit)(\triv)$ (at $\bm{x} = \bm{1}$). By applying $\ungel$, we turn this into a
  proof that the two endpoints are related.
  \[
    \ungel{\bm{x}.F(\Gel{\bm{x}}{X}{\unit}{R})(\gel{\bm{x}}{a}{\triv}{\dlam{\_}{a}})} \in R\pair{FXa}{F(\unit)(\triv)}
  \]
  By definition of $R$, this means we have a term of type $\Path{X}{FXa}{a}$ in context $\GG$. Function
  extensionality for paths now gives the desired result.
\end{proof}

We note that, despite the theorem above, not every term of type $\picl{X}{\UKan}{\arr{X}{X}}$ is
\emph{exactly} equal to $\lam{\_}{\lam{a}{a}}$; indeed, there are functions such as
$\lam{X}{\lam{a}{\coe{\_.X}{0}{1}{a}}}$ which are not exactly equal to the identity function. The theorem
above shows that such functions are the same as $\lam{\_}{\lam{a}{a}}$ \emph{up to a path}. The example of
$\lam{X}{\lam{a}{\coe{\_.X}{0}{1}{a}}}$ exposes a notable feature of parametric cubical type theory: we can
prove uniformity theorems up to a path even in the presence of Kan operations which evaluate by case analysis
on their type arguments.

As a simple application, we can use \cref{thm:polymorphic-identity} to prove some contractibility results for
polymorphic operators on paths.

\begin{corollary}
  The types
  \begin{enumerate}
  \item $\picl*{X}{\UKan}{\picl{a}{X}{\Path{X}{a}{a}}}$,
  \item $\picl*{X}{\UKan}{\picl{a,b}{X}{\arr{\Path{X}{a}{b}}{\Path{X}{b}{a}}}}$,
  \item $\picl*{X}{\UKan}{\picl{a,b,c}{X}{\arr{\Path{X}{a}{b}}{\arr{\Path{X}{b}{c}}{\Path{X}{a}{c}}}}}$,
  \end{enumerate}
  are all contractible.
\end{corollary}
\begin{proof}
  These three types are pairwise equivalent by way of \cref{rec:singleton-contractibility}, so it suffices to
  prove that the first is contractible. By function extensionality for paths, we have the following
  equivalence.
  \[
    (\picl{X}{\UKan}{\picl{a}{X}{\Path{X}{a}{a}}})
    \simeq \Path{\picl{X}{\UKan}{X \to X}}{\lam{\_}{\lam{a}{a}}}{\lam{\_}{\lam{a}{a}}}
  \]
  The path type of a contractible type is contractible \citepalias[Lemmas 3.11.3 and 3.11.10]{hott-book}, so
  it follows from \cref{thm:polymorphic-identity} that the right-hand type is contractible, whence the
  left-hand type is as well.
\end{proof}

\subsection{Leibniz equality}
\label{sec:examples:leibniz}

As a second example of characterizing a polymorphic function type, we show that the type $\Path{A}{a}{a'}$ of
paths between elements $a,a'$ of a bridge-discrete type $A$ is equivalent to the type
$\picl{X}{A \to \UKan}{Xa \to Xa'}$ of proofs that they are Leibniz equal. This example highlights the role of
bridge-discreteness, an assumption that will generally appear in parametricity theorems involving a fixed type
(here $A$). In short, bridge-discreteness of $A$ ensures that the relational interpretation
$\Bridge{\bm{x}.A \to \UKan}{P}{Q}$ is equivalent to the type $\picl{a}{A}{Pa \times Qa \to \UKan}$ of
pointwise relations on $P$ and $Q$, not only to the type
$\picl{a,a'}{A}{\Bridge{A}{a}{a'} \to Pa \times Qa' \to \UKan}$ as it is in the general case.

\begin{theorem}
  Let $\coftype{A}{\UKan}$ be a bridge-discrete type. Then we have a family of equivalences
  $\picl{a,a'}{A}{(\picl{X}{A \to \UKan}{Xa \to Xa'} \simeq \Path{A}{a}{a'})}$.
\end{theorem}
\begin{proof}
  \newcommand{\encode}{\mathsf{in}_A}
  \newcommand{\decode}{\mathsf{out}_A}

  As $A$ is bridge-discrete, we have a family of functions
  $T \in \picl{a,a'}{A}{\Bridge{A}{a}{a'} \to \Path{A}{a}{a'}}$ which are form an inverse to $\loosen{A}$.
  Abbreviate $E_A(a,a') \eqdef \picl{X}{A \to \UKan}{Xa \to Xa'}$. By \cref{lem:path-retract-to-equiv}, it
  suffices to show that $E_A(a,a')$ is a retract of $\Path{A}{a}{a'}$ for every $a,a' : A$. In the forward
  direction, we have the following family of functions.
  \[
    \encode \eqdef \lam{a}{\lam{a'}{\lam{f}{f(\lam{b}{\Path{A}{a}{b}})(\dlam{\_}{a})}}}
    \in \picl{a,a'}{A}{E_A(a,a') \to \Path{A}{a}{a'}}
  \]
  Conversely, the following gives us a family of candidate inverses.
  \[
    \decode \eqdef \lam{a}{\lam{a'}{\lam{q}{\lam{X}{\lam{u}{\coe{z.X(\dapp{q}{z})}{0}{1}{u}}}}}} \in \picl{a,a'}{A}{\Path{A}{a}{a'} \to E_A(a,a')}
  \]
  To show that these form a retract, we use a parametricity argument. Set
  $\GG \eqdef (\oft{a,a'}{A}, \oft{X}{A \to \UKan}, \oft{u}{Xa})$.  We define a family of relations
  $R \in \picl{b,b'}{A}{\Bridge{A}{b}{b'} \to (\Path{A}{a}{b} \times Xb') \to \UKan}$ in this context as
  follows.
  \[
    R \eqdef \lam{b}{\lam{b'}{\lam{q}{\lam{\pair{p}{v}}{\Path{x.X(\dapp{Tbb'q}{x})}{\decode aa'pXu}{v}}}}}
  \]
  We have some $L \in Raa(\blam{\_}{a})\pair{\dlam{\_}{a}}{u}$, as we have a path
  $\dlam{y}{\coe{\_.Xa}{y}{1}{u}} \in \Path{Xa}{\decode aa' (\dlam{\_}{a})Xu}{u}$ and know that $Xa$ is
  path-equal to $X(\dapp{Taa(\blam{\_}{a})}{x})$ by \cref{lem:loosen-refl}.
  
  Let $f : E_A(a,a')$, and choose $\bm{x}$ fresh. We apply $f$ to the relation $R$ and its inhabitant $L$ in
  the following way, using $\extent$ to create a bridge in the function type $A \to \UKan$.
  \[
    f
    (\lam{t}{\extent{\bm{x}}{t}{b.\Path{A}{a}{b}}{b'.Xb'}{b.b'.q.\blam{\bm{x}}{\Gel{\bm{x}}{\Path{A}{a}{b'}}{Xb}{Rbb'q}}}})
    (\gel{\bm{x}}{\dlam{\_}{a}}{u}{L})
  \]
  This is a bridge over $\bm{x}.\Gel{\bm{x}}{\Path{A}{a}{a'}}{Xa'}{Ra'a'(\blam{\_}{a'})}$ between
  $\encode aa' f$ and $fXu$. Applying $\ungel$ thus gives a term of type
  $\Path{x.X(\dapp{Ta'a'(\blam{\_}{a'})}{x})}{\decode aa'(\encode aa'f)Xu}{fXu}$. We recall that
  $\dapp{Ta'a'(\blam{\_}{a'})}{x}$ is path-equal to $a'$ by \cref{lem:loosen-refl}, so this gives a term of
  type $\Path{Xa'}{\decode aa'(\encode aa'f)Xu}{fXu}$. Finally, we obtain a term of type
  $\Path{E_A(a,a')}{\decode aa'(\encode aa'f)}{f}$ by function extensionality for paths.
\end{proof}

\subsection{Bridges of booleans}
\label{sec:examples:bool}

For this example, we assume the existence of a boolean type $\bool$ with elements $\true$,$\false$ and an
eliminator $\ifb$. For convenience, we assume an exact $\eta$-principle
$\eqtm{\oft{b}{\bool}}{b}{\ifb{\_.\bool}{b}{\true}{\false}}{\bool}$. This principle is derivable up to a path
in any case, but having an exact equality simplifies the proofs. (It holds for the boolean type defined in
\citep{chtt-iii}, but not for the ``weak boolean'' type also defined there.)

A standard argument characterizes the type $\Path{\bool}{b}{b'}$ for all $b,b'$: it is contractible when
$b \eq b' \eq \true$ or $b \eq b' \eq \false$, and it is empty when $b \eq \true$ and $b' \eq \false$ or vice
versa. In this section, we will show that $\bool$ is bridge-discrete, meaning that $\Bridge{\bool}{b}{b'}$
satisfies the same characterization. Where the calculation of $\Path{\bool}$ uses a universe, our calculation
of $\Bridge{\bool}$ will also use the fact that the universe is relativistic, i.e., that it is closed under
$\Gel$-types. This situation has an interesting parallel in the case of higher inductive types, where
univalence is used to calculate path types.

Before proceeding with the proof, we can give an intuitive argument why relativity generally enables the
characterization of bridges in inductive types. As \cref{sec:examples:id} effectively does for the $\unit$
type, we can use relativity to prove that an inductive type is equivalent to its Church encoding; in the case
of $\bool$, this is $\picl{X}{\UKan}{X \to X \to X}$. When the Church encoding is composed of types whose
bridges we already understand (such as $\UKan$ and $\to$), we can calculate \emph{its} bridge type; for
$\bool$, we see that $\Bridge{\picl{X}{\UKan}{X \to X \to X}}{G_0}{G_1}$ is equivalent to the following type.
\[
  \picl*{X_0,X_1}{\UKan}{\picl*{\overline{X}}{X_0 \times X_1 \to \UKan}{\picl*{t_0,f_0}{X_0}{\picl*{t_1,f_1}{X_1}{\picl*{\overline{t}}{\overline{X}\pair{t_0}{t_1}}{\picl{\overline{f}}{\overline{X}\pair{f_0}{f_1}}{\overline{X}\pair{G_0X_0t_0f_0}{G_1X_1t_1f_1}}}}}}}
\]
With a second parametricity argument, we can then show that \emph{this} type is an encoding for (the
appropriate index of) $\Path{\bool}$. While our proof of \cref{thm:bool-bridge} is somewhat more direct (we do
not explicitly use the Church encoding), its conceptual structure follows this outline.

\begin{definition}
  For any $\coftype[\cx!\apartcx{\Phi}{\bm{r}}!]{A}{\UKan}$, define $\PathGel{\bm{r}}{A} \eqdef \Gel{\bm{r}}{A}{A}{a.a'.\Path{A}{a}{a'}}$.
\end{definition}

\begin{theorem}
  \label{thm:bool-bridge}
  $\bool$ is bridge-discrete.
\end{theorem}
\begin{proof}
  \NewDocumentCommand{\tighten}{o g}{\mathsf{tighten}\IfValueT{#1}{_{#1}}\IfValueT{#2}{(#2)}}
  \NewDocumentCommand{\tightenrefl}{o g}{\mathsf{tighten{\mh}refl}\IfValueT{#1}{_{#1}}\IfValueT{#2}{(#2)}}

  By \cref{cor:retract-to-discrete}, it suffices to show that $\Bridge{\bool}{M_0}{M_1}$ is a retract of
  $\Path{\bool}{M_0}{M_1}$ for every $M_0,M_1$. For any $\coftype{Q}{\Bridge{\bool}{M_0}{M_1}}$, we first
  define $\coftype[\cx!\Phi,\bm{x}!]{\tighten[\bm{x}]{Q}}{\PathGel{\bm{x}}{\bool}}$ by
  \[
    \tighten[\bm{x}]{Q} \eqdef \ifb{\_.\PathGel{\bm{x}}{\bool}}{\bapp{Q}{\bm{x}}}{\gel{\bm{x}}{\true}{\true}{\dlam{\_}{\true}}}{\gel{\bm{x}}{\false}{\false}{\dlam{\_}{\false}}}.
  \]
  Observe that $\tighten[\bm{x}]{Q}$ has endpoints
  $\tighten[\bm{\Ge}]{Q} \eq \ifb{\_.\bool}{M_\Ge}{\true}{\false} \eq M_{\Ge}$ for $\Ge \in \{0,1\}$. Thus we
  may set $\coftype{\tighten{Q} \eqdef \ungel{\bm{x}.\tighten[\bm{x}]{Q}}}{\Path{\bool}{M_0}{M_1}}$.  For
  the map in the other direction, we take the previously-defined $\loosen{\bool}$. It remains to construct a
  path from $\loosen{\bool}{\tighten{Q}}$ to $Q$ for every $\coftype{Q}{\Bridge{\bool}{M_0}{M_1}}$.

  As a preliminary, we observe that $\tighten$ takes reflexive bridges to reflexive paths. For
  $\coftype{M}{\bool}$ we have a term
  $\coftype{\tightenrefl{M}}{\Path{\Path{\bool}{M}{M}}{\tighten{\blam{\_}{M}}}{\dlam{\_}{M}}}$ defined by case
  analysis:
  \[
    \tightenrefl{M} \eqdef
    \ifb{b.\Path{\Path{\bool}{b}{b}}{\tighten{\blam{\_}{b}}}{\dlam{\_}{b}}}{M}{\dlam{\_}{\dlam{\_}{\true}}}{\dlam{\_}{\dlam{\_}{\false}}}.
  \]
  Next, we define a two-dimensional relation $R$ of type
  \[
    \picl{b_{00},b_{01},b_{10},b_{11}}{\bool}{\Path{\bool}{b_{00}}{b_{10}} \times \Bridge{\bool}{b_{01}}{b_{11}} \to \Path{\bool}{b_{00}}{b_{01}} \times \Path{\bool}{b_{10}}{b_{11}} \to \UKan}
  \]
  by
  \[
    R \eqdef \lam{b_{00}}{\lam{b_{01}}{\lam{b_{10}}{\lam{b_{11}}{\lam{\pair{p}{q}}{\lam{\pair{p_0}{p_1}}{\Path{z.\Bridge{\bool}{\dapp{p_0}{z}}{\dapp{p_1}{z}}}{\loosen{\bool}{p}}{q}}}}}}}.
  \]
  Pictorially, an inhabitant of $R\etc{b_{ij}}\pair{p}{q}\pair{p_0}{p_1}$ is a filler for the
  following square.
  \[
    \begin{tikzpicture}
      \draw (0 , 2) [thick,->] to node [above] {\small $\bm{x}$} (0.5 , 2) ;
      \draw (0 , 2) [->] to node [left] {\small $z$} (0 , 1.5) ;
      \node () at (7, 2) { };
      \node (tl) at (1.5 , 2) {$b_{00}$} ;
      \node (tr) at (5.5 , 2) {$b_{10}$} ;
      \node (bl) at (1.5 , 0) {$b_{01}$} ;
      \node (br) at (5.5 , 0) {$b_{11}$} ;
      \draw (tl) [thick,->] to node [above] {$\bapp{\loosen{\bool}{p}}{\bm{x}}$} (tr) ;
      \draw (tl) [->] to node [left] {$\dapp{p_0}{z}$} (bl) ;
      \draw (tr) [->] to node [right] {$\dapp{p_1}{z}$} (br) ;
      \draw (bl) [thick,->] to node [below] {$\bapp{q}{\bm{x}}$} (br) ;
    \end{tikzpicture}
  \]
  We now convert this relation into a two-dimensional bridge in $\UKan$ using iterated $\Gel$-types. As a
  first step, we define a one-dimensional bridge of relations
  $\coftype[\cx!\Phi,\bm{x}!]{R_{\bm{x}}}{\PathGel{\bm{x}}{\bool} \times \bool \to \UKan}$ as the following
  term.
  \[
    R_{\bm{x}} \eqdef \lam{t}{\bigextent{\bm{x}}{t}{\pair{b_{00}}{b_{01}}.\Path{\bool}{b_{00}}{b_{01}}}{\pair{b_{10}}{b_{11}}.\Path{\bool}{b_{10}}{b_{11}}}{\pair{b_{00}}{b_{01}}.\pair{b_{10}}{b_{11}}.u.\\\quad \Gel{\bm{x}}{\Path{\bool}{b_{00}}{b_{01}}}{\Path{\bool}{b_{10}}{b_{11}}}{R\etc{b_{ij}}\pair{\ungel{\bm{x}.\fst{\bapp{u}{\bm{x}}}}}{\blam{\bm{x}}{\snd{\bapp{u}{\bm{x}}}}}}}}
  \]
  Now we apply a second $\Gel$, defining
  $\coftype[\cx!\Phi,\bm{x},\bm{y}!]{R_{\bm{x},\bm{y}} \eqdef
    \Gel{\bm{y}}{\PathGel{\bm{x}}{\bool}}{\bool}{R_{\bm{x}}}}{\UKan}$. Observe that this type square satisfies
  the following boundary conditions.
  \[
    R_{\bm{0},\bm{y}} \eq \PathGel{\bm{y}}{\bool} \quad
    R_{\bm{1},\bm{y}} \eq \PathGel{\bm{y}}{\bool} \quad
    R_{\bm{x},\bm{0}} \eq \PathGel{\bm{x}}{\bool} \quad
    R_{\bm{x},\bm{1}} \eq \bool
  \]
  We next define two terms
  $\coftype[\cx!\Phi,\bm{x},\bm{y}!]{T_{\bm{x},\bm{y}},F_{\bm{x},\bm{y}}}{R_{\bm{x},\bm{y}}}$ as follows.
  \begin{align*}
    T_{\bm{x},\bm{y}} &\eqdef \gel{\bm{y}}{\gel{\bm{x}}{\true}{\true}{\dlam{\_}{\true}}}{\true}{\gel{\bm{x}}{\dlam{\_}{\true}}{\dlam{\_}{\true}}{\loosenrefl{\bool}{\true}}} \\
    F_{\bm{x},\bm{y}} &\eqdef \gel{\bm{y}}{\gel{\bm{x}}{\false}{\false}{\dlam{\_}{\false}}}{\false}{\gel{\bm{x}}{\dlam{\_}{\false}}{\dlam{\_}{\false}}{\loosenrefl{\bool}{\false}}}
  \end{align*}
  These terms serve to witness the truth of
  $R(\true)(\true)(\true)(\true)\pair{\dlam{\_}{\true}}{\blam{\_}{\true}}\pair{\dlam{\_}{\true}}{\dlam{\_}{\true}}$
  and
  $R(\false)(\false)(\false)(\false)\pair{\dlam{\_}{\false}}{\blam{\_}{\false}}\pair{\dlam{\_}{\false}}{\dlam{\_}{\false}}$ respectively.

  Now, given $\coftype{Q}{\Bridge{\bool}{M_0}{M_1}}$, we define
  $\coftype[\cx!\Phi,\bm{x},\bm{y}!]{I_{\bm{x},\bm{y}} \eqdef
    \ifb{\_.R_{\bm{x},\bm{y}}}{\bapp{Q}{\bm{x}}}{T_{\bm{x},\bm{y}}}{F_{\bm{x},\bm{y}}}}{R_{\bm{x},\bm{y}}}$. By
  inspection of the definition of $\tighten[\bm{x}]$, we see that this term has the following boundary.
  \[
    \begin{tikzpicture}
      \draw (-3 , 2) [thick,->] to node [above] {\small $\bm{x}$} (-2.5 , 2) ;
      \draw (-3 , 2) [thick,->] to node [left] {\small $\bm{y}$} (-3 , 1.5) ;
      \node () at (10, 2) { };
      \node (tl) at (1.5 , 2) {$M_0$} ;
      \node (tr) at (6.5 , 2) {$M_0$} ;
      \node (bl) at (1.5 , 0) {$M_1$} ;
      \node (br) at (6.5 , 0) {$M_1$} ;
      \draw (tl) [thick,->] to node [above] {$\tighten[\bm{x}]{Q}$} (tr) ;
      \draw (tl) [thick,->] to node [left] {$\tighten[\bm{y}]{\blam{\_}{M_0}}$} (bl) ;
      \draw (tr) [thick,->] to node [right] {$\tighten[\bm{y}]{\blam{\_}{M_1}}$} (br) ;
      \draw (bl) [thick,->] to node [below] {$\bapp{Q}{\bm{x}}$} (br) ;
      \node at (4, 1) {$\ifb{\_.R_{\bm{x},\bm{y}}}{\bapp{Q}{\bm{x}}}{T_{\bm{x},\bm{y}}}{F_{\bm{x},\bm{y}}}$} ;
    \end{tikzpicture}
  \]
  Thus we have
  $\coftype[\cx!\Phi,\bm{x}!]{\ungel{\bm{y}.I_{\bm{x},\bm{y}}}}{R_{\bm{x}}\pair{\tighten[\bm{x}]{Q}}{\bapp{Q}{\bm{x}}}}$
  with $\ungel{\bm{y}.I_{\bm{\Ge},\bm{y}}} \eq \tighten{\blam{\_}{M_\Ge}}$ for $\Ge \in \{0,1\}$. In turn, we
  have
  \[
    \coftype{\ungel{\bm{x}.\ungel{\bm{y}.I_{\bm{x},\bm{y}}}}}{\Path{z.\Bridge{\bool}{\dapp{\tighten{\blam{\_}{M_0}}}{z}}{\dapp{\tighten{\blam{\_}{M_1}}}{z}}}{\loosen{\bool}{\tighten{Q}}}{Q}}.
  \]
  Finally, we can transform this term into a term of type
  $\Path{z.\Bridge{\bool}{M_0}{M_1}}{\loosen{\bool}{\tighten{Q}}}{Q}$ by rewriting along $\tightenrefl{M_0}$
  and $\tightenrefl{M_1}$ with $\coe$.
\end{proof}

Assuming the existence of an empty type $\void$, we have the following corollary.

\begin{corollary}
  There is a term of type $\Bridge{\bool}{\true}{\false} \to \void$.
\end{corollary}

\subsection{Excluding an excluded middle}
\label{sec:examples:lem}

Exploiting the lack of a bridge between $\true$ and $\false$, we can refute the law of excluded middle for
propositions as formulated in \citepalias[\S3.4]{hott-book}. For other results which follow from the
refutation of this principle, see \citep{booij17}.

\begin{lemma}
  \label{lem:u-to-discrete-constant}
  Let $\cwftypek{A}$ be bridge-discrete. Every function $f : \UKan \to A$ is constant, in the sense that there
  exists some $M \in A$ such that $\picl{X}{\UKan}{\Path{A}{fX}{M}}$ is inhabited.
\end{lemma}
\begin{proof}
  Let $f : \UKan \to A$. Set $M \eqdef fA$ (the choice of $A$ here is immaterial). For any $X : \UKan$, we
  have a bridge $\coftype{\blam{\bm{x}}f(\Gel{\bm{x}}{A}{X}{\_.\_.A})}{\Bridge{A}{fA}{fX}}$, which gives rise
  to a path of type $\Path{A}{M}{fX}$ by the assumption that $A$ is bridge-discrete.
\end{proof}

\begin{theorem}
  Define the \emph{weak law of the excluded middle} $\cwftypek{\LEMW}$ by
  \[
    \LEMW \eqdef \picl{X}{\UKan}{\sigmacl{b}{\bool}{\ifb{\UKan}{b}{\negate{X}}{\negate{\negate{X}}}}}.
  \]
  There is a term of type $\LEMW \to \void$.
\end{theorem}
\begin{proof}
  Suppose we have $f : \LEMW$. Then $\coftype{\lam{X}{\fst{\LEMW(X)}}}{\UKan \to \bool}$. By
  \cref{lem:u-to-discrete-constant,thm:bool-bridge} this function must be constant. On the other hand, it is
  easy to see that $\fst{f(\void)}$ must be $\true$ and $\fst{f(\unit)}$ must be $\false$. Thus we have a
  contradiction.
\end{proof}

\begin{corollary}
  Define the \emph{law of the excluded middle} $\cwftypek{\LEM}$ by
  \[
    \LEM \eqdef \picl{X}{\UKan}{\isProp{X} \to \sigmacl{b}{\bool}{\ifb{\UKan}{b}{X}{\negate{X}}}}.
  \]
  Then there is a term of type $\LEM \to \void$.
\end{corollary}
\begin{proof}
  $\LEM$ implies $\LEMW$, as any type of the form $\negate{A}$ is a proposition.
\end{proof}

We note that the stronger principle
$\LEM_\infty \eqdef \picl{X}{\UKan}{\sigmacl{b}{\bool}{\ifb{\UKan}{b}{X}{\negate{X}}}}$ is already refutable
in homotopy type theory using univalence \citepalias[Corollary 3.2.7]{hott-book}. As with $\LEM_\infty$ in
homotopy type theory, the refutation of $\LEM$ need not be taken as a sign that parametric cubical type theory
is philosophically anti-classical, merely as a sign that propositionality is not a sufficiently restrictive
notion of irrelevance in this setting. In a system with a propositional truncation $\trunc{-}$ which
identifies all terms up to path equality, $\LEM$ is equivalent to the principle
$\picl{X}{\UKan}{\trunc{\sigmacl{b}{\bool}{\ifb{\UKan}{b}{X}{\negate{X}}}}}$. While this principle is
refutable, replacing $\trunc{-}$ with an operator which erases computational content, such as Nuprl's
\emph{squash type} \cite[\S10.3]{nuprl}, gives a principle which is perfectly consistent with parametric
cubical type theory.

\subsection{Polymorphic functions on higher inductive types}
\label{sec:examples:hits}

\newcommand{\suspmap}[3]{\ensuremath{\mathsf{susp{\mh}map}_{#1,#2}(#3)}}
\newcommand{\suspeta}[2]{\ensuremath{\mathsf{susp{\mh}\eta}_{#1}(#2)}}
\NewDocumentCommand\get{g g g}{\ensuremath{\mathsf{get}\IfValueT{#1}{_{#1,#2,#3}}}}

As a final example, we return to the problem of characterizing polymorphic functions, this time between higher
inductive types. As our test case, we take the \emph{suspension} type constructor. Given $\cwftypek{A}$, its
suspension type $\susp{A}$ is generated by three constructors: $\north \in \susp{A}$, $\south \in \susp{A}$,
and for every $a : A$ a path $\merid{y}{a} \in \susp{A}$ with $\merid{0}{a} \eq \north$ and
$\merid{1}{a} \eq \south$ \citepalias[\S6.5]{hott-book}. We write $\suspelim$ for its eliminator. We refer to
\cite{cavallo19} and \cite{coquand18} for more complete accounts of higher inductive types in cubical type
theory.

We aim to characterize the type $\picl{X}{\UKan}{\susp{X} \to \susp{X}}$ of polymorphic endofunctions on
suspensions. Intuitively, such a function is completely determined by where it sends the poles $\north$ and
$\south$. If it sends both to $\north$ or both to $\south$, then it must be a constant function. If it sends
$\north$ to $\north$ and $\south$ to $\south$, then it must be the identity function. Finally, if it sends
$\north$ to $\south$ and $\south$ to $\north$, then it must send each meridian $\dlam{y}{\merid{y}{a}}$ to its
inverse path (defined using $\hcom$).

Before we prove the main theorem, we first prove a general lemma which pushes suspension past $\Gel$ types
given by functional relations. This is a result of a kind with the main theorem of \cref{sec:examples:bool},
though in this case we need only one direction of the equivalence.

\begin{definition}
  \label{def:gr-types}
  Given $\coftype[\cx!\apartcx{\Phi}{\bm{r}}!]{F}{A \to B}$, define $\Gr{\bm{r}}{A}{B}{F} \eqdef \Gel{\bm{r}}{A}{B}{a.b.\Path{B}{Fa}{b}}$.
\end{definition}

\begin{definition}
  Given $\cwftypek{A,B}$ and $\coftype{F}{A \to B}$, define
  \[
    \coftype{\suspmap{A}{B}{F} \eqdef \lam{t}{\suspelim{\_.\susp{B}}{t}{\north}{\south}{x.a.\merid{x}{Fa}}}}{\susp{A} \to \susp{B}}.
  \]
\end{definition}

\begin{lemma}
  Given $\cwftypek{A,B}$ and $\coftype{T}{\susp{A}}$, there is a term
  \[
    \coftype{\suspeta{A}{T}}{\Path{\susp{A}}{T}{\suspmap{A}{A}{\lam{a}{a}}(T)}}.
  \]
\end{lemma}
\begin{proof}
  Set $\suspeta{A}{T} \eqdef \suspelim{t.\Path{\susp{A}}{t}{\suspmap{A}{A}{\lam{a}{a}}(t)}}{T}{\dlam{\_}{\north}}{\dlam{\_}{\south}}{x.a.\dlam{\_}{\merid{x}{a}}}$.
\end{proof}

\begin{lemma}
  Given $\coftype{A,B}{\UKan}$ and $\coftype{F}{A \to B}$, there is a term
  \[
    \coftype{\get{A}{B}{F}}{\Bridge{\bm{x}.\susp{\Gr{\bm{x}}{A}{B}{F}} \to \Gr{\bm{x}}{\susp{A}}{\susp{B}}{\suspmap{A}{B}{F}}}{\lam{t}{t}}{\lam{u}{u}}}.
  \]
\end{lemma}
\begin{proof}
  Define
  \[
    \get{A}{B}{F} \eqdef \blam{\bm{x}}{\lam{t}{\bighcom{C}{1}{0}{\suspelim{\_.C}{t}{N}{S}{y.g.M}}{%
          \begin{array}{l}
            \tube{\bm{x}=\bm{0}}{w.\dapp{\suspeta{A}{t}}{w}} \\
            \tube{\bm{x}=\bm{1}}{w.\dapp{\suspeta{B}{t}}{w}}
          \end{array}}}}
  \]
  where
  \begin{align*}
    C &\eqdef \Gr{\bm{x}}{\susp{A}}{\susp{B}}{\suspmap{A}{B}{F}} \\
    N &\eqdef \gel{\bm{x}}{\north}{\north}{\dlam{\_}{\north}} \\
    S &\eqdef \gel{\bm{x}}{\south}{\south}{\dlam{\_}{\south}} \\
    M &\eqdef \extent{\bm{x}}{g}{a.\merid{y}{a}}{b.\merid{y}{b}}{a.b.u.\gel{\bm{x}}{\merid{y}{a}}{\merid{y}{b}}{\dlam{z}{\merid{y}{\dapp{\ungel{\bm{x}.\bapp{u}{\bm{x}}}}{z}}}}}.
        \qedhere
  \end{align*}
\end{proof}

We use the following lemma to see that for any $f : \picl{X}{\UKan}{\susp{X} \to \susp{X}}$, the image of a
meridian of $\susp{X}$ by $f$ is uniquely determined.

\begin{lemma}
  \label{lem:merid-point-contr}
  The type $\picl{X}{\UKan}{X \to \susp{X}}$ is contractible.
\end{lemma}
\begin{proof}
  It suffices to show that $\picl{X}{\UKan}{X \to \susp{X}}$ is a retract of $\susp{\unit}$ \citepalias[Lemma
  3.11.7]{hott-book}; the latter is contractible by a standard argument. We have maps in either direction as
  follows.
  \begin{align*}
    \lam{m}{m(\unit)(\triv)} &\in (\picl{X}{\UKan}{X \to \susp{X}}) \to \susp{\unit} \\
    \lam{t}{\lam{X}{\lam{a}{\suspmap{\unit}{X}{\lam{\_}{a}}(t)}}} &\in \susp{\unit} \to (\picl{X}{\UKan}{X \to \susp{X}})
  \end{align*}
  To establish that these constitute a retract, we need to show that, for every
  $m : \picl{X}{\UKan}{X \to \susp{X}}$, $X : \UKan$, and $a : X$, we have a path from
  $\suspmap{\unit}{X}{\lam{\_}{a}}(m(\unit)(\triv))$ to $mXa$. We construct such a path with a parametricity
  argument, using $\get$ to extract a path from a suspended $\Gr$-type:
  \[
    \ungel{\bm{x}.\bapp{\get{\unit}{X}{\lam{\_}{a}}}{\bm{x}}(m(\Gr{\bm{x}}{\unit}{X}{\lam{\_}{a}})(\gel{\bm{x}}{\triv}{a}{\dlam{\_}{a}}))}
  \]
  has type $\Path{\susp{X}}{\suspmap{\unit}{X}{\lam{\_}{a}}(m(\unit)(\triv))}{mXA}$.
\end{proof}

\begin{theorem}
  There is an equivalence
  $(\picl{X}{\UKan}{\susp{X} \to \susp{X}}) \simeq \bool \times \bool$.
\end{theorem}
\begin{proof}
  We will construct an equivalence
  $(\picl{X}{\UKan}{\susp{X} \to \susp{X}}) \simeq \susp{\void} \times \susp{\void}$; it is straightforward to
  check that $\susp{\void}$ is equivalent to $\bool$. As usual, we go by \cref{rec:qequiv-to-equiv}. We will
  construct an inverse to the following map.
  \begin{align*}
    F \eqdef \lam{k}{\pair{k(\void)(\north)}{k(\void)(\south)}} &\in (\picl{X}{\UKan}{\susp{X} \to \susp{X}}) \to \susp{\void} \times \susp{\void}
  \end{align*}
  Set
  $I \eqdef \lam{X}{\suspmap{\void}{X}{\lam{v}{\voidelim{\_.\susp{X}}{v}}}} \in \picl{X}{\UKan}{\susp{\void}
    \to \susp{X}}$. By \cref{lem:merid-point-contr}, we have a term
  $C \in \picl*{X}{\UKan}{\picl{n,s}{\susp{\void}}{X \to \Path{\susp{X}}{IXn}{IXs}}}$, as this type is
  equivalent to
  $\picl{n,s}{\susp{\void}}{\Path{\picl{X}{\UKan}{X \to \susp{X}}}{\lam{X}{\lam{\_}{IXn}}}{\lam{X}{\lam{\_}{IXs}}}}$
  and every path type of a contractible type is contractible \citepalias[Lemmas 3.11.3 and
  3.11.10]{hott-book}. We these in hand, we define the candidate inverse map as follows.
  \begin{align*}
    G \eqdef \lam{\pair{n}{s}}{\lam{X}{\lam{t}{\suspelim{\_.\susp{X}}{t}{IXn}{IXs}{y.a.\dapp{CXnsa}{y}}}}}
  \end{align*}
  It is straightforward to check that for any $d : \susp{\void} \times \susp{\void}$, $F(Gd)$ is connected by
  a path to $d$. For the other inverse condition, we use $\Gel$-types. Let
  $k : \picl{X}{\UKan}{\susp{X} \to \susp{X}}$ be given. We can define
  \begin{align*}
    P_\north &\eqdef \lam{X}{\ungel{\bm{x}.k(\Gr{\bm{x}}{\void}{X}{\lam{v}{\voidelim{\_.\susp{X}}{v}}})(\gel{\bm{x}}{\north}{\north}{\dlam{\_}{\north}})}}
  \end{align*}
  which has type $\picl{X}{\UKan}{\Path{\susp{X}}{G(Fk)X(\north)}{kX(\north)}}$ and an analogous term
  $P_\south$ of type $\picl{X}{\UKan}{\Path{\susp{X}}{G(Fk)X(\south)}{kX(\south)}}$. Finally, we have a term
  \begin{align*}
    P_{\merid} \in \picl*{X}{\UKan}{\picl{a}{X}{\Path{y.\Path{\susp{X}}{G(Fk)X(\merid{y}{a})}{kX(\merid{y}{a})}}{P_\north X}{P_\south X}}},
  \end{align*}
  because this type is an iterated path type of $\picl{X}{\UKan}{X \to \susp{X}}$ and therefore contractible
  by \cref{lem:merid-point-contr}. We assemble these three cases to prove the inverse condition:
  \begin{align*}
    \lam{X}{\lam{t}{\suspelim{t.\Path{\susp{X}}{G(Fk)Xt}{kXt}}{t}{P_\north X}{P_\south X}{y.a.\dapp{P_{\merid} Xa}{y}}}}
  \end{align*}
  has type $\picl*{X}{\UKan}{\picl{t}{\susp{X}}{\Path{\susp{X}}{G(Fk)Xt}{kXt}}}$.
\end{proof}

\section{Related and future work}
\label{sec:related}

\begin{figure}
  \centering
  \[\def\arraystretch{1.2}
    \begin{array}{|c|c|}
      \hline
      \text{Parametric cubical type theory} & \text{\cite{bernardy15}} \\ \hline
      \Bridge{\bm{x}.A}{a_0}{a_1} & A \ni_{\bm{x}} a \\
      \blam{\bm{x}}{a} & a \cdot {\bm{x}} \\
      \bapp{p}{\bm{x}} & (a,_{\bm{x}} p) \\
      \extent{\bm{x}}{b}{a_0.t_0}{a_1.t_1}{a_0.a_1.c.u} & \langle \lam{a}{t} ,_{\bm{x}} \lam{a}{\lam{c}{u}} \rangle (b) \\
      \Gel{\bm{x}}{A_0}{A_1}{a_0.a_1.R} & (a : A) \times_{\bm{x}} R \\
      \gel{\bm{x}}{a_0}{a_1}{c} & (a,_{\bm{x}} c) \\
      \ungel{\bm{x}.a} & a \cdot {\bm{x}} \\
      \hline
    \end{array}
  \]
  \caption{Translation dictionary for parametric type theory.}
  \label{fig:translation}
\end{figure}

\paragraph{Parametric type theory}

The concept of parametricity originates with \cite{reynolds83}, who gave a relational interpretation of
simply-typed $\lambda$-calculus with type variables in order to show that polymorphic functions treat their
type arguments parametrically. Parametricity and relational interpretations have since been used for myriad
purposes; we will keep our attention on the road to internal parametricity. \cite{plotkin93} define an
external relational logic for proving properties of terms in the polymorphic $\lambda$-calculus, with axioms
providing access to parametricity. \cite{bernardy10} observe that for a sufficiently expressive theory, such
as dependent type theory, the relational interpretation can be defined in the \emph{same} theory. However,
their interpretation function remains external. \cite{krishnaswami13} define a relational model of extensional
dependent type theory and observe that its parametricity theorems can be internalized as axioms with
computational content. However, each use of parametricity requires a new modification to the theory. Finally,
\cite{bernardy12} complete the internalization of parametricity, adding internal operators
$- \in \llbracket - \rrbracket$ and $\llbracket - \rrbracket$ which compute the relational interpretations of
types and terms respectively. Notably, these operators have computational content, so the extended theory
remains constructive. Later work substantially simplifies their original theory by using dimension variables
\citep{bernardy13,bernardy15,moulin16}. Beyond being internal, their parametricity is also higher-dimensional:
as relations are internal to the theory, they themselves have relational structure. Higher-dimensional (or
\emph{proof-relevant}) logical relations and parametricity have also been explored by \cite{benton14},
\cite{ghani15}, and \cite{sojakova18}.

Our theory is, for the most part, a direct extension of that of \citeauthor{bernardy15} However, we do make a
few changes beyond the obvious addition of cubical structure. First and most superficially, we use different
notation for its parametricity constructs in order to match the cubical constructs; we include a translation
dictionary in \cref{fig:translation}. Our category of dimension contexts has two constants $(\bm{0},\bm{1})$
where theirs has one ($0$), giving a theory of binary rather than unary parametricity. (As such, pairs of
terms in our notation correspond to single terms in theirs). The analogue of the equivalence
$\Bridge{\bm{x}.\Gel{\bm{x}}{A}{B}{R}}{M}{N} \simeq R\pair{M}{N}$, which we prove using the rules for $\Gel$
(\cref{lem:link-unlink}), is a judgmental equality called \textsc{Pair-Pred}. With this equality, the bridge
abstraction and application operators can do double duty as $\ungel$ and $\gel$, as shown in
\cref{fig:translation}. However, validating this equality apparently requires a change to the usual presheaf
semantics of type theory with dimensions, replacing sets with \emph{$I$-sets} \citep[Chapter 3,
\S5.2]{moulin16}. Although we have not presented a presheaf semantics here, translating operational
definitions to denotational definitions is straightforward enough that we are confident $I$-sets are not
necessary to model our proof theory. As \cite{nuyts17} observe, \citeauthor{bernardy15} do not address the
lack of an identity extension lemma. Our introduction of bridge-discreteness, proofs that various types are
bridge-discrete (\cref{lem:compound-bridge-discrete,thm:bool-bridge}), and observation that the
bridge-discrete sub-universe is relativistic (\cref{thm:bdisc-relativity}) go some way towards addressing this
lacuna. Our approach is maximally noncommittal: no types are made bridge-discrete by fiat, but one may always
work in the bridge-discrete fragment.

A second line of related work is the parametric type theory of \cite{nuyts17}, which builds on the work of
\citeauthor{bernardy15}\ as well as cubical type theory. Like us, they extend dependent type theory with
contexts of bridge and path dimensions. However, their work centers around around modalities which mediate
between the two contexts; in particular, types and elements are checked under different modalities. Thus, the
idea that parametrically polymorphic functions do not inspect the type they are supplied plays a central role,
though they decouple the type/element and parametric/continuous distinctions to some extent. In contrast, our
work (like that of \citeauthor{bernardy15}) focuses to the relational interpretation aspect of
parametricity. Although we use the same ``bridge'' and ``path'' terminology, their category of contexts
differs substantially from ours: bridge variables can be substituted for path variables, and both bridge and
path variables are structural. The former means that the equivalent of $\loosen$ exists on the level of the
base category. This map gives rise to a parametric modality, among others; a function is parametric in its
input when it takes bridges to paths. We can simulate non-dependent parametric functions in our system given a
propositional truncation $\trunc{-}$: a function is parametric when its image \citepalias[Definition
7.6.3]{hott-book} is bridge-discrete. (Simulating dependent functions is also possible, but more involved.)
\[
  \isParametric{A}{B}{F} \eqdef \isBDisc{\sigmacl{b}{B}{\trunc{\Fiber{A}{B}{F}{b}}}}
\]
Their notion of path is also distinct from ours, although they play a similar role. Paths do not support any
kind of Kan operations in general; they are used by way of a \emph{path degeneracy} axiom that turns
homogeneous paths (of the form $\Path{\_.A}{M_0}{M_1}$, where in general $A$ depends on a \emph{bridge}
variable) into elements of the Martin-L\"of identity type $\Id{A}{M_0}{M_1}$. Due to this and other axioms
(such as function extensionality for $\Id$), the theory is not computational. Bridges in the universe can be
defined using one of two operators, $\mathsf{Glue}$ and $\mathsf{Weld}$. The former originates with
\cite{cchm} and is a more general form of $\V$; the latter is its dual. The issues sketched in \cref{sec:gel}
with using a $\V$-like operator for bridges are sidestepped by checking the relation argument under a
modality. Unfortunately, doing so precludes higher-dimensional parametricity. To rectify this, \cite{nuyts18}
introduce a system with an infinite ladder of increasingly permissive relations, paths and bridges being the
first two rungs.

As part of an application, \citeauthor{nuyts17}\ introduce a type $\mathsf{Size}$ which has natural numbers
for elements but codiscrete bridge structure. This raises the question of inductive types with bridge
constructors: higher inductive types for the bridge direction. Unfortunately, here the use of substructural
dimensions is an obstacle, a fact which originally motivated the use of structural dimensions in cubical type
theory. The problem may be seen by considering a type $\bint$ generated by points $\zero$, $\one$ and a bridge
$\seg \in \Bridge{\bint}{\zero}{\one}$. For any type $A$, we would expect an equivalence fitting into the
following diagram.
\[
  \begin{tikzcd}[row sep=4em,column sep=5em]
    (\bint \to A) \ar{d}[left,font=\normalsize]{\lam{b}{\pair{b(\zero)}{b(\one)}}} \ar{r}{\simeq} & \sigmacl{a_0,a_1}{A}{\Bridge{A}{a_0}{a_1}} \ar{d}[font=\normalsize]{\lam{\pair{a_0,a_1}{q}}{\pair{a_0}{a_1}}} \\
    A \times A \ar[equals]{r} & A \times A
  \end{tikzcd}
\]
Observe that on the left hand side, we have a diagonal map
$\delta \eqdef \lam{b}{\lam{i}{bii}} \in (\bint \to \bint \to A) \to (\bint \to A)$ such that
$\delta b (\zero) \eq b(\zero)(\zero)$ and $\delta b (\one) \eq b(\one)(\one)$. But no such map exists on the
right hand side, precisely because bridge dimensions are substructural. In a substructural cubical type theory
like that of \cite{bch}, there is still some hope: although there is no diagonal map in the base category, one
can define a diagonal map on the level of paths using the Kan operations. For bridges, however, this is
impossible. It remains to be seen whether some restricted class of ``bridge inductive types'' can be given a
usable proof theory.

\paragraph{Cubical type theory}

The cubical side of the theory, as well as the operational presentation, is adopted from \cite{angiuli18}
practically without change. \cite{bch} and \cite{cchm} have also developed cubical type theories, which differ
in choice of cube category and formulation of the Kan operations. We believe that the parametric additions we
make to cartesian cubical type theory could easily be replayed on top of any of these theories: one simply
needs to add (a) a substructural context of bridge dimensions and (b) $\bm{r} = \bm{\Ge}$ equations in the
language of composition constraints (i.e., generating cofibrations). The latter addition does not interfere
with the definition of the $\forall x$ operator on constraints used by \citeauthor{cchm}:
$\forall x.({\bm{r}} = {\bm{\Ge}})$ can be defined as ${\bm{r}} = {\bm{\Ge}}$. (Note that a $\forall \bm{x}$
operator is also necessary in order to define $\hcom$ in $\Gel$-types.)

The bridge side of our theory is similar to the cubical type theory of \cite{bch}, mostly insofar as that
theory is (not coincidentally) similar to \citeauthor{bernardy15}'s parametric type theory. As there is no
composition or coercion along bridges, definitions are generally simpler in our setting than in theirs,
especially where the universe is concerned.

\paragraph{Directed type theory}

Parametric cubical type theory bears a close resemblance to the directed type theory of \cite{riehl17}. Like
our theory, it has two directions of higher structure. One direction is given by identity types with the
axioms of homotopy type theory; this corresponds to our path direction. The other is given by dimension
variables (and a language of inequality constraints) and corresponds to our bridge direction. The intended
model of this theory is the category of Reedy fibrant presheaves over $\Delta \times \Delta$, which parallels
our anticipated model in suitably fibrant presheaves over
$\mathbb{C}_{(\mathrm{we},\cdot)} \times \mathbb{C}_{(\mathrm{wec},\cdot)}$ (in the terminology of
\cite{buchholtz17}). In the case of directed type theory, the intent is to carve out the sub-universe of
\emph{Segal types}, those that support a directed composition operation in the ``bridge'' direction. The
question of relativity arises under the name of \emph{directed univalence}, but the most naive formulation
appears to be false \citep{riehl18}.

\paragraph{Future work}

As mentioned in the introduction, one motivation for this work is to prove coherence theorems for functions on
higher inductive types. From \cref{sec:examples:hits}, we see that the behavior of a polymorphic function
$\picl{X}{\UKan}{\susp{X} \to \susp{X}}$ can be analyzed by checking its behavior on zero-dimensional
constructors. As such, one can avoid the higher-dimensional constructions that are necessary to prove
properties of a function $\susp{A} \to \susp{A}$ at a particular $A$. A more complicated case of interest is
that of the \emph{smash product}, a certain binary operator on pointed types (elements of
$\UPtd \eqdef \sigmacl{X}{\UKan}{X}$). To show that the smash product is commutative and associative is
difficult, to show that the commutator and associator satisfy coherence laws even more so
\citep{brunerie18}. However, once the commutator and associator have been constructed, the other theorems can
be posed in terms of characterizing polymorphic pointed functions between smash products, specifically terms
of type $\picl{X_1,\ldots,X_n}{\UPtd}{\bigwedge_{i \le n} X_i \pto \bigwedge_{i \le n} X_i}$ for various
$n$. We conjecture that such terms can be characterized in parametric cubical type theory by their behavior on
zero-dimensional constructors, and that this can be proven for all $n$ uniformly.

Of course, we would like to have such results not only for parametric cubical type theory but for ordinary
cubical type theory. Not every theorem of parametric cubical type theory is true in cubical type theory, but
it is reasonable to suspect that some class of statements is transferable. To our knowledge, the only work in
this vein is the interpretation in used in \citep{bernardy12} to prove strong normalization for their
parametric type theory, which is also discussed in \cite[Chapter 1, \S3.6]{moulin16}.


\appendix

\section{$\D$-relation interface}
\label{app:lemmas}

\begin{lemma}[Introduction]
  \label{lem:introduction}
  Let $\Ga$ be a value $\D$-relation. If for every $\tds{\D'}{\psi}{\D}$, either
  $\Ga_\psi(\td{M}{\psi},\td{M'}{\psi})$ or $\Tm{\Ga}_\psi(\td{M}{\psi},\td{M'}{\psi})$, then
  $\Tm{\Ga}(M,M')$.
\end{lemma}
\begin{proof}
  Let $\tds{\D_1}{\psi_1}{\D}$ and $\tds{\D_2}{\psi_2}{\D_1}$ be given. We
  divide into three cases.
  \begin{description}
  \item{(aa)} $\Ga_{\psi_1}(\td{M}{\psi_1},\td{M'}{\psi_1})$ and $\Ga_{\psi_1\psi_2}(\td{M}{\psi_1\psi_2},\td{M'}{\psi_1\psi_2})$.

    Then we have
    \[
      \begin{array}{ccc}
        \td{M}{\psi_1} \evals \td{M}{\psi_1} & \td{M}{\psi_1\psi_2} \evals \td{M}{\psi_1\psi_2} & \td{M}{\psi_1\psi_2} \evals \td{M}{\psi_1\psi_2} \\
        \td{M'}{\psi_1} \evals \td{M'}{\psi_1} & \td{M'}{\psi_1\psi_2} \evals \td{M'}{\psi_1\psi_2} & \td{M'}{\psi_1\psi_2} \evals \td{M'}{\psi_1\psi_2}
      \end{array}
    \]
    with $\Ga_{\psi_1\psi_2}(\td{M}{\psi_1\psi_2},\td{M'}{\psi_1\psi_2})$.
  \item{(ab)} $\Ga_{\psi_1}(\td{M}{\psi_1},\td{M'}{\psi_1})$ and
    $\Tm{\Ga}_{\psi_1\psi_2}(\td{M}{\psi_1\psi_2},\td{M'}{\psi_1\psi_2})$.

    By $\Tm{\Ga}_{\psi_1\psi_2}(\td{M}{\psi_1\psi_2},\td{M'}{\psi_1\psi_2})$, we have
    $\td{M}{\psi_1\psi_2} \evals M_{12}$ and $\td{M'}{\psi_1\psi_2} \evals M_{12}'$ with
    $\Ga_{\psi_1\psi_2}(M_{12},M_{12}')$. Thus
    \[
      \begin{array}{ccc}
        \td{M}{\psi_1} \evals \td{M}{\psi_1} & \td{M}{\psi_1\psi_2} \evals M_{12} & \td{M}{\psi_1\psi_2} \evals M_{12} \\
        \td{M'}{\psi_1} \evals \td{M'}{\psi_1} & \td{M'}{\psi_1\psi_2} \evals M_{12}' & \td{M'}{\psi_1\psi_2} \evals M_{12}'
      \end{array}
    \]

    with $\Tm{\Ga}_{\psi_1\psi_2}(M_{12},M_{12}')$.
  \item{(b$*$)} $\Tm{\Ga}_{\psi_1}(\td{M}{\psi_1},\td{M'}{\psi_1})$.

    By $\Tm{\Ga}_{\psi_1}(\td{M}{\psi_1},\td{M'}{\psi_1})$, we have
    \[
      \begin{array}{ccc}
        \td{M}{\psi_1} \evals M_1 & \td{M_1}{\psi_2} \evals M_2 & \td{M}{\psi_1\psi_2} \evals M_{12} \\
        \td{M'}{\psi_1} \evals M'_1 & \td{M'_1}{\psi_2} \evals M'_2 & \td{M'}{\psi_1\psi_2} \evals M'_{12}
      \end{array}
    \]
    with $\Ga_{\psi_1\psi_2}(V,V')$ for all $V \in \{M_2,M_{12}\}$ and $V' \in \{M_2',M_{12}'\}$. \qedhere
  \end{description}
\end{proof}

\begin{lemma}[Coherent expansion]
  \label{lem:expansion}
  Let $\Ga$ be a value $\D$-PER and let $\tmj[\D]{M,M'}$. If for every $\tds{D'}{\psi}{\D}$, there exists
  $M''$ such that $\td{M}{\psi} \msteps M''$ and $\Tm{\Ga}_\psi(M'',\td{M'}{\psi})$, then $\Tm{\Ga}(M,M')$.
\end{lemma}
\begin{proof}
  Let $\tds{\D_1}{\psi_1}{\D}$ and $\tds{\D_2}{\psi_2}{\D_1}$ be given. By assumption, there exists $M_1''$
  such that $\td{M}{\psi_1} \msteps M_1''$ and $\Tm{\Ga}_{\psi_1}(M_1'',\td{M'}{\psi_1})$.  By
  $\Tm{\Ga}_{\psi_1}(M_1'',\td{M'}{\psi_1})$ applied at the substitutions $\id$ and $\psi_2$, we see that
  \[
    \begin{array}{ccc}
      M_1'' \evals M_1 & \td{M_1}{\psi_2} \evals M_2 & \td{M_1''}{\psi_2} \evals M_{12} \\
      \td{M'}{\psi_1} \evals M'_1 & \td{M'_1}{\psi_2} \evals M'_2 & \td{M'}{\psi_1\psi_2} \evals M'_{12}
    \end{array}
  \]
  with $\Ga_{\psi_1\psi_2}(V,V')$ for $V \in \{M_2,M_{12}\}$ and $V' \in \{M_2',M_{12}'\}$. Likewise, we have
  some $M_{12}''$ such that $\td{M}{\psi_1\psi_2} \msteps M_{12}''$ and
  $\Tm{\Ga}_{\psi_1\psi_2}(M_{12}'',\td{M'}{\psi_1\psi_2})$. By the latter, we have $M_{12}'' \evals N_{12}$
  and $\td{M'}{\psi_1\psi_2} \evals N_{12}'$ with $\Ga_{\psi_1\psi_2}(N_{12},N_{12}')$. Note that
  $M_{12}' = N_{12}'$ by determinism of the operational semantics. As $\Ga$ is a PER, we thus have
  $\Ga_{\psi_1\psi_2}(N_{12},M'_2)$.

  Combining this data, we have
  \[
    \begin{array}{ccc}
      \td{M}{\psi_1} \evals M_1 & \td{M_1}{\psi_2} \evals M_2 & \td{M}{\psi_1\psi_2} \evals N_{12} \\
      \td{M'}{\psi_1} \evals M'_1 & \td{M'_1}{\psi_2} \evals M'_2 & \td{M'}{\psi_1\psi_2} \evals M'_{12}
    \end{array}
  \]
  with $\Ga_{\psi_1\psi_2}(V,V')$ for all $V \in \{M_2,N_{12}\}$ and $V' \in \{M_2',M_{12}'\}$.
\end{proof}

\begin{lemma}[Evaluation]
  \label{lem:evaluation}
  Let $\Ga$ be a value-coherent $\D$-PER and let $\tmj[\D]{M,M'}$ with $\Tm{\Ga}(M,M')$. Then $M \evals V$ and
  $M' \evals V'$ with $\Tm{\Ga}(Q,Q')$ holds for all $Q \in \{M,V\}$ and $Q' \in \{M',V'\}$.
\end{lemma}
\begin{proof}
  By $\Tm{\Ga}(M,M')$, we have $M \evals V$ and $M' \evals V'$ with $\Ga(V,V')$, which implies
  $\Tm{\Ga}(V,V')$ by value-coherence of $\Ga$. We show $\Tm{\Ga}(M,V')$; the proof of $\Tm{\Ga}(V,M')$ is
  symmetric. Let $\tds{\D_1}{\psi_1}{\D}$ and $\tds{\D_2}{\psi_2}{\D_1}$. By
  $\Tm{\Ga}(M,M')$ applied at the substitutions $\id$ and $\psi_1$, we have
  \[
    \begin{array}{ccc}
      M \evals V & \td{V}{\psi_1} \evals V_1 & \td{M}{\psi_1} \evals M_1 \\
      M' \evals V' & \td{V'}{\psi_1} \evals V'_1 & \td{M'}{\psi_1} \evals M'_1
    \end{array}
  \]
  with, in particular, $\Ga_{\psi_1}(M_1,V_1')$ and $\Ga_{\psi_1}(V_1,M_1')$. By value-coherence, we then have
  $\Tm{\Ga}_{\psi_1}(M_1',V_1')$, which implies that $\td{M_1'}{\psi_2} \evals M_2$ and
  $\td{V_1'}{\psi_2} \evals V_2$ with $\Ga_{\psi_1\psi_2}(M_2,V_2')$; symmetrically, we have
  $\td{V_1}{\psi_2} \evals V_2$ and $\td{M_1'}{\psi_2} \evals M_2'$ with $\Ga_{\psi_1\psi_2}(V_2,M_2')$. We
  now apply $\Tm{\Ga}(M,M')$ at the substitutions $\id$ and $\psi_1\psi_2$ to obtain
  \[
    \begin{array}{ccc}
      M \evals V & \td{V}{\psi_1\psi_2} \evals V_{12} & \td{M}{\psi_1\psi_2} \evals M_{12} \\
      M' \evals V' & \td{V'}{\psi_1\psi_2} \evals V'_{12} & \td{M'}{\psi_1\psi_2} \evals M'_{12}
    \end{array}
  \]
  with, in particular, $\Ga_{\psi_1\psi_2}(M_{12},V'_{12})$. Finally, applying $\Tm{\Ga}(M,M')$ and
  $\Tm{\Ga}(V,V')$ at $\psi_1$ and $\psi_2$ gives us
  $\Ga_{\psi_1\psi_2}(M_{12},M_2')$, $\Ga_{\psi\psi_2}(M_2,M_{12}')$, $\Ga_{\psi_1\psi_2}(V_{12},V_2')$ and
  $\Ga_{\psi\psi_2}(V_2,V_{12}')$. We now have
  \[
    \begin{array}{cccc}
      \Ga_{\psi_1\psi_2}(M_2,V_2') & \Ga_{\psi_1\psi_2}(M_{12},V'_{12}) & \Ga_{\psi_1\psi_2}(M_{12},M_2') & \Ga_{\psi_1\psi_2}(V_{12},V_2') \\
      \Ga_{\psi_1\psi_2}(V_2,M_2') & \Ga_{\psi_1\psi_2}(V_{12},M'_{12}) & \Ga_{\psi_1\psi_2}(M_2,M_{12}') & \Ga_{\psi_1\psi_2}(V_2,V_{12}')
    \end{array}
  \]
  Combining these via transitivity of $\Ga$ gives the desired result: $\Ga_{\psi_1\psi_2}(W,W')$ for all
  $W \in \{M_2,M_{12}\}$ and $W' \in \{V_2',V_{12}'\}$.
\end{proof}

\begin{lemma}[Formation]
  \label{lem:formation}
  Let $\tau$ be a bridge-path type system, let $\tmj[\D]{A,A'}$, and let $\Ga$ be a value
  $\D$-relation. If for every $\tds{\D'}{\psi}{\D}$, either
  $\PTy{\tau}(\D',\td{A}{\psi},\td{A'}{\psi},\td{\Ga}{\psi})$ holds or
  $\tau(\D',\td{A}{\psi},\td{A'}{\psi},\Ga_\psi)$ holds, then $\PTy{\tau}(\D,A,A',\Ga)$.
\end{lemma}
\begin{proof}
  A straightforward adaptation of the proof of \cref{lem:introduction}.
\end{proof}

\begin{lemma}[Coherent type expansion]
  \label{lem:type-expansion}
  Let $\tau$ be a bridge-path type system, let $\tmj[\D]{A,A'}$, and let $\Ga$ be a value
  $\D$-relation. If for all $\tds{\D'}{\psi}{\D}$, there exists $A''$ such that $\td{A}{\psi} \msteps A''$
  and $\PTy{\tau}(\D',A'',\td{A'}{\psi},\td{\Ga}{\psi})$, then $\PTy{\tau}(\D,A,A',\Ga)$.
\end{lemma}
\begin{proof}
  A straightforward adaptation of the proof of \cref{lem:expansion}.
\end{proof}

\begin{lemma}[Elimination]
  \label{lem:elimination}
  Let $\evalcx{\D_0}{\C,\C',\T}{\D}$, $\bridgesj[\D]{\bm{\rho}}$, and let $\Ga$ be a value-coherent
  $(\apartcx{\D}{\bm{\rho}}\D_0)$-PER. Suppose that for every $\tds{\D'}{\psi}{\D}$ with $\D'$ disjoint from
  $\D_0$, we have
  \begin{enumerate}
  \item $\ceqtypep[\D']{\plug{\td{\T}{\psi}}{M}}{\plug{\td{\T}{\psi}}{M'}}$ for all
    $\Tm{\Ga}_{\apartcx{\psi}{\bm{\rho}} \times \D_0}(M,M')$,
  \item $\td{\C}{\psi}$,$\td{\C'}{\psi}$ are eager and
    $\ceqtm[\D']{\plug{\td{\C}{\psi}}{V}}{\plug{\td{\C'}{\psi}}{V'}}{\plug{\td{\T}{\psi}}{V}}$ for all
    $\Ga_{\apartcx{\psi}{\bm{\rho}} \times \D_0}(V,V')$.
  \end{enumerate}
  Then $\ceqtm[\D]{\plug{\C}{M}}{\plug{\C'}{M'}}{\plug{\T}{M}}$ for every $\Tm{\Ga}(M,M')$.
\end{lemma}
\begin{proof}
  Let $\Tm{\Ga}(M,M')$, $\tds{\D_1}{\psi_1}{\D}$ and $\tds{\D_2}{\psi_2}{\D_1}$ be given; by $\Ga$-varying
  $\D_0,M,M'$, we may assume that $\D_0$ is disjoint from $\D_1$ and $\D_2$.  By applying $\Tm{\Ga}(M,M')$ at
  $\apartcx{\psi_1}{\bm{\rho}} \times \D_0$ and $\apartcx{\psi_2}{\td{\bm{\rho}}{\psi_1}} \times \D_0$, we
  have
  \[
    \begin{array}{ccc}
      \td{M}{\psi_1} \evals M_1 & \td{M_1}{\psi_2} \evals M_2 & \td{M}{\psi_1\psi_2} \evals M_{12} \\
      \td{M'}{\psi_1} \evals M'_1 & \td{M'_1}{\psi_2} \evals M'_2 & \td{M'}{\psi_1\psi_2} \evals M'_{12}
    \end{array}
  \]
  where $\Ga_{\apartcx{(\psi_1\psi_2)}{\bm{\rho}} \times \D_0}(V,V')$ for $V \in \{M_2,M_{12}\}$ and
  $V' \in \{M'_2,M'_{12}\}$; by applying it at $\apartcx{\psi_1}{\bm{\rho}} \times \D_0$ and $\id$, we also
  know $\Ga_{\apartcx{\psi_1}{\bm{\rho}} \times \D_0}(M_1,M_1')$.  Thus
  $\ceqtm[\D_1]{\plug{\td{\C}{\psi_1}}{M_1}}{\plug{\td{\C'}{\psi_1}}{M_1'}}{\td{\plug{\T}{M}}{\psi_1}}$;
  applying this at $\id$ and $\psi_2$ gives
  \[
    \begin{array}{ccc}
      \plug{\td{\C}{\psi_1}}{M_1} \evals N_1 & \td{N_1}{\psi_2} \evals N_2 & \td{\plug{\td{\C}{\psi_1}}{M_1}}{\psi_2} \evals N_{12} \\
      \plug{\td{\C'}{\psi_1}}{M'_1} \evals N'_1 & \td{N'_1}{\psi_2} \evals N'_2 & \td{\plug{\td{\C}{\psi_1}}{M'_1}}{\psi_2} \evals N'_{12} \\
    \end{array}
  \]
  with $\vper{\plug{\T}{M}}_{\psi_1\psi_2}(V,V')$ for $V \in \{N_2,N_{12}\}$ and $V' \in \{N'_2,N'_{12}\}$.
  For any $t,t' \in \{2,{12}\}$ we have
  $\ceqtm[\D_2]{\plug{\td{\C}{\psi_1\psi_2}}{M_t}}{\plug{\td{\C'}{\psi_1\psi_2}}{M'_{t'}}}{\td{\plug{\T}{M}}{\psi_1\psi_2}}$,
  which gives
  \[
    \begin{array}{cc}
      \plug{\td{\C}{\psi_1\psi_2}}{M_t} \evals P_t  &
      \plug{\td{\C'}{\psi_1\psi_2}}{M'_{t'}} \evals P'_{t'}
    \end{array}
  \]
  with $\vper{\plug{\T}{M}}_{\psi_1\psi_2}(P_t,P'_{t'})$.  Using that
  $\td{\C}{\psi_1},\td{\C'}{\psi_1},\td{\C}{\psi_1\psi_2},\td{\C'}{\psi_1\psi_2}$ are all eager, we assemble
  the above to get
  \[
    \begin{array}{ccc}
      \td{\plug{\C}{M}}{\psi_1} \evals N_1 & \td{N_1}{\psi_2} \evals N_2 & \td{\plug{\C}{M}}{\psi_1\psi_2} \evals P_{12} \\
      \td{\plug{\C'}{M'}}{\psi_1} \evals N'_1 & \td{N'_1}{\psi_2} \evals N'_2 & \td{\plug{\C'}{M'}}{\psi_1\psi_2} \evals P'_{12}.
    \end{array}
  \]
  Now, observe that we have $\plug{\td{\C}{\psi_1\psi_2}}{M_1\psi_2} \evals N_{12}$ as well as
  $\plug{\td{\C}{\psi_1\psi_2}}{M_2} \evals P_2$; as $\td{M_1}{\psi_2} \evals M_{12}$ and
  $\td{\C}{\psi_1\psi_2}$ is eager, this means $N_{12} = P_2$ by determinism. Likewise $N'_{12} = P'_2$. Using
  transitivity of $\vper{\plug{\T}{M}}$, we can therefore conclude that
  $\vper{\plug{\T}{M}}_{\psi_1\psi_2}(V,V')$ for $V \in \{N_2,P_{12}\}$ and $V' \in \{N_2',P'_{12}\}$.

  Note that throughout this proof, we have implicitly identified indices of $\T$; for example, we use the fact
  that $\td{\vper{\plug{\T}{M}}}{\psi_1}$ and $\vper{\plug{\td{\T}{\psi_1}}{M_1}}$ are equal. Here we use the
  assumption that $\Ga$ is value-coherent, which gives $\Tm{\Ga}_{\apartcx{\psi_1}{\bm{\rho}} \times \D_0}(\td{M}{\psi_1},M_1)$ by
  \cref{lem:evaluation}.
\end{proof}


\section{Fixed-point construction}
\label{app:fixed-point}

To construct concrete examples of bridge-path type systems, we use a minor variation on the fixed-point
construction introduced by \cite{chtt-iii} following \cite{allen87}. We sketch here the additions necessary to
accommodate bridge variables and the $\Bridge$ and $\Gel$ types. As mentioned above, we satisfy ourselves with
a single universe, but it is not significantly more difficult to construct a type system with an infinite
hierarchy.

\begin{definition}
  Define an operator $\F_{\mathrm{B}}$ on candidate bridge-path type systems as follows.
  \[
    \begin{array}{rcl}
      \F_{\mathrm{B}}(\tau) &\eqdef& \{(\cx*,\Bridge{\bm{x}.A}{M_0}{M_1},\Bridge{\bm{x}.A}{M_0}{M_1},\BridgeR[\tau]{\bm{x}.A}{M_0}{M_1}_\id) \\
               && \;{ } \mid \relcts*{\tau}{\ceqtypek[\cx!\Phi,\bm{x}!]{A}{A'}} \land (\forall \Ge)\;\relcts*{\tau}{\ceqtm{M_\Ge}{M_\Ge'}{\bsubst{A}{\bm{\Ge}}{\bm{x}}}} \} \\[0.3em]
               &\cup& \{(\cx*!\Phi,\bm{x}!,\Gel{\bm{x}}{A}{B}{a.b.R},\Gel{\bm{x}}{A'}{B'}{a.b.R'},\GelR[\tau]{\bm{x}}{A}{B}{a.b.R}_{\id}) \\
               && \;{ } \mid \relcts*{\tau}{\ceqtypek{A}{A'}} \land \relcts*{\tau}{\ceqtypek{B}{B'}} \land { }\\
               && \;{ } \mathrel{\phantom{\mid}} \relcts*{\tau}{\eqtypek{\oft{a}{A},\oft{b}{B}}{R}{R'}} \} \\[0.3em]
    \end{array}
  \]
  We assume the existence of an analogous operator $\F_{\mathrm{C}}$ on candidate bridge-path type systems
  with clauses for each type former of cubical type theory: pair, function, $\Path$-, $\V$-, and $\fcom$-types
  \citep[\S3.1]{chtt-iii}. These can all be defined uniformly in the bridge context $\Phi$.
\end{definition}

\begin{definition}
  Define the candidate bridge-path type system
  $\tau_0 \eqdef \lfp{\tau.\F_{\mathrm{B}}(\tau) \cup \F_{\mathrm{C}}(\tau)}$, the least fixed-point of
  $\tau \mapsto \F_{\mathrm{B}}(\tau) \cup F_C(\tau)$ in the lattice of candidate bridge-path type systems
  ordered by inclusion (regarded as quaternary relations).
\end{definition}

\begin{definition}
  Define the candidate bridge-path type system $\tau_1$ as follows.
  \[
    \tau_1 \eqdef \lfp{\tau.\F_{\mathrm{B}}(\tau) \cup F_C(\tau) \cup \{(\cx*,\UKan,\UKan,\phi) \mid \phi(A,A') \iff \relcts*{\tau_0}{\ceqtypek{A}{A'}}\}}
  \]
\end{definition}

\begin{proposition}
  $\tau_0$ and $\tau_1$ are bridge-path type systems.
\end{proposition}
\begin{proof}
  See \cite[Theorem 16]{chtt-iii}.
\end{proof}

\begin{proposition}
  $\tau_0$ and $\tau_1$ is closed under $\Bridge$- and $\Gel$-types and the constructs of cubical type
  theory. Moreover, $\tau_1$ contains a univalent universe $\UKan$ which is closed under these same type
  formers.
\end{proposition}


\section{Maps of smash products}
\label{sec:smash}

\newcommand{\concinv}[5]{\mathsf{conc{\mh}inv}_{#1}^{#2,#3}(#4,#5)}
\newcommand{\cnxor}[4]{\mathsf{cnx{\mh}or}_{#1}^{#2,#3}(#4)}
\newcommand{\smeta}[2]{\wedge\mh\mathsf{eta}_{#1,#2}}
\newcommand{\gluelpath}[3]{\mathsf{gluel{\mh}path}(#1,#2,#3)}
\newcommand{\gluerpath}[3]{\mathsf{gluer{\mh}path}(#1,#2,#3)}
\newcommand{\smgraph}[1]{\wedge\mh\mathsf{graph}^{\bm{x}}}
\newcommand{\which}[1]{\mathsf{which}(#1)}
\newcommand{\workhorse}[1]{\mathsf{workhorse}(#1)}

In this appendix, added in July 2019, we characterize the pointed maps
$(X, Y : \UPtd) \to X \wedge_* Y \to_* X \wedge_* Y$ between binary smash products. We focus on the part of
the argument that requires internal parametricity directly. The remainder can be conducted in ordinary cubical
type theory extended with an internally expressible parametricity hypothesis; we have formalized this segment
in \citepalias[\texttt{cool/parametric-smash}]{redtt}, so we will be less formal here.

The \emph{smash product} is a higher inductive type that defines a binary operation on \emph{pointed types},
types paired with a specified basepoint. We write $\UPtd \eqdef \sigmacl{X}{\UU}{X}$ for the universe of
pointed types. To improve readability, we adopt a convention of writing $A_* \in \UPtd$ for a pointed type,
$A \eqdef \fst{A_*}$ for its underlying type, and $a_0 \eqdef \snd{A_*}$ for its basepoint. Given two pointed
types $A_*,B_* \in \UPtd$, we have their \emph{pointed function type}
$A_* \pto B_* \eqdef \sigmacl{f}{A \to B}{\Path{B}{fa_0}{b_0}}$ of basepoint-preserving functions, itself
pointed by the pointed constant function. Again, given $F_* \in A_* \pto B_*$, we write $F$ and $f_0$ for its
first and second components. We have pointed $\Gr$-types (\cref{def:gr-types}) for pointed functions: given
$A_*,B_* \in \UPtd$ and $F_* \in A_* \to_* B_*$, we write
$\Gr*{\bm{r}}{A_*}{B_*}{F_*} \eqdef \pair{\Gr{\bm{r}}{A}{B}{F}}{\gel{\bm{r}}{a_0}{b_0}{f_0}} \in \UPtd$.

Given $A_*,B_* \in \UPtd$, their smash product is defined as the following higher inductive type, here
expressed using the schema of \cite{cavallo19}.

\begin{quote}
\dataheading{A_* \wedge B_*}{\UU} \\
\niceconstr{\smpair}[\oft{a}{A},\oft{b}{B}]{A_* \wedge B_*} \\
\niceconstr{\smbasel}{A_* \wedge B_*} \\
\niceconstr{\smbaser}{A_* \wedge B_*} \\
\niceconstr{\smgluel^x}[\oft{b}{B}]{A_* \wedge B_*}{\tube{x=0}{\smbasel} \mid \tube{x=1}{\smpair{a_0}{b}}} \\
\niceconstr{\smgluer^x}[\oft{a}{A}]{A_* \wedge B_*}{\tube{x=0}{\smbaser} \mid \tube{x=1}{\smpair{a}{b_0}}}
\end{quote}

The smash product can itself be made a pointed type:
$A_* \wedge_* B_* \eqdef \pair{A_* \wedge B_*}{\smpair{a_0}{b_0}}$.  The operator $X_* \wedge_* -$ is left
adjoint to the pointed function space $X_* \pto -$, and so plays an important role in homotopy theory. Our
goal is to prove the following.

\begin{proposition*}
  Any function $f_* : \picl{X_*,Y_*}{\UPtd}{X_* \wedge_* Y_* \to_* X_* \wedge_* Y_*}$ is connected by a path
  to either the polymorphic identity or the polymorphic constant function.
\end{proposition*}

First, we introduce a few lemmas of general use in cubical type theory.

\begin{definition}[Concatenation by inverse]
  let $\coftype{M}{A}$, $\dimj{r}$, and $\coftype[\cx{\Psi,x}]{N}{A}$ with
  $\ceqtm[\cx<r=1>]{M}{\dsubst{N}{1}{x}}{A}$ be given. For any $\dimj{s}$, define
  $\coftype{\concinv{A}{r}{s}{M}{x.N}}{A}$ as follows.
  \[
    \concinv{A}{r}{s}{M}{x.N} \eqdef \hcom{A}{1}{s}{M}{\tube{r=0}{\_.M},\tube{r=1}{x.N}}
  \]
  The term $\concinv{A}{r}{0}{M}{x.N}$ is the result of concatenating $M$ (as a path in direction $r$) with
  the inverse of $x.N$; we will need the general form $\concinv{A}{r}{s}{M}{x.N}$ to relate the composite to
  other terms.
\end{definition}

\begin{definition}[$\vee$-connection]
  Let a type $\cwftypek{A}$ and $\coftype[\cx{\Psi,x}]{P}{A}$ be given. For $\dimj{r,s}$, define
  $\coftype{\cnxor{A}{r}{s}{x.P}}{A}$ following \citepalias[\texttt{prelude/connection}]{redtt}.
  \[
    \cnxor{A}{r}{s}{x.P} \eqdef
    \bighcom{A}{1}{0}{\dsubst{P}{1}{x}}{
      \begin{array}{lcl}
        \arraytube{r=0}{y.\hcom{A}{1}{s}{\dsubst{P}{1}{x}}{\tube{y=0}{x.P},\tube{y=1}{\_.\dsubst{P}{1}{x}}}} \\
        \arraytube{r=1}{\_.\dsubst{P}{1}{x}} \\
        \arraytube{s=0}{y.\hcom{A}{1}{r}{\dsubst{P}{1}{x}}{\tube{y=0}{x.P},\tube{y=1}{\_.\dsubst{P}{1}{x}}}} \\
        \arraytube{s=1}{\_.\dsubst{P}{1}{x}}
      \end{array}
    }
  \]
  Note that this term satisfies the following equations, and so plays the role played by the term
  $\dsubst{P}{r \lor s}{x}$ in cubical type theories with connections \citep{cchm,orton16}.
  \begin{align*}
    \cnxor{A}{r}{0}{x.P} &= \dsubst{P}{r}{x} \in A &
    \cnxor{A}{r}{1}{x.P} &= \dsubst{P}{1}{x} \in A \\
    \cnxor{A}{0}{s}{x.P} &= \dsubst{P}{s}{x} \in A &
    \cnxor{A}{1}{s}{x.P} &= \dsubst{P}{1}{x} \in A
  \end{align*}
\end{definition}

\begin{lemma}[$\eta$ for $\smelim$]
  Fix pointed types $A_*,B_* : \UPtd$. We have a term $\smeta{A_*}{B_*}$ of the following type.
  \[
     \picl{c}{A_* \wedge B_*}{\Path{A_* \wedge B_*}{\smelim{A_* \wedge B_*}{c}{a.b.\smpair{a}{b}}{\smbasel}{\smbaser}{x.b.\smgluel{x}{b}}{x.a.\smgluer{x}{a}}}{c}}
  \]
\end{lemma}
\begin{proof}
  Set $F \eqdef \lam{c}{\smelim{A_* \wedge B_*}{c}{a.b.\smpair{a}{b}}{\smbasel}{\smbaser}{x.b.\smgluel{x}{b}}{x.a.\smgluer{x}{a}}}$. Define $\smeta{A_*}{B_*}$ to be the following term.
  \[
     \lam{c}{\smelim{c.\Path{A_* \wedge B_*}{Fc}{c}}{c}{a.b.\dlam{\_}{\smpair{a}{b}}}{\dlam{\_}{\smbasel}}{\dlam{\_}{\smbaser}}{x.b.\dlam{\_}{\smgluel{x}{b}}}{x.a.\dlam{\_}{\smgluer{x}{a}}}} \qedhere
  \]
\end{proof}

The smash product also has a functorial action on pointed functions. For the next several lemmas, we fix
pointed types $A_*,B_*,C_*,D_* : \UPtd$ and functions $f_* : A_* \to_* C_*$, $g_* : B_* \to_* D_*$. Our aim is
to prove a \emph{graph lemma} relating the smash product of the $\Gr$-types for
functions $f_*$ and $g_*$ to the $\Gr$-type for $f_* \wedge g_*$.

\begin{definition}
  Define $f_* \wedge g_* : A_* \wedge B_* \to C_* \wedge D_*$ as follows.
  \[
   f_* \wedge g_* \eqdef \lam{k}{\bigsmelim{C_* \wedge D_*}{k}{a.b.\smpair{fa}{gb}}{\smbasel}{\smbaser}{x.b.\concinv{C_* \wedge D_*}{x}{0}{\smgluel{x}{gb}}{y.\smpair{\dapp{f_0}{y}}{gb}}}{x.a.\concinv{C_* \wedge D_*}{x}{0}{\smgluer{x}{fa}}{y.\smpair{fa}{\dapp{g_0}{y}}}}}
 \]
\end{definition}

We will need the following two lemmas, which analyze the behavior of the functorial action on the path
constructors of the smash product.

\begin{lemma}
  Let $b : B$, $d : D$, and $p : \Path{D}{gb}{d}$ be given. Then we have a term of the following type.
  \[
    \coftype{\gluelpath{b}{d}{p}}{\Path{y.\Path{C_* \wedge D_*}{\smbasel}{\smpair{\dapp{f_0}{y}}{\dapp{p}{y}}}}{\dlam{z}{(f_* \wedge g_*)(\smgluel{z}{b})}}{\dlam{z}{\smgluel{z}{d}}}}
  \]
\end{lemma}
\begin{proof}
  Define $\gluelpath{b}{d}{p}$ to be the following term.
  \[
    \dlam{y}{\dlam{z}{\bighcom{C_* \wedge D_*}{1}{0}{\smgluel{z}{\dapp{p}{y}}}{
      \begin{array}{lcl}
        \arraytube{y=0}{w.\concinv{C_* \wedge D_*}{z}{w}{\smgluel{z}{gb}}{y.\smpair{\dapp{f_0}{y}}{gb}}} \\
        \arraytube{y=1}{\_.\smgluel{z}{d}} \\
        \arraytube{z=0}{\_.\smbasel} \\
        \arraytube{z=1}{w.\smpair{\cnxor{A}{y}{w}{v.\dapp{f_0}{v}}}{\dapp{p}{y}}}
      \end{array}
    }}} \qedhere
  \]
\end{proof}

\begin{lemma}
  Let $a : A$, $c : C$, and $p : \Path{C}{fa}{c}$ be given. Then we have a term of the following type.
  \[
    \coftype{\gluerpath{a}{c}{p}}{\Path{y.\Path{C_* \wedge D_*}{\smbaser}{\smpair{\dapp{p}{y}}{\dapp{g_0}{y}}}}{\dlam{z}{(f_* \wedge g_*)(\smgluer{z}{a})}}{\dlam{z}{\smgluer{z}{c}}}}
  \]
\end{lemma}
\begin{proof}
  Define $\gluerpath{a}{c}{p}$ to be the following term.
  \[
    \dlam{y}{\dlam{z}{\bighcom{C_* \wedge D_*}{1}{0}{\smgluel{z}{\dapp{p}{y}}}{
      \begin{array}{lcl}
        \arraytube{y=0}{w.\concinv{C_* \wedge D_*}{z}{w}{\smgluer{z}{fa}}{y.\smpair{fa}{\dapp{g_0}{y}}}} \\
        \arraytube{y=1}{\_.\smgluer{z}{c}} \\
        \arraytube{z=0}{\_.\smbaser} \\
        \arraytube{z=1}{w.\smpair{\dapp{p}{y}}{\cnxor{A}{y}{w}{v.\dapp{g_0}{v}}}}
      \end{array}
    }}} \qedhere
  \]
\end{proof}

\begin{theorem}[Graph Lemma for $\wedge$]
  \label{thm:smash-graph-lemma}
  For any fresh $\bm{x}$, there is a map
  \[
    \smgraph{\bm{x}} \in \Gr*{\bm{x}}{A_*}{C_*}{f_*} \wedge \Gr*{\bm{x}}{B_*}{D_*}{g_*} \to \Gr{\bm{x}}{A_* \wedge B_*}{C_* \wedge D_*}{f_* \wedge g_*}
  \]
  equal to the identity function on $A_* \wedge_* B_*$ when $\bm{x} = \bm{0}$ and on $C_* \wedge_* D_*$ when
  $\bm{x} = \bm{1}$.
\end{theorem}
\begin{proof}
  Let us abbreviate $G \eqdef \Gr{\bm{x}}{A_* \wedge B_*}{C_* \wedge D_*}{f_* \wedge g_*}$. We define the map
  by smash product induction. We start with the $\smpair$ case, defining
  $Q_{\smpair} \in \Gr{\bm{x}}{A}{C}{f} \to \Gr{\bm{x}}{B}{D}{g} \to G$ as follows.
  \begin{align*}
    T_{a,c,p} &\eqdef \lam{n}{\bigextent{\bm{x}}{n}{b.\smpair{a}{b}}{d.\smpair{c}{d}}{b.d.q.\blam{\bm{x}}{\gel{\bm{x}}{\smpair{a}{b}}{\smpair{c}{d}}{\dlam{y}{\smpair{\dapp{\ungel{p}}{y}}{\dapp{\ungel{q}}{y}}}}}}} \\[0.3em]
    Q_{\smpair} &\eqdef \lam{m}{\bigextent{\bm{x}}{m}{a.\lam{b}{\smpair{a}{b}}}{c.\lam{d}{\smpair{c}{d}}}{a.c.p.\blam{\bm{x}}{T_{a,c,p}}}}
  \end{align*}
  Second, define $Q_\smbasel,Q_\smbaser \in G$ as follows.
  \begin{align*}
    Q_\smbasel &\eqdef \gel{\bm{x}}{\smbasel}{\smbasel}{\dlam{\_}{\smbasel}} &
    Q_\smbaser &\eqdef \gel{\bm{x}}{\smbaser}{\smbaser}{\dlam{\_}{\smbaser}}
  \end{align*}
  Third, define $Q_{\smgluel} \in \picl{n}{\Gr{\bm{x}}{B}{D}{g}}{\Path{G}{Q_\smbasel}{Q_{\smpair}(\gel{\bm{x}}{a_0}{c_0}{f_0})(n)}}$ as follows.
  \begin{align*}
    Q_{\smgluel} &\eqdef \lam{n}{\bigextent{\bm{x}}{n}{b.\dlam{z}{\smgluel{z}{b}}}{d.\dlam{z}{\smgluel{z}{d}}}{b.d.q.\dlam{z}{\gel{\bm{x}}{\smgluel{z}{b}}{\smgluel{z}{d}}{\dlam{y}{\dapp{\dapp{\gluelpath{b}{d}{\ungel{q}}}{y}}{z}}}}}}
  \end{align*}
  Likewise, define
  $Q_{\smgluer} \in
  \picl{m}{\Gr{\bm{x}}{A}{C}{f}}{\Path{G}{Q_\smbaser}{Q_{\smpair}(m)(\gel{\bm{x}}{b_0}{d_0}{g_0}}}$ as
  follows.
  \begin{align*}
    Q_{\smgluer} &\eqdef \lam{m}{\bigextent{\bm{x}}{m}{a.\dlam{z}{\smgluer{z}{a}}}{c.\dlam{z}{\smgluer{z}{c}}}{a.c.p.\gel{\bm{x}}{\smgluer{z}{a}}{\smgluer{z}{c}}{\dlam{y}{\dapp{\dapp{\gluerpath{a}{c}{\ungel{p}}}{y}}{z}}}}}
  \end{align*}
  Finally, we assemble the five cases to define $\smgraph{\bm{x}}$, using the $\eta$-principle for the smash
  product to ensure that the function is the identity when $\bm{x} = \bm{0}$ or $\bm{x} = \bm{1}$.
  \[
    \smgraph{\bm{x}} \eqdef
    \lam{g}{\bighcom{G}{0}{1}{\bigsmelim{G}{g}{m.n.Q_{\smpair}mn}{Q_\smbasel}{Q_\smbaser}{n.Q_{\smgluel}n}{m.Q_{\smgluer}m}}{
        \begin{array}{lcl}
          \arraytube{\bm{x}=\bm{0}}{y.\dapp{\smeta{A_*}{B_*}(g)}{y}} \\
          \arraytube{\bm{x}=\bm{1}}{y.\dapp{\smeta{C_*}{D_*}(g)}{y}}
        \end{array}
        }} \qedhere
  \]
\end{proof}

\begin{lemma}
  For any $a : \bool_* \wedge \bool_*$, there is a term $\which{a}$ of the following type.
  \[
    \sigmacl{k}{\bool}{\Path{\bool_* \wedge \bool_*}{a}{\ifb{\bool_* \wedge \bool_*}{k}{\smpair{\true}{\true}}{\smpair{\false}{\false}}}}
  \]
\end{lemma}
\begin{proof}
  This is a consequence of the fact that $\bool_*$ is a unit for $\wedge$; see
  \citepalias[\texttt{pointed.smash}]{redtt}.
\end{proof}

\begin{lemma}[Workhorse lemma]
  \label{lem:smash-workhorse}
  Write $P \eqdef \picl{X_*,Y_*}{\UPtd}{X \to Y \to X_* \wedge Y_*}$. For any $f : P$, there is a term
  $\workhorse{f}$ of the following type.
  \[
    \sigmacl{k}{\bool}{\Path{P}{f}{\lam{X_*}{\lam{Y_*}{\ifb{X \to Y \to X_* \wedge Y_*}{k}{\lam{\_}{\lam{\_}{\smpair{x_0}{y_0}}}}{\lam{a}{\lam{b}{\smpair{a}{b}}}}}}}}
  \]
\end{lemma}
\begin{proof}
  For the first component, we take $\fst{\which{f(\bool_*)(\bool_*)(\false)(\false)}}$. For the second, we go
  by function extensionality. Let $X_*, Y_* : \UPtd$, $a : X$, and $b : Y$ be given. We have a pointed map
  $g^X_* : \bool_* \to X_*$ taking $\true$ to $x_0$ and $\false$ to $a$, likewise $g^Y_* : \bool_* \to Y_*$
  taking $\true$ to $y_0$ and $\false$ to $b$.

  Fix a fresh bridge dimension $\bm{x}$ and define the following pointed $\Gel$ types.
  \begin{align*}
    G^X_* &\eqdef \Gr*{\bm{x}}{\bool_*}{X_*}{g^X_*} &
    G^Y_* &\eqdef \Gr*{\bm{x}}{\bool_*}{Y_*}{g^Y_*}
  \end{align*}
  We apply $f$ at these types, followed by the elements of $G^X$ and $G^Y$ corresponding to $a$ and $b$.
  \[
    f(G^X_*)(G^Y_*)(\gel{\bm{x}}{\false}{a}{\dlam{\_}{a}})(\gel{\bm{x}}{\false}{b}{\dlam{\_}{b}}) \in G^X_* \wedge G^Y_*
  \]
  At $\bm{x} = \bm{0}$, this is $f(\bool_*)(\bool_*)(\false)(\false) : \bool_* \wedge \bool_*$; at
  $\bm{x} = \bm{1}$, it is $fX_*Y_*ab : X_* \wedge Y_*$. By \cref{thm:smash-graph-lemma}, we obtain a term in
  $\Gr{\bm{x}}{\bool_* \wedge \bool_*}{X_* \wedge Y_*}{g^X_* \wedge g^Y_*}$ with the same endpoints. Applying
  $\ungel$, we get a proof that these endpoints are in the graph of $g^X_* \wedge g^Y_*$, i.e., that
  $(g^X_* \wedge g^Y_*)(f(\bool_*)(\bool_*)(\false)(\false))$ is path-equal to $fX_*Y_*ab$.  If $k$ is
  $\true$, this means that $fX_*Y_*ab$ is $\smpair{x_0}{y_0}$; if $k$ is false, that $fX_*Y_*ab$ is
  $\smpair{a}{b}$.
\end{proof}

This concludes the part of the argument that uses internal parametricity directly, i.e., mentions bridge
variables. The remainder of the proof can be conducted in ordinary cubical type theory by assuming
\cref{lem:smash-workhorse} as an axiom; we have done so in \citepalias[\texttt{cool/parametric-smash}]{redtt}.

\begin{corollary}
  \label{cor:smash-workhorse-set}
  The type $P$ defined in \cref{lem:smash-workhorse} is a set.
\end{corollary}
\begin{proof}
  The lemma shows that $P$ is a retract of $\bool$; any retract of a set is a set \citepalias[Theorem
  7.1.4]{hott-book}.
\end{proof}

Finally, we prove the main theorem. The central idea is that the behavior of a given
$f_* : \picl{X_*,Y_*}{\UPtd}{X_* \wedge_* Y_* \to_* X_* \wedge_* Y_*}$ on each constructor, as well as its
basepoint-preservation path, can be cast as an element of or path in the type $P$ we have already
characterized.

\begin{theorem}
  For any $f_* : \picl{X_*,Y_*}{\UPtd}{X_* \wedge_* Y_* \to_* X_* \wedge_* Y_*}$, there is a term of the
  following type.
  \[
    \sigmacl{k}{\bool}{\Path{P}{f}{\lam{X_*}{\lam{Y_*}{\ifb{X_* \wedge Y_* \to_* X_* \wedge Y_*}{k}{\pair{\lam{\_}{\smpair{x_0}{y_0}}}{\dlam{\_}{\smpair{x_0}{y_0}}}}{\pair{\lam{s}{s}}{\dlam{\_}{\smpair{x_0}{y_0}}}}}}}}
  \]
\end{theorem}
\begin{proof}
  For the first component, we take $\fst{\which{f(\bool_*)(\bool_*)(\false)(\false)}}$. For the second, we go
  by function extensionality. Let $X_*,Y_* : \UPtd$ be given. Given $X_*,Y_*$, we write $fX_*Y_*$ for the
  function underlying $f_*X_*Y_*$.

  Write $P \eqdef \picl{X_*,Y_*}{\UPtd}{X \to Y \to X_* \wedge Y_*}$ as in \cref{lem:smash-workhorse}. First,
  we isolate the behavior of $f_*$ on the $\smpair$ constructor:
  $\lam{X_*}{\lam{Y_*}{\lam{a}{\lam{b}{fX_*Y_*(\smpair{a}{b})}}}} : P$. By \cref{lem:smash-workhorse}, this
  is one of two functions. We aim to show that this is the only degree of freedom available to $f_*$.

  The values of $f$ on the $\smbasel$ and $\smbaser$ constructors are uniquely determined up to a path by the
  fact that $f_*X_*Y_*$ is basepoint-preserving, as $\smbasel$ and $\smbaser$ are connected to the basepoint
  of $X_* \wedge_* Y_*$ by $\smgluel{-}{y_0}$ and $\smgluer{-}{x_0}$ respectively.

  For $\smgluel$, we consider the term
  $H \eqdef \dlam{x}{\lam{X_*}{\lam{Y_*}{\lam{a}{\lam{b}{fX_*Y_*(\smgluel{x}{b})}}}}}$, which is a path in $P$
  from $\lam{X_*}{\lam{Y_*}{\lam{a}{\lam{b}{fX_*Y_*(\smbasel)}}}}$ to
  $\lam{X_*}{\lam{Y_*}{\lam{a}{\lam{b}{fX_*Y_*(\smpair{x_0}{b})}}}}$. By \cref{cor:smash-workhorse-set}, we
  know that the paths types of $P$ are all propositions, so $H$ is uniquely determined up to a path. So, then,
  is the behavior of $f_*$ on $\smgluel$ terms. The same applies to $\smgluer$.

  Finally, write
  $f_0 : \picl{X_*,Y_*}{\UPtd}{\Path{X_* \wedge Y_*}{fX_*Y_*(\smpair{x_0}{y_0})}{\smpair{x_0}{y_0}}}$ for the
  proof that $f$ preserves the basepoint of $X_* \wedge_* Y_*$. As with $\smgluel$, we prove that $f_0$ is
  uniquely determined by reformulating it as a path in $P$, namely the path
  $\lam{x}{\lam{X_*}{\lam{Y_*}{\lam{a}{\lam{b}{\dapp{f_0X_*Y_*}{x}}}}}}$ connecting
  $\lam{X_*}{\lam{Y_*}{\lam{a}{\lam{b}{fX_*Y_*(\smpair{x_0}{y_0})}}}}$ to
  $\lam{X_*}{\lam{Y_*}{\lam{a}{\lam{b}{\smpair{x_0}{y_0}}}}}$.
\end{proof}


\bibliographystyle{plainnat}
\bibliography{main}

\end{document}
